%% file: MedialAxis-arXiv.tex
\pdfoutput=1

\documentclass{article}

\newcommand{\shortcite}[1]{\cite{#1}}

\usepackage{hyperref}

\hypersetup{pdfauthor={Wouter van Toll, Atlas F. Cook IV, Marc J. van Kreveld, and Roland Geraerts},pdftitle={The Medial Axis of a Multi-Layered Environment and its Application as a Navigation Mesh}}

\newcommand{\noop}[1]{}

\newenvironment{NewContent}{\color{black}}{\ignorespacesafterend}
\newcommand \NewContentInline[1] {{\color{black} #1}}

\input{packages.tex}

\begin{document}

\title{\textbf{\NewContentInline{The Medial Axis of a Multi-Layered Environment and its Application as a Navigation Mesh}} 
\footnote{This research has been supported by the COMMIT project (\url{http://www.commit-nl.nl/}).
Authors' e-mail addresses: \texttt{W.G.vanToll@uu.nl}; \texttt{acook4@hawaii.edu}; \texttt{M.J.vanKreveld@uu.nl}; \texttt{R.J.Geraerts@uu.nl}.}} 

\author{Wouter van Toll \\ Utrecht University
\and Atlas F. Cook IV \\ University of Hawaii at Manoa
\and Marc J. van Kreveld \\ Utrecht University
\and Roland Geraerts \\ Utrecht University}

\maketitle

\begin{abstract} 
Path planning for walking characters in complicated virtual environments is a fundamental task in simulations and games. 
\begin{NewContent}
A \emph{navigation mesh} is a data structure that allows efficient path planning. 
The Explicit Corridor Map (ECM) is a navigation mesh based on the \emph{medial axis}. 
It enables path planning for disk-shaped characters of any radius.

In this paper, we formally extend the medial axis (and therefore the ECM) to 3D environments in which characters are constrained to walkable surfaces. 
Typical examples of such environments are multi-storey buildings, train stations, and sports stadiums. 
We give improved definitions of a \emph{walkable environment} (WE: a description of walkable surfaces in 3D) 
and a \emph{multi-layered environment} (MLE: a subdivision of a WE into connected layers). 
We define the medial axis of such environments based on projected distances on the ground plane. 
For an MLE with $n$ boundary vertices and $k$ connections, we show that the medial axis has size \BigO{n}, 
and we present an improved algorithm that constructs the medial axis in \BigO{n \log n \log k} time. 

The medial axis can be annotated with nearest-obstacle information to obtain the ECM navigation mesh. 
\end{NewContent}
Our implementations show that the ECM can be computed efficiently for large 2D and multi-layered environments, 
and that it can be used to compute paths within milliseconds.
This enables simulations of large virtual crowds of heterogeneous characters in real-time.
\end{abstract}

\input{section-introduction.tex}
\input{section-relatedwork.tex}
\input{section-2dma.tex}
\input{section-mle.tex}
\input{section-mlma.tex}

\input{section-applications.tex}

\input{section-implementation.tex}
\input{section-experiments.tex}
\input{section-conclusions.tex}


\section*{Acknowledgments}
We thank Elena Khramtcova for fruitful discussions on deletions in Voronoi diagrams,
Arne Hillebrand for helping us construct the 3D environments, 
and Incontrol Simulation Solutions for supplying the \emph{Stadium} environment.

\bibliographystyle{plain}
\bibliography{../bibliography-acm}

\end{document}

%% file: packages.tex
\usepackage{moreverb,url,hyperref}

\usepackage[top=3cm,bottom=3cm,left=3cm,right=3cm]{geometry}

\usepackage{subfigure}
\newcommand \NiceSubref[1] {\subref{#1}}

\usepackage{amsthm}
\usepackage{amssymb,amsmath,amsfonts}

\usepackage{comment}

\usepackage{ctable}
\usepackage{multirow}

\newcommand \etal {et al.\ }

\newcommand \linesegment[2] {\ensuremath{\overline{#1 #2}}}
\newcommand \interior[1] {\textit{Int}(\ensuremath{#1})}

\newcommand \BigO[1] {\ensuremath{\mathcal{O}(#1)}}

\newcommand \BigTheta[1] {\ensuremath{\Theta(#1)}}

\newcommand \Env {\ensuremath{\mathcal{E}}}
\newcommand \Efree {\ensuremath{\Env_{\textit{free}}}}
\newcommand \Eobs {\ensuremath{\Env_{\textit{obs}}}}
\newcommand \EfreeI[1] {\ensuremath{\Env_{\textit{free},#1}}}

\newcommand \MedialAxis[1] {\ensuremath{\textit{MA}(#1)}}

\newcommand \ConnectionSetOld {\ensuremath{\mathcal{C}'}}
\newcommand \ConnectionSetNew {\ensuremath{\mathcal{C}''}}
\newcommand \MedialAxisOld {\MedialAxis{\Env, \ConnectionSetOld}}
\newcommand \MedialAxisNew {\MedialAxis{\Env, \ConnectionSetNew}}

\newcommand \SumOfIterations {\ensuremath{\sum_{i=0}^{k-1} m_i \log m_i}}

\newcommand \Retraction[1] {\ensuremath{\textit{Retr}(#1)}}

\usepackage[capitalize,noabbrev]{cleveref}

\newtheorem{theorem}{Theorem}[section]
\newtheorem{definition}{Definition}
\newtheorem{lemma}{Lemma}[section]
\newtheorem{corollary}{Corollary}
\newtheorem{property}{Property}[section]

%% file: section-introduction.tex
\begin{figure}
	\centering
	\includegraphics[width=0.7\textwidth]{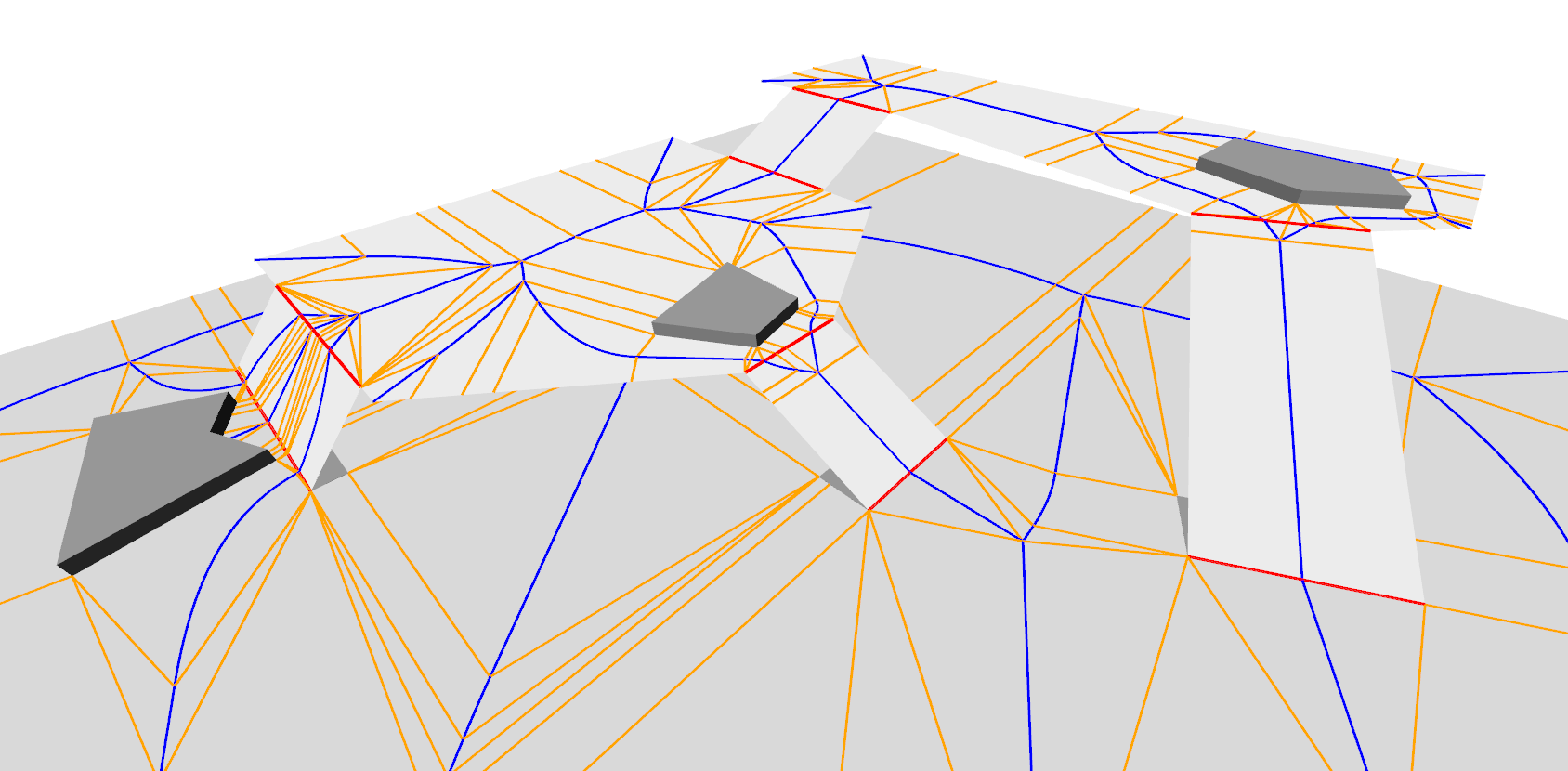}
	\caption{3D view of a multi-layered environment and its ECM navigation mesh. 
	The ECM is a medial axis (the dark graph) annotated with nearest-obstacle information (the light line segments).
	\label{fig:ramps-3d}}
\end{figure}

\section{Introduction}

In many simulations and gaming applications, virtual characters need to plan and traverse visually convincing paths through a complicated environment in real-time.
We focus on entities that move along walkable surfaces; we will refer to these entities as `characters'.
Characters should move smoothly and avoid collisions with obstacles and other characters.
The environment in which they move is typically three-dimensional, but characters are constrained to walkable surfaces. 
These surfaces form the \emph{walkable environment} (WE). 
It is often useful to subdivide the WE into planar layers.
We refer to such a subdivision as a \emph{multi-layered environment} (MLE).

A \emph{navigation mesh} is \NewContentInline{a subdivision of the WE into polygons for the purpose of path planning}. 
The \emph{dual graph} of this mesh has a vertex for each polygon in the mesh and an edge for each pair of adjacent polygons.
Characters can search in this graph to find \emph{global} routes, which they traverse while \emph{locally} avoiding other characters.

\begin{NewContent}
\bigskip

\noindent
In this paper, we work towards a refined definition and construction algorithm of the Explicit Corridor Map (ECM) \cite{Geraerts2010-ECM,vanToll2011-MultiLayered}. 
The ECM is a navigation mesh based on the medial axis. 
The medial axis of an environment is the set of all points in the environment with more than one closest obstacle; 
we will define this more precisely in \cref{s:2dma} (for 2D environments) and \cref{s:mlma:definition} (for walkable environments).

First, we revisit the medial axis in 2D. 
Next, we give improved definitions of walkable environments (WEs) and multi-layered environments (MLEs), 
and we define the medial axis of a WE and MLE based on projected distances onto the ground plane.
For an MLE with $n$ boundary vertices and $k$ connections between layers, 
where each connection is a line segment when projected onto the ground plane, 
we show how to compute the medial axis in \BigO{n \log n \log k} time. 
Our algorithm uses several non-trivial insights into the characteristics of a WE. 

The medial axis of an MLE can easily be annotated to obtain the ECM navigation mesh.
\end{NewContent}
\cref{fig:ramps-3d} shows the ECM for a simple multi-layered environment.
The ECM can be used to plan paths for disk-shaped characters of any radius, which is typically not possible when using an arbitrary subdivision into polygons.
Because the ECM is a sparse graph that uses only \BigO{n} storage, \NewContentInline{it is suitable for efficient path planning.} 
It also supports various operations that are important for crowd simulation, such as finding the nearest static obstacle to a query point.
Furthermore, it can be updated in real-time when obstacles appear or disappear.
Combined with algorithms for path following and collision avoidance, the ECM can be used to simulate large crowds of heterogeneous characters in real-time.

\subsection{Contributions}

This paper formalizes and extends previous conference publications \cite{Geraerts2010-ECM,vanToll2011-MultiLayered} 
by defining and proving the essential characteristics of the medial axis and ECM in multi-layered environments.
Compared to our previous work, the main \emph{contributions} of this paper are the following: 

\begin{NewContent}
	\begin{itemize}
	\item We give improved definitions of walkable environments (WEs), multi-layered environments (MLEs), 
	and the medial axis of a WE or MLE based on projected distances (\cref{s:mle,s:mlma:definition}). 
	\item We solve a problem in our previous construction algorithm for the medial axis of an MLE \cite{vanToll2011-MultiLayered}. 
	We present an improved algorithm, and we prove that this algorithm is correct and that it runs in \BigO{n \log n \log k} time (\cref{s:mlma}). 
	\item We give a refined definition of the ECM navigation mesh (\cref{s:ecm}).
	\item We show via experiments that our implementation can efficiently generate the ECM and compute paths in large multi-layered environments (\cref{s:experiments}).
\end{itemize}
\end{NewContent}

%% file: section-relatedwork.tex
\section{Related Work on Path Planning and Navigation Meshes} \label{s:relatedwork}

This section summarizes related research on path planning, navigation meshes, and crowd simulation.
Several good overview books of this research area exist \cite{Kallmann2016-PathPlanningBook,Thalmann2013-CrowdSimulation}. 
We will focus on the topics that are particularly relevant to \NewContentInline{the ECM navigation mesh. 
Related work on its underlying structure, the medial axis, will be covered in \cref{s:2dma}.}


\subsection{Path Planning} \label{s:relatedwork:pathplanning}

In traditional motion planning, a robot needs to compute a collision-free trajectory from one configuration to another \cite{Latombe1991-MotionPlanning,LaValle2006-PlanningAlgorithms}. 
The number of degrees of freedom for the robot determines the dimensionality of the \emph{configuration space}.
For high-dimensional spaces, exact solutions are typically intractible and \emph{probabilistic} methods are often successful \cite{Kavraki1996-PRM,Kuffner2000-RRT}.
Such methods do not compute the configuration space explicitly: instead, they represent it using a graph of sampled collision-free configurations and motions. 
\NewContentInline{Many techniques for robot motion planning use a sampled version of the Voronoi diagram or medial axis 
\cite{Holleman2000-MedialAxisMP,Lien2003-MedialAxisMP,Garber2004-VoronoiDiagramMP}.}

In crowd simulation, characters are typically modelled as disks that move along walkable surfaces, 
which enables different types of algorithms than in traditional motion planning. 
Path planning for disks can be solved by inflating the obstacles in the environment by the character's radius (i.e.\ by computing Minkowski sums)
and then computing a path for a point character \cite{Latombe1991-MotionPlanning,LaValle2006-PlanningAlgorithms}.
This approach requires a separate inflation process for each distinct radius. 

To plan a shortest path for a point in a two-dimensional \NewContentInline{space}, one can use a shortest-path map \cite{Hershberger1999-ShortestPathMap} or a visibility graph, 
which has \BigO{n^2} edges for an environment with $n$ obstacle vertices \cite{Ghosh2007-Visibility}.
The Visibility-Voronoi Complex \cite{Wein2007-VisibilityVoronoi} is an extension that implicitly encodes the visibility graph for all disk sizes; 
it can generate the visibility graph for a particular radius on the fly.

Path planning for large crowds typically uses a less complex graph (or \emph{roadmap}) that does not always yield the shortest path. 
\NewContentInline{The medial axis is such a roadmap; its application to path planning for disk-shaped characters is referred to as the `retraction method' \cite{ODunlaing1985-Retraction}, 
and we use it as a basis for our ECM navigation mesh.
Roadmaps have been used frequently for crowd simulation in 2D and 3D environments \cite{Hoecker2010-FastestPath,Rodriguez2011-Roadmaps}.}
However, \NewContentInline{purely graph-based techniques} are not ideal for crowd simulation: 
characters would need to follow the edges of the graph exactly (which leads to unnatural motion and collisions between characters),
or they would have to perform expensive geometric tests to check how they can deviate from an edge.

Another popular approach to path planning is to use a \emph{grid} that subdivides the environment into regular cells. 
Grids are easy to implement and well-studied by the path planning community 
(see e.g.\ \cite{Garcia2014-GridPlanning,Koenig2002-DStarLite,Lee2013-GridPlanning,Sturtevant2012-Benchmarks}). 
However, grids have resolution problems: 
a coarse grid (with few cells) does not capture the environment's details, 
whereas a fine grid (with many cells) quickly becomes too costly to store and query.


\subsection{Navigation Meshes in 2D Environments} \label{s:relatedwork:navmeshes-2d}

A \emph{navigation mesh} subdivides the \NewContentInline{configuration space} into polygonal regions \cite{Snook2000-NavMeshes}.
A global path in a navigation mesh can be found by performing A* search \cite{Hart1968-AStar} on the dual graph of the mesh.
The result is a sequence of regions to move through, such that characters can use the available space to locally adjust their movements during the simulation. 
\NewContentInline{This makes navigation meshes more flexible for crowd simulation than standard graphs.
Navigation meshes are also more scalable to large environments than grids \cite{vanToll2016-ComparativeStudy}.}

Many navigation meshes exist for 2D environments; they are typically \emph{exact} subdivisions of the \NewContentInline{configuration space}. 
Examples are the Local Clearance Triangulation (LCT) \cite{Kallmann2014-LCT} and the Explicit Corridor Map (ECM) \cite{Geraerts2010-ECM}.
These navigation meshes have the advantage that they encode \emph{clearance} information. 
As such, they can be used to compute paths for disk-shaped characters with an arbitrary radius without explicitly inflating any obstacles.

\subsection{Navigation Meshes in 3D Environments} \label{s:relatedwork:navmeshes-3d}

Navigation meshes in 3D are usually designed for environments with a consistent direction of gravity.
\NewContentInline{We will use the term \emph{walkable environment} (WE) to denote the surfaces on which characters can walk.}
Many navigation mesh techniques \NewContentInline{automatically convert a 3D environment to a WE} by discretizing the environment into traversable and non-traversable cubes, or \emph{voxels}.
\NewContentInline{The pioneering work by Pettr\'e \etal \shortcite{Pettre2005-NavGraph} 
essentially computes an \emph{approximation} of the multi-layered medial axis that we define in this paper. 
Their navigation mesh} supports arbitrary character sizes but uses overlapping disks that do not completely cover the \NewContentInline{environment}.
The Recast Navigation toolkit \cite{Mononen-Recast} 
is very popular in the computer games industry, e.g.\ it is included in the Unity game engine \shortcite{Unity}. 
Oliva and Pelechano have presented a similar method called NEOGEN \shortcite{Oliva2013-Neogen}, 
and they have investigated path planning for disks in arbitrary navigation meshes \shortcite{Oliva2013-Clearance}.
A disadvantage of voxels is that they do not scale well to large environments because these environments require many voxels to obtain sufficient precision. 
Therefore, exact alternatives to 3D filtering are also being investigated \cite{Polak2016-Filtering}. 

For many applications, it is useful to subdivide the walkable environment into layers such that each layer can be treated as a planar problem space.
We will refer to such a subdivision as a \emph{multi-layered environment} (MLE). 
In this paper, we extend the \NewContentInline{medial axis and the} Explicit Corridor Map to MLEs.
There are several ways to obtain an MLE from a WE.
Deusdado \etal have used rendering techniques assuming certain properties such as axis-alignment \shortcite{Deusdado2008-MultiLayered}, 
and Whiting \etal have shown how to extract layers from a CAD drawing \shortcite{Whiting2007-MultiLayered}.
For an arbitrary WE represented by a triangle mesh, 
Hillebrand \shortcite{Hillebrand2016-MLE} has proven that obtaining an optimal MLE (with a minimum number of connections) is NP-hard in the number of triangles, 
but he has shown that good results can be obtained using heuristics \shortcite{Hillebrand2016-PEEL}.

\begin{NewContent}
\bigskip

\noindent
A consequence of splitting a walkable environment into layers and treating each layer as a 2D space is that \emph{height differences} along the surface are ignored. 
In all navigation meshes based on this principle (including ours), slopes are not considered to affect the length of a path, 
and path lengths are effectively projected onto the ground plane. 
This can be seen as a disadvantage. 

However, finding short paths on terrains with height differences is known to be substantially more difficult \cite{Kaul2013-ShortestPathsOnTerrains}. 
Recently, researchers have suggested to use the 3D environment itself as a navigation data structure \cite{Ricks2014-WholeSurface,Berseth2015-NavMesh3D}, 
which coincidentally also supports environments with arbitrary gravity directions. 
At the time of writing, these types of solutions are less mature; 
for instance, it is unclear how to compute paths with arbitrary clearance, and how to properly extend collision-avoidance algorithms to this domain. 

By contrast, \emph{projected distances} provide a simplification that allow solutions in 2D to be extended to 3D environments with a consistent gravity direction.
In this paper, we will provide this extension for the medial axis and the ECM. 
Extending these concepts further to also encode height differences is a topic for future work.

\subsection{Comparing Navigation Meshes} \label{s:relatedwork:comparativestudy}

We have recently conducted a comparative study of navigation meshes \cite{vanToll2016-ComparativeStudy}, 
using unified definitions, objective quality metrics, and a benchmark suite that runs on a single computer.
This comparison included the multi-layered ECM from this paper, 
several other state-of-the-art navigation meshes \cite{Kallmann2014-LCT,Pettre2005-NavGraph,Oliva2013-Neogen,Mononen-Recast}, 
and a simple grid-based method. 
Our study confirmed that voxel-based extraction of a WE is not very scalable to large environments, which justifies a search for alternatives.

For an analysis of the theoretical and practical (dis)advantages of the ECM compared to other navigation meshes, 
we refer the reader to this comparative study \cite{vanToll2016-ComparativeStudy}. 
We will not repeat the comparison here; instead, we will focus on definitions, algorithms, proofs, and experiments specifically relevant to the medial axis and ECM. 
\end{NewContent}


\subsection{Crowd Simulation} \label{s:relatedwork:crowdsimulation}

Navigation meshes are useful for simulating crowds of virtual characters with individual properties and goals.
Path planning on a navigation mesh leads to an \emph{indicative route} that can be traversed smoothly in real-time \cite{Jaklin2013-MIRAN}.
Each character can locally avoid collisions with other characters 
using forces \cite{Helbing1995-SocialForces,Reynolds1999-Steering} or velocity selection \cite{Karamouzas2010-CollisionAvoidance,vandenBerg2011-ORCA,Moussaid2011-CollisionAvoidance}.
\NewContentInline{Many other algorithms exist that can improve global coordination and local behaviour.}
These components can be mixed arbitrarily in a multi-level crowd simulation framework \cite{vanToll2015-Framework}.

Other crowd simulation algorithms aim to unify the global and local planning levels by defining a \emph{potential field}: 
a grid representation of the environment that stores the optimal walking direction \NewContentInline{(towards a particular goal region)} in each cell. 
These directions are updated in real-time in response to the crowd's movement.
\NewContentInline{While potential fields in general can contain local minima in which characters would get stuck, 
various solutions specific to crowd simulation have alleviated this problem via global optimization \cite{Narain2009-DenseCrowds,Patil2010-NavigationFields,Treuille2006-ContinuumCrowds}.} 
Potential fields can efficiently model dense crowds with many characters sharing the same goals and properties.
However, because each goal region and behavior type requires its own potential field, 
these methods do not allow for large \emph{heterogeneous} crowds in which each character has different properties.
If individuality is required, navigation meshes with local collision avoidance are preferred.

%% file: section-2dma.tex
\section{Preliminaries: The Medial Axis in 2D Environments} \label{s:2dma}

In this section, we give a formal definition of a two-dimensional virtual environment and its medial axis. 
\NewContentInline{These concepts are not novel, but they are required for understanding the extension to WEs and MLEs that we will present next.}

\subsection{2D Environment} \label{s:2dma:definition-env}

Let \Env\ be a bounded two-dimensional planar environment with polygonal obstacles.
We define the \emph{obstacle space} \Eobs\ as the union of all obstacles, including the boundary of the environment.
The complement of \Eobs\ is the \emph{free space} \Efree. 
An example of such an environment is shown in \cref{fig:u-obstacles}.
Let $n$ be the number of vertices \NewContentInline{that define the boundary of \Efree}\ using interior-disjoint simple polygons, line segments, and points.
We call $n$ the \emph{complexity} of \Env.

\subsection{Medial Axis} \label{s:2dma:definition-ma}

The medial axis is a variant of the \emph{Voronoi diagram} (VD) \cite{Okabe2000-Voronoi,Aurenhammer2013-Voronoi}.
For a planar set of point sites, the VD is a subdivision of the plane into cells such that all points in a cell have the same nearest site.
The \emph{edges} of the VD are parts of bisectors: line segments or half-lines on which every point is equidistant to two sites. 
These bisectors meet at \emph{vertices} that are equidistant to at least three sites.

The VD can be extended to handle line segments and polygons as sites. 
This version is sometimes called the \emph{generalized Voronoi diagram} or GVD \cite{Lee1981-GVD,Geraerts2010-ECM}.
The edges of a GVD consist of line segments and parabolic arcs, and degree-2 vertices occur \NewContentInline{at the positions where} a bisector changes its shape. 
Several robust implementations of the GVD exist \cite{Hoff1999-GVD,Karavelas2004-GVD,Held2011-Vroni,Boost}.
There are multiple definitions of the GVD, differing mostly in how they handle site vertices shared by multiple sites.
The term `generalized Voronoi diagram' is also used for other generalizations of the VD \cite{Okabe2000-Voronoi}.
\bigskip

\noindent
\NewContentInline{The \emph{medial axis} has been extensively studied in the field of computational geometry, usually for 2D polygons with or without holes 
\cite{Blum1967-MedialAxis,Preparata1977-MedialAxis,Lee1982-MedialAxis,Chin1999-MedialAxis,Wilmarth1999-Retraction}.
As mentioned in \cref{s:relatedwork:pathplanning}, the medial axis has also been used frequently in the context of motion planning.}

Informally, the medial axis of a polygon $P$ can be seen as the GVD of $P$'s boundary segments, restricted to (i.e.\ intersected with) the interior of $P$.
Several definitions of the medial axis exist; they mainly differ in their choice of which edges to prune from the GVD.
We therefore give our own definition, which is comparable to those by Preparata and Lee \cite{Preparata1977-MedialAxis,Lee1982-MedialAxis}, 
\NewContentInline{but applied specifically to a 2D environment and its free space}.

\begin{definition}[Medial axis, 2D] \label{def:2dma}
	For a bounded 2D environment \Env\ with closed polygonal obstacles, let $\textit{ma}(\Env)$ be the set of all points in \Efree\ 
	that have at least two distinct equidistant nearest points \NewContentInline{on the boundary of \Efree}, in terms of 2D Euclidean distance. 
	The medial axis \MedialAxis{\Env}\ is the topological closure of $\textit{ma}(\Env)$.
\end{definition}

\begin{figure*}[t]
	\centering
	\subfigure[2D environment \label{fig:u-obstacles}]{
		\centering
		\includegraphics[height=0.3\textwidth]{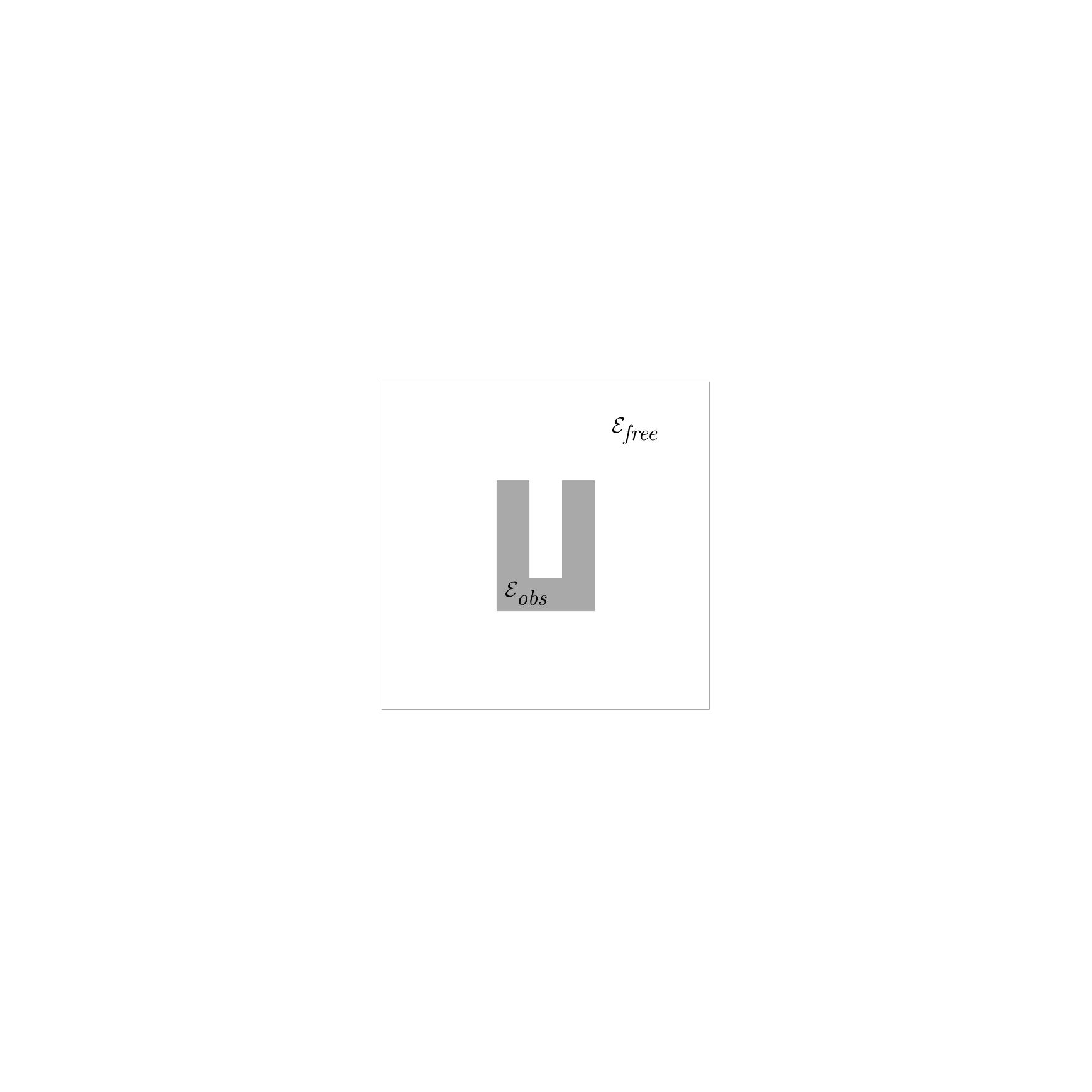}
	}
	\subfigure[Medial axis \label{fig:u-ma}]{
		\includegraphics[height=0.3\textwidth]{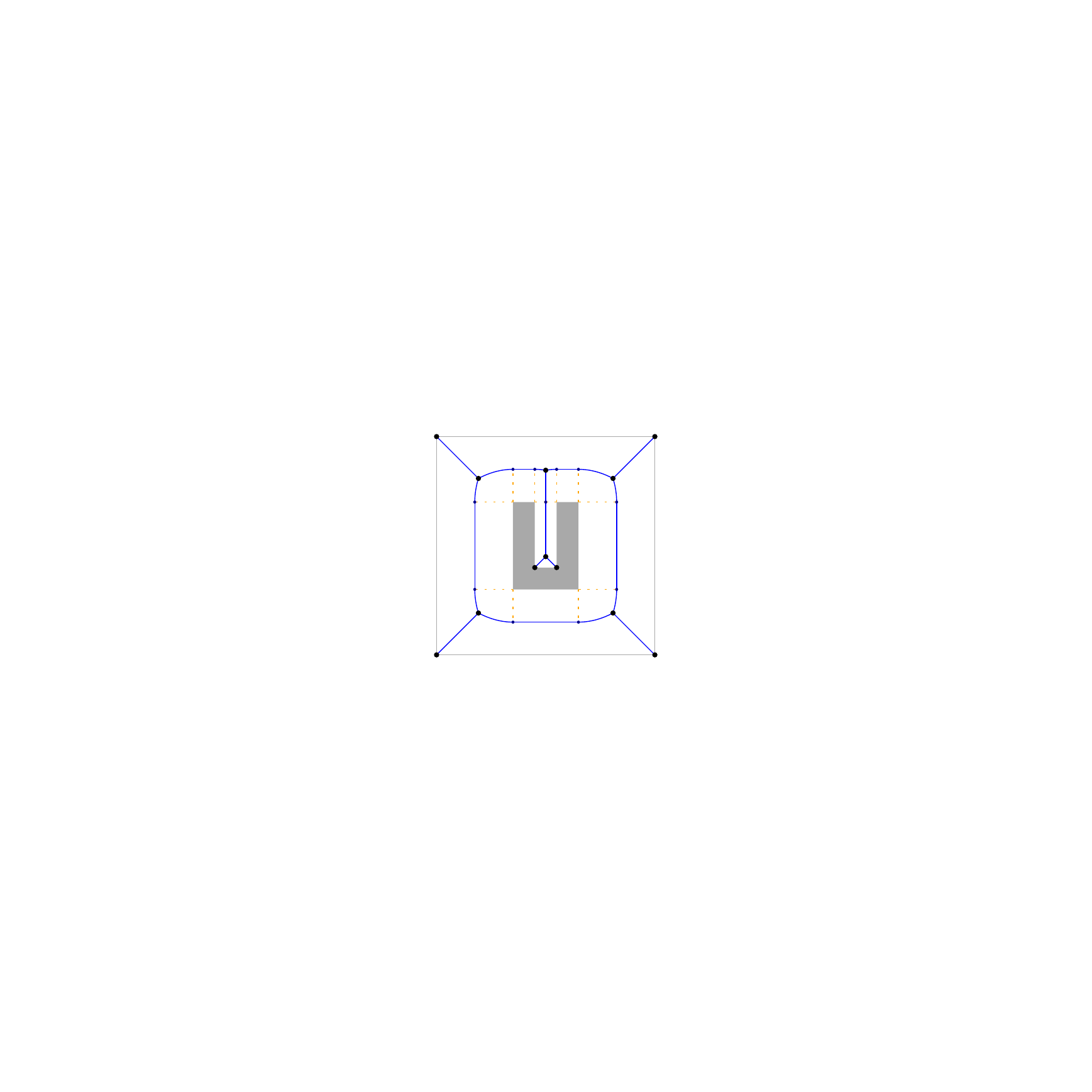}
	}
	\caption{A simple 2D environment and its medial axis. 
	\NiceSubref{fig:u-obstacles} The obstacle space \Eobs\ (shown in gray) consists of line segments and polygons. Its complement is the free space \Efree.
	\NiceSubref{fig:u-ma} The medial axis is a graph through \Efree. 
	True vertices are shown as large dots.
	Semi-vertices (small dots) occur when a bisector's generator changes.
	This is indicated by dashed segments, which are not part of the graph.
	\label{fig:u-obstacles-ma}}
\end{figure*}

\cref{fig:u-ma} shows the medial axis of an example environment. 
Since the medial axis is a pruned Voronoi diagram, it forms a plane graph (a planar graph embedded in 2D). 
The term `closure' ensures that \NewContentInline{degree-1 medial axis vertices (e.g.\ at the corners of the bounding box) are also included}.
Note that the medial axis does not run into \NewContentInline{obstacle corners with an angle of $\leq 180$ degrees (\emph{convex} corners), 
such as most corners of the `U' shape in \cref{fig:u-ma}.
Such edges \emph{do} appear in a GVD of line segment sites because these corners are then shared by multiple sites.}

Each medial axis arc $A$ is the bisector of two \emph{generators}: the endpoints or segments of \Eobs\ that are nearest to $A$.
If one generator is a line segment and the other is a point, then $A$ is a parabolic arc; otherwise, $A$ is a line segment.

In this paper, we refer to all \NewContentInline{medial axis} vertices of degree 1, 3, or higher as \emph{true vertices}.
We refer to the degree-2 vertices as \emph{semi-vertices} because the medial axis only changes its shape at these points.
Observe from \cref{fig:u-ma} that a semi-vertex occurs when the medial axis crosses a normal vector at a convex obstacle corner.
We define an \emph{edge} as a sequence of medial axis arcs between two true vertices.

\subsection{Complexity of the Medial Axis} \label{s:2dma:complexity}

\NewContentInline{It is well-known that the generalized Voronoi diagram (GVD) of non-crossing line segments with $n$ distinct endpoints has \BigO{n} vertices and edges 
and can be computed in \BigO{n \log n} time \cite{Aurenhammer2013-Voronoi,deBerg2008-CompGeom}.
The three most common construction algorithms for the point-site Voronoi diagram 
(plane sweep \cite{Fortune1987-Voronoi}, randomized incremental construction \cite{Green1978-IncrementalVoronoi}, and divide-and-conquer \cite{Shamos1975-Voronoi}) 
can each be extended to support line-segment sites as well \cite{Aurenhammer2013-Voronoi}.}

The medial axis has the same asymptotic size of \BigO{n} because it is a pruned GVD. 
The medial axis can be obtained from the GVD without increasing the overall asymptotic running time \cite{Kirkpatrick1979-MedialAxis,Chin1999-MedialAxis}.

%% file: section-mle.tex
\section{Definitions of Walkable and Multi-Layered Environments} \label{s:mle}

In this section, we define the types of environments \emph{embedded in 3D} for which we want to construct a navigation mesh. 
Our main assumption is that there is a consistent direction of gravity throughout an environment. 
For example, we support multi-storey buildings, but not spherical planets. 
As explained in \cref{s:relatedwork}, this is a common assumption for many navigation meshes, 
and we consider other 3D surfaces to be outside the scope of this paper.

Throughout this paper, a \emph{3D environment} is a collection of polygons in $\mathbb{R}^3$.
These polygons may include floors, ceilings, walls, or any other type of geometry.
\cref{fig:3d-environment-original} shows a simple example of a 3D environment. 

\subsection{Walkable Environment} \label{s:mle:walkable}

\NewContentInline{Informally, a \emph{walkable environment} (WE) can be thought of} as a set of polygonal surfaces in $\mathbb{R}^3$ on which characters can walk.
Characters can move from one polygon onto another if these polygons are connected in 3D.
The \emph{free space} \Efree\ of a WE is simply the entire set of surfaces. 
Unlike in 2D environments, the obstacle space \Eobs\ is not intuitively defined, 
but we will sometimes refer to points on the boundary of \Efree\ as `obstacle points'.
\NewContentInline{\cref{fig:3d-environment-walkable} shows a simple example of a walkable environment.}

A walkable environment may consist of multiple connected components: for example, consider two islands with no bridge connecting them.
In topological terms, each component is an orientable 2-manifold (a \emph{surface}) with a boundary.
This intuitively means that the WE has a `top' and `bottom' side, and any point on the bottom side cannot be reached from a point on the top side without intersecting a boundary.
Geometrically, we are only interested in the top side, i.e.\ the floors and not the ceilings. 
The WE is also what we call \emph{direction-consistent}: slopes are allowed, but there is a single gravity direction for the entire environment. 
This leads to the following formal definition:

\begin{definition}[Walkable environment]
	A walkable environment (WE) is a set of interior-disjoint polygons in $\mathbb{R}^3$.
	Topologically, each connected component of a WE is an orientable 2-manifold with a boundary.
	Geometrically, the WE is direction-consistent: there exists a horizontal ground plane $P$ below the WE 
	such that for any non-boundary point $q$, the infinitesimal neighborhood $\sigma(q)$ of $q$ does not overlap itself when projected vertically down onto $P$.
	\NewContentInline{Semantically, the WE represents all polygons on which characters can walk.}
\end{definition}

\NewContentInline{In theory, it does not matter how a walkable environment is created.
In practice, a WE} can be obtained from a 3D environment by filtering out unwalkable parts, 
e.g.\ surfaces that are too steep and surfaces along which the ceiling is too low for characters to pass under. 
\NewContentInline{(Note that filtering out steep surfaces leads to direction-consistent output.)}
Such a filtering process typically also merges polygons that are nearly adjacent; for example, staircases are converted into ramps. 
As explained in \cref{s:relatedwork:navmeshes-3d}, the details of these filtering techniques are outside the scope of this paper. 
Voxel-based implementations exist \cite{Pettre2005-NavGraph,Oliva2013-Neogen,Mononen-Recast}, and exact alternatives are in development \cite{Polak2016-Filtering}. 

Once more, it is important to note that the entire WE \emph{can be self-overlapping} when projected onto the ground plane $P$.
This is the main difference to 2D environments, \NewContentInline{and it is the main reason that we introduce \emph{multi-layered environments} next.}

\begin{figure*}[t]
	\centering
	\subfigure[3D environment \label{fig:3d-environment-original}]{
		\includegraphics[width=0.3\textwidth]{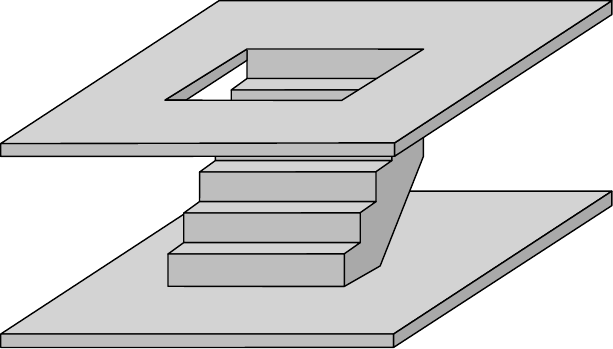}
	}
	\subfigure[Walkable environment \label{fig:3d-environment-walkable}]{
		\includegraphics[width=0.3\textwidth]{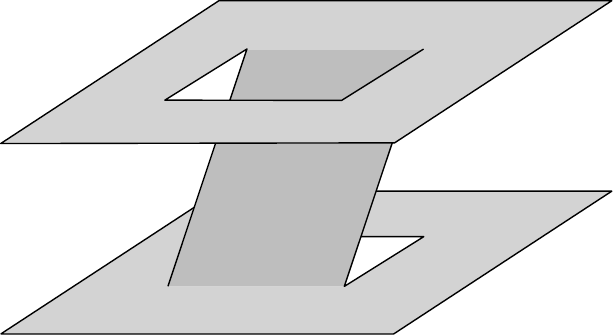}
	}
	\subfigure[Multi-layered environment \label{fig:3d-environment-multilayered}]{
		\includegraphics[width=0.3\textwidth]{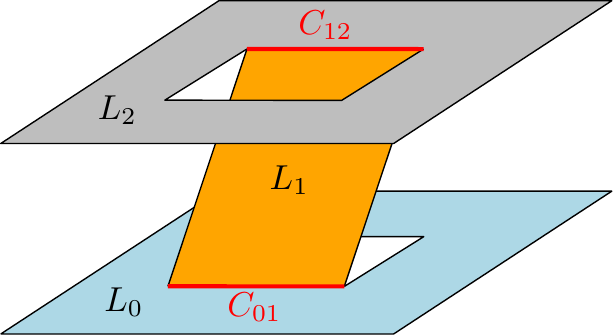}
	}
	\caption{A simple 3D environment for which we want to compute a navigation mesh.
	\NiceSubref{fig:3d-environment-original} The original environment is a collection of polygons in 3D.
	\NiceSubref{fig:3d-environment-walkable} A walkable environment (WE) is a set of polygons along which characters can walk.
	\NiceSubref{fig:3d-environment-multilayered} A multi-layered environment (MLE) is a subdivision of the WE into 2D layers. 
	Connections between layers are shown as bold line segments in this example.
	\label{fig:3d-environment}}
\end{figure*}

\subsection{Multi-Layered Environment} \label{s:mle:multilayered}

A \emph{multi-layered environment} (MLE) is a subdivision of a WE into \emph{layers} 
\NewContentInline{such that each individual layer can be projected onto the ground plane $P$ without overlap}.
Although a single layer does not need to have a particular meaning, a typical example of a layer is one floor of a building.

A subdivision into layers is useful for many purposes, including visualization (each layer can be drawn in 2D), 
identification (all surface points can be uniquely specified using a 2D position and a layer ID), 
and \NewContentInline{the construction of geometric data structures (existing 2D construction algorithms can be applied to each layer). 
In \cref{s:mlma}, we will use the concept of layers to compute the medial axis of a WE.}

The layers of an MLE are connected by \NewContentInline{curves} which we call \emph{connections}.
Intuitively, they are the `cuts' that were introduced during the subdivision into layers, 
and they are the edges along which the layers can be `glued together' to obtain the original WE. 
\NewContentInline{To facilitate the algorithm of \cref{s:mlma}, we require that connections have particular geometric properties.} 
Formally, we define an MLE as follows:

\begin{definition}[Multi-layered environment] \label{def:mle}
	A multi-layered environment (MLE) is a walkable environment (WE) that has been subdivided into $l$ planar layers, $\mathcal{L}=\{L_0,\ldots,L_{l-1}\}$, 
	using a set $\mathcal{C}=\{C_0,\ldots,C_{k-1}\}$ of $k$ connections. 
	
	Each layer $L_i \in \mathcal{L}$ is a set of walkable surfaces that are non-overlapping when projected onto the ground plane $P$.
	The free space \EfreeI{i} of $L_i$ is the union of all polygons in $L_i$. 
	Combining the free space of all layers yields the free space \Efree\ of the original WE.
	
	\NewContentInline{Each connection $C_q \in \mathcal{C}$ is a curve with the following properties:}
	\begin{itemize}
		\item It lies on the shared boundary of two layers $L_i$ and $L_j$ ($i \neq j$), thus connecting the walkable polygons of these layers. 
		\item Its endpoints lie on existing boundary vertices of \Efree, so its endpoints are impassable obstacles.
		\begin{NewContent}
		\item Its interior lies entirely inside the interior of \Efree, so it follows the surface of the WE without intersecting the boundary of \Efree. 
		\item Its interior does not intersect any other connections.
		\item Its projection onto the ground plane $P$ is a straight line segment. 
		\end{NewContent}
	\end{itemize}
\end{definition}

\cref{fig:3d-environment-multilayered} shows an example of a multi-layered environment. 
Note that the MLE is still \emph{embedded in 3D}, but each individual layer \emph{can} be projected onto $P$ without self-overlap, if desired.
Therefore, the projection of a layer $L_i$ onto $P$ is essentially a 2D environment with obstacles as described in \cref{s:2dma:definition-env}. 
The boundary vertices of these obstacles are also boundary vertices of \Efree. 
\begin{NewContent}
(Conversely, if we embed a 2D environment in $\mathbb{R}^3$, we obtain a special case of a WE, which is also an MLE with one layer and no connections.)

The connections in the example of \cref{fig:3d-environment-multilayered} are straight line segments in $\mathbb{R}^3$, 
but this is not necessary in general: a connection can be any curve that satisfies the constraints given in Theorem \ref{def:mle}.
These constraints are important for our construction algorithm in \cref{s:mlma}. 
\end{NewContent}

Two layers $L_i$ and $L_j$ may be connected through multiple connections at different positions; for example, imagine a bridge that connects to the same layer at both ends.
Also, a subdivision into layers is usually not unique: any subdivision that meets the requirements described above is acceptable.
\NewContentInline{As described in \cref{s:relatedwork:navmeshes-3d}, 
obtaining an MLE with a minimum number of connections is NP-hard \cite{Hillebrand2016-MLE}, 
but there are several heuristic approaches to obtaining a valid MLE.}

\subsection{Complexity of a Multi-Layered Environment} \label{s:mle:complexity}

The complexity of an MLE is given by the number of connections $k$ and the number of obstacle vertices $n$ in all layers combined.
Let $n_i$ be the number of obstacle vertices in a layer $L_i$.
We define $n$ as $\sum_{i=0}^{l-1} n_i$.
Note that a vertex occurs in multiple layers if it is an endpoint of a connection.
The following lemma bounds the number of connections.

\begin{lemma} \label{lem:mle:NumberOfConnections}
	For any multi-layered environment with $l$ layers and $n$ obstacle vertices, the number of connections $k$ is \BigO{n}.
\end{lemma}
\begin{proof}
	Let $n_i$ be the number of obstacle vertices in a layer $L_i$.
	By definition, $n = \sum_{i=0}^{l-1} n_i$.
	In each layer $L_i$, the number of connections is bounded by the maximum number of non-intersecting line segments that can be drawn between its $n_i$ vertices.
	Euler's formula for planar graphs implies that this is \BigO{n_i}.
	Therefore, the total number of connections is \BigO{\sum_{i=0}^{l-1} n_i} = \BigO{n}.
\end{proof}

%% file: section-mlma.tex
\section{The Medial Axis in Multi-Layered Environments} \label{s:mlma}

In this section, we first define \NewContentInline{the medial axis} for walkable and multi-layered environments. 
Because our definitions do not require a particular subdivision into layers, they apply to both WEs and MLEs.
Next, we show how to compute \NewContentInline{the medial axis} in \BigO{n \log n \log k} time for an MLE with $n$ boundary vertices and $k$ connections. 

\subsection{Definition and Properties} \label{s:mlma:definition}

To define the medial axis for walkable and multi-layered environments, we need a notion of distance and path length. 
We will use the direction-consistency of the WE to define \emph{projected distances} in which height differences are ignored. 
\NewContentInline{Again, we acknowledge that this is not the same as the 3D distance on a surface. 
In \cref{s:relatedwork:navmeshes-3d}, we have argued why this simplification is useful.}

For two points $s$ and $g$ in a WE or MLE, let $\pi(s,g)$ be a path from $s$ to $g$ through \Efree\ along the walkable surfaces.
We define the \emph{projected length} of $\pi(s,g)$ as the curve length of $\pi(s,g)$ when projected vertically onto the ground plane $P$. 
This projected path can intersect itself: for instance, consider a path along a spiral staircase.

Let $\pi^*(s,g)$ be a path from $s$ to $g$ with \NewContentInline{the smallest projected length (among all possible paths from $s$ to $g$).}
We define the \emph{projected distance} $d_P(s,g)$ between $s$ and $g$ as the projected length of $\pi^*(s,g)$.
That is, $d_P(s,g)$ ignores any height differences along paths from $s$ to $g$.
\cref{fig:mle-distance} shows an example of projected distances.

A shortest path $\pi^*(s,g)$ is \emph{unobstructed} if it does not have any bending points around obstacles. 
The projection of an unobstructed path onto $P$ is a single line segment, so its projected length is simply the 2D Euclidean distance between $s$ and $g$ (when projected onto $P$).
The following properties hold:

\begin{property}[Straight-line property] \label{prop:mlma:StraightLine}
	The shortest path $\pi^*(q,n_q)$ from any point $q \in \Efree$ to any of its nearest boundary points $n_q$ is unobstructed.
\end{property}
\begin{property}[Empty-circle property] \label{prop:mlma:EmptyCircle}
	Let $q$ be a point in \Efree\ and let $n_q$ be a nearest boundary point to $q$, at projected distance $d = d_P(q,n_q)$. 
	For all points $q' \in \Efree$ for which $d_P(q,q') \leq d$, the shortest path $\pi^*(q,q')$ is unobstructed. 
	When projected onto $P$, these points form a disk with radius $d$.
\end{property}

\NewContentInline{If a WE has been converted to an MLE (i.e.\ if it has been subdivided into layers), 
then the nearest boundary point $n_q$ to a point $q \in \Efree$\ may lie on a different layer than $q$ itself. 
Consequently, the empty disk around $q$ may span multiple layers.
This does not matter for our definitions.}

We now define the medial axis based on the function $d_P$:

\begin{definition}[Medial axis, multi-layered] \label{def:mlma}
	For a walkable or multi-layered environment \Env\ with free space \Efree, let $\textit{ma}(\Env)$ be the set of all points in \Efree\ 
	that have at least two equidistant nearest points on the boundary of \Efree\ with respect to the projected distance function $d_P$.
	The medial axis \MedialAxis{\Env}\ is the topological closure of $\textit{ma}(\Env)$.
\end{definition}

\cref{fig:mle-ma} shows the medial axis of an example walkable environment.
Because the remainder of this paper is based on the projected distance function, we will often omit the term `projected' when discussing distances and path lengths.

If an environment \Env\ consists of a single layer $L_i$, then $\MedialAxis{\Env} = \MedialAxis{L_i}$. 
\NewContentInline{Likewise, if we treat a 2D environment as an MLE with a single layer, 
then Definition \ref{def:mlma} is actually a generalization of Definition \ref{def:2dma}, 
and we have obtained a single definition for the medial axis of 2D environments, MLEs, and WEs.

The medial axis becomes more interesting if \Env\ is a WE that overlaps itself when projected onto $P$, 
or (equivalently) if \Env\ is an MLE that contains overlapping layers.
In these cases, $\MedialAxis{\Env}$ is typically not planar,} 
but intuitively, it is \emph{locally} similar to a 2D medial axis everywhere due to the straight-line and empty-circle properties. 
We will use these properties to prove that our construction algorithm for the multi-layered medial axis is correct.

\begin{figure*}[t]
	\centering
	\subfigure[Distances \label{fig:mle-distance}]{
		\includegraphics[width=0.45\textwidth]{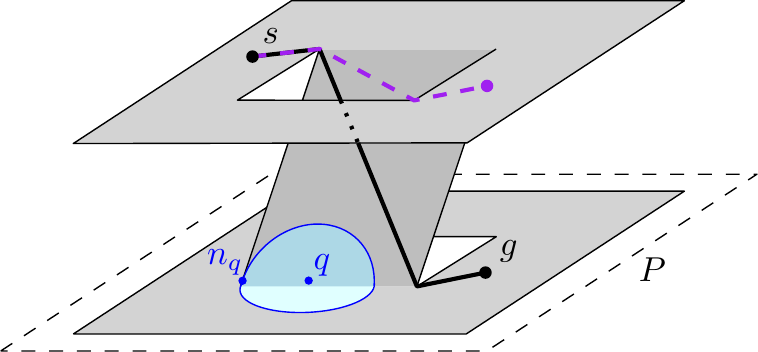}
	}
	\subfigure[Medial axis \label{fig:mle-ma}]{
		\includegraphics[width=0.45\textwidth]{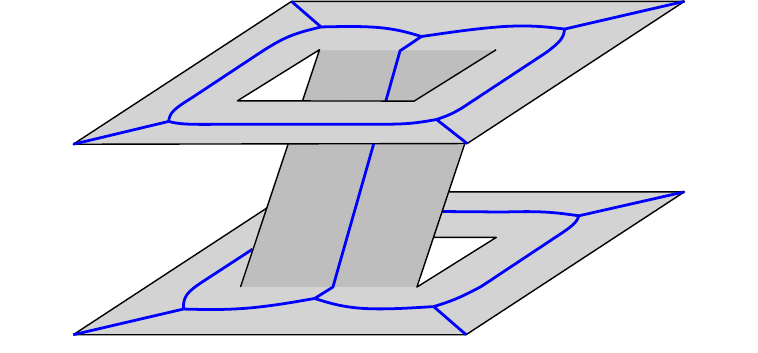}
	}
	\caption{The medial axis of a walkable environment \Env\ is based on path lengths projected onto the ground plane $P$. 
	\NiceSubref{fig:mle-distance} The shortest path between two points $s$ and $g$ is shown as a bold curve.
	Its projected length is indicated by the dashed curve.
	For a non-boundary point $q \in \Efree$\ with a nearest obstacle point $n_q$, 
	the set of points in \Efree\ within a distance of $d_P(q,n_q)$ from $q$ is a disk when projected onto $P$.
	\NiceSubref{fig:mle-ma} The medial axis \MedialAxis{\Env}\ is drawn on the surfaces of \Env.
	\label{fig:mle-distance-ma}}
\end{figure*}


\subsection{Construction Algorithm Outline} \label{s:mlma:algorithm-outline}

We now give an outline of our algorithm that computes the medial axis of a multi-layered environment \Env.
The result is also the medial axis of the corresponding \emph{walkable} environment.
However, our algorithm makes use of the fact that the \emph{two-dimensional} medial axis is easy to compute.
For this reason, we assume that the environment has been partitioned into layers. 
We acknowledge that this is a necessary pre-processing step that can be solved using other algorithms \cite{Hillebrand2016-PEEL}.
Our construction algorithm consists of the following steps:

\begin{enumerate}
	\item For each individual layer $L_i$, project $L_i$ onto $P$ and compute its 2D medial axis, while treating all of its connections as \emph{closed} impassable obstacles. 
	This yields exactly the medial axis \MedialAxis{L_i} according to the projected distance function $d_P$, 
	but under the assumption that each connection \NewContentInline{in $\mathcal{C}$} is an obstacle. 
	The result for all layers combined is the medial axis of \Env\ \NewContentInline{under the same assumption.} 
	We denote this result by \MedialAxis{\Env, \mathcal{C}}.
	\item Given \MedialAxisOld\ with $\ConnectionSetOld \subseteq \mathcal{C}$, choose a closed connection $C_q \in \ConnectionSetOld$. 
	\NewContentInline{By definition, the obstacle associated with $C_q$ is a line segment when projected onto $P$.}
	\emph{Open} the connection $C_q$ by removing its \emph{interior} as an obstacle and repairing the medial axis in its neighborhood. 
	(The endpoints of the connection will remain obstacles because they are on the boundary of \Efree.)
	\NewContentInline{The result is the medial axis of \Env\ in which $C_q$ is no longer treated as an impassable obstacle, 
	i.e.\ it is \MedialAxisNew\ with $\ConnectionSetNew = \ConnectionSetOld \setminus \{C_q\}$.
	In \MedialAxisNew, there are new edges of the medial axis that pass through $C_q$, and the neighborhood of $C_q$ is now connected.}
	In \cref{s:mlma:opening}, we will describe our algorithm for opening a connection.
	\item Repeat step 2 until all connections are open. The result is $\MedialAxis{\Env, \emptyset} = \MedialAxis{\Env}$. 
\end{enumerate}

In short, we initially treat all connections as closed and then iteratively remove them as obstacles.
\NewContentInline{Because connections project to line segments,} opening a connection is essentially the deletion of a line segment Voronoi site \cite{Aurenhammer2013-Voronoi} 
but with the extra difficulty that the neighborhood of the deleted site may span multiple layers. 
We will explain this further in \cref{s:mlma:opening}. 
For now, it is sufficient to know that existing deletion algorithms for Voronoi sites in 2D 
\cite{Aggarwal1989-LinearVoronoi,Devillers1999-VoronoiDeletion,Khramtcova2015-LinearDeletion} cannot immediately be applied. 
\cref{s:mlma:opening} will present an alternative algorithm.

\subsection{Properties of a Closed Connection} \label{s:mlma:connectionproperties}

To develop an algorithm for opening a closed connection, we must first study the properties of such a connection. 
Consider a closed connection between two layers $L_i$ and $L_j$, as in \cref{fig:mlma-merge-1}. 
We will now refer to this connection as $C_{ij}$ to emphasize to which layers it is associated.
This notation is not unique because $L_i$ and $L_j$ may be connected via other connections as well. 
However, in our discussion of opening a single connection chosen by the main algorithm, it should be clear to which instance we are referring.

\subsubsection{Sides} 
The connection is currently treated as an impassable obstacle between $L_i$ and $L_j$.
Thus, it occurs as an obstacle for the medial axis on two `sides'. 
We define the \emph{side} $S_i$ as the set of all walkable surfaces and boundary points that are currently reachable from $C_{ij}$ by starting in $L_i$.
Likewise, the side $S_j$ consists of all surfaces and obstacle points that can be reached from $C_{ij}$ by starting in $L_j$.
These sides are also annotated in \cref{fig:mlma-merge-1}.

A side $S_i$ at this point in our algorithm is not necessarily the same as a layer $L_i$ in the environment. 
The side $S_i$ includes at least the part of $L_i$ that has $C_{ij}$ on its boundary. 
If other connections are already open, then $S_i$ may contain other layers as well.
If sufficiently many connections have been opened such that $L_i$ and $L_j$ are already connected via another route, then $S_i$ and $S_j$ are even equal.
However, for our algorithm, it does not matter which layers are already included in $S_i$ or $S_j$, 
and it is useful to speak of two different sides of the connection.

\begin{figure*}[t]
	\centering
	\subfigure[3D view \label{fig:mlma-merge-1}]{
		\includegraphics[height=31mm]{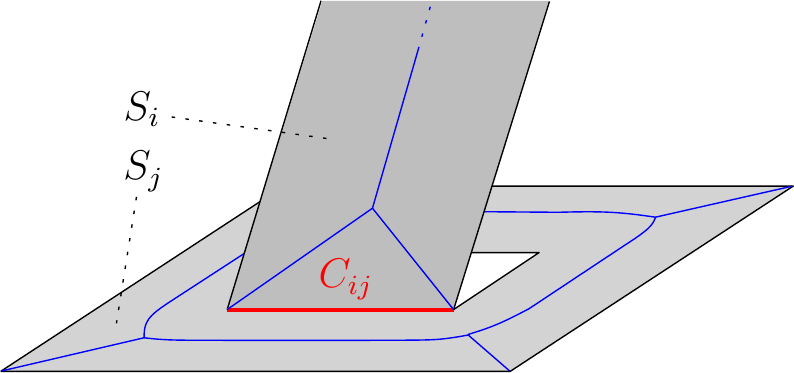}
	}
	\subfigure[2D view \label{fig:mlma-merge-2}]{
		\includegraphics[height=31mm]{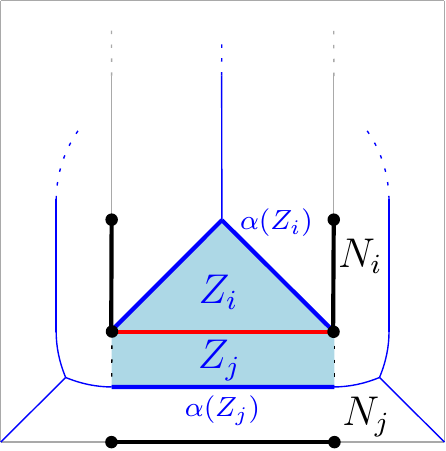}
	}
	\subfigure[Medial axis of $N_{ij}$ \label{fig:mlma-merge-3}]{
		\includegraphics[height=31mm]{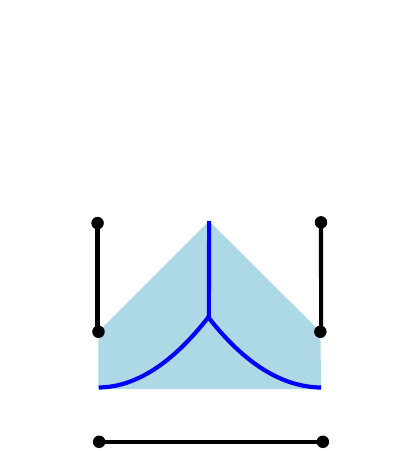}
	}
	\caption{Opening a connection $C_{ij}$ in an MLE. 
	\NiceSubref{fig:mlma-merge-1} Initially, $C_{ij}$ is an obstacle on both sides, $S_i$ and $S_j$.
	\NiceSubref{fig:mlma-merge-2} 2D top view of the area around $C_{ij}$.
	The \emph{influence zone} $Z_{ij} = Z_i \cup Z_j$ is shaded.
	The obstacle points $N_{ij} = N_i \cup N_j$ that are nearest to $Z_{ij}$ are shown in bold black.
	\NiceSubref{fig:mlma-merge-3} When opening $C_{ij}$, the medial axis changes only inside $Z_{ij}$.
	This medial axis $M_Z$ is defined by $N_{ij}$.
	\label{fig:mlma-merge}}
\end{figure*}

\subsubsection{Influence Zone} \label{s:mlma:connectionproperties:influencezone}

When we open $C_{ij}$, we effectively remove the \emph{interior} of $C_{ij}$, denoted by \interior{C_{ij}}, as an obstacle from the environment. 
We do not remove the endpoints because they will remain obstacles in the WE.

Therefore, we need to determine a new nearest obstacle for all points in the WE that were previously nearest to \interior{C_{ij}}.
Let the \emph{influence zone} $Z_{ij}$ be the closure of the set of all points in \Env\ that currently have \interior{C_{ij}} as a nearest obstacle.
Observe from \cref{fig:mlma-merge-2} that $Z_{ij}$ consists of two parts: one on side $S_i$ and the other on side $S_j$. 
(Conceptually, it does not matter if $S_i$ and $S_j$ are already equal.)
For convenience, we will refer to these parts as $Z_i$ and $Z_j$, respectively. 

\begin{lemma} \label{lem:mlma:InfluenceZoneOnly}
	If the interior of a connection $C_{ij}$ is removed as an obstacle, the medial axis changes only inside the influence zone $Z_{ij}$.
\end{lemma}
\begin{proof}
	By the definition of $Z_{ij}$, removing \interior{C_{ij}} causes the nearest obstacle points to change only inside (and on the boundary of) $Z_{ij}$.
	After all, the other points in \Efree\ were already closer to other obstacles.
	A consequence is that opening $C_{ij}$ causes the \emph{medial axis} to change only in $Z_{ij}$.
\end{proof}

Lemma \ref{lem:mlma:InfluenceZoneOnly} implies that \MedialAxisOld\ (the current medial axis with $C_{ij}$ as an obstacle) 
and \MedialAxisNew\ (the medial axis without $C_{ij}$ as an obstacle, which we want to compute) are equal except in $Z_{ij}$.
It is therefore useful to analyze the shape of $Z_{ij}$.

\begin{lemma} \label{lem:mlma:InfluenceZoneBoundary1}
	The influence zone $Z_{ij}$ is bounded by the two lines perpendicular to $C_{ij}$ through $C_{ij}$'s endpoints.
\end{lemma}
\begin{proof}
	(We will refer to these two lines as the \emph{endpoint normals} of $C_{ij}$.)
	Consider any point $p \in \Efree$\ that is \emph{not} between or on the endpoint normals of $C_{ij}$.
	Such a point cannot be closest to \interior{C_{ij}} because it must be closer to an endpoint of $C_{ij}$ or to another obstacle in the environment.
	Therefore, $p$ cannot be in $Z_{ij}$.
\end{proof}

\begin{lemma} \label{lem:mlma:InfluenceZoneBoundary2}
	$Z_i$ and $Z_j$ are both bounded by a sequence of medial axis arcs.
	Both sequences, denoted by $\alpha(Z_i)$ and $\alpha(Z_j)$, are uninterrupted and monotone with respect to the line supporting $C_{ij}$.
\end{lemma}
\begin{proof}
	We prove the lemma for $Z_i$; the proof for $Z_j$ is analogous.
	$Z_i$ is bounded by a set of medial axis arcs $\alpha(Z_i)$ that have a nearest obstacle point on $C_{ij}$.
	For every point $z$ on $\alpha(Z_i)$, the nearest point $c$ on $C_{ij}$ can be reached via a line segment \linesegment{z}{c}\ that is perpendicular to $C_{ij}$.
	If this were not true, then another obstacle would be in the way and $c$ would not be a nearest obstacle point.	
	Furthermore, $c$ is a nearest obstacle point for all points on \linesegment{z}{c} because this nearest obstacle cannot change when moving from $z$ to $c$. 
	Thus, \linesegment{z}{c}\ lies entirely inside $Z_i$.
	Because $Z_i$ consists of infinitely many line segments that all have an endpoint on $C_{ij}$ and that are all perpendicular to $C_{ij}$, 
	the boundary $\alpha(Z_i)$ is monotone with respect to $C_{ij}$.
	
	Finally, the definition of an MLE enforces that \interior{C_{ij}} does not intersect any obstacles. 
	Because of this, every point on \interior{C_{ij}} has some free space in its neighborhood: 
	for each $c \in \interior{C_{ij}}$, the line segment \linesegment{z}{c}\ exists and has non-zero length.
	The endpoints of $C_{ij}$ are the only points where $z$ and $c$ can be equal.
	This proves that $\alpha(Z_i)$ is a single sequence of arcs.
\end{proof}

These lemmas have the following consequences:
\begin{corollary} \label{cor:mlma:InfluenceZoneBoundary}
	The boundary of the influence zone $Z_{ij}$ is a single closed loop consisting of $\alpha(Z_i)$, $\alpha(Z_j)$, and the endpoint normals of $C_{ij}$.
\end{corollary}
\begin{corollary} \label{cor:mlma:InfluenceZonePlanar}
	The influence zone $Z_{ij}$ can be projected onto the ground plane $P$ without overlap. 
	This projection is a single shape without holes (because it has only one boundary).
\end{corollary}

\subsubsection{Neighbor Set} \label{s:mlma:connectionproperties:neighborset}

Next, we determine which obstacles are required to update the medial axis inside $Z_{ij}$.
On one side $S_i$ of the connection, let $N_i$ be the set of all obstacle points that are nearest to at least one point on $\alpha(Z_i)$, 
excluding \interior{C_{ij}} itself.
(On the other side $S_j$, let $N_j$ be defined analogously.)
These are the obstacle points that (together with $C_{ij}$) generate the arcs in $\alpha(Z_i)$.
We exclude \interior{C_{ij}} from $N_i$ because we will be removing this interior as an obstacle.
We do explicitly include the \emph{endpoints} of $C_{ij}$ in $N_i$ because these will remain obstacles.

We define the \emph{neighbor set} $N_{ij}$ as the union of $N_i$ and $N_j$.
Informally, this set contains the `Voronoi neighbors' of the connection.
$N_{ij}$ consists of line segments and points on the boundary of \Efree.
These are not necessarily the complete original boundary segments, but only the parts that are actually relevant for $Z_{ij}$.
Note that the neighbors can originate from many different layers and that they can even be (parts of) other connections that are still closed.
A neighbor set is illustrated in \cref{fig:mlma-merge-2}.

We will now prove that $N_{ij}$ contains the obstacle points that define the medial axis in $Z_{ij}$ when \interior{C_{ij}} is removed.
Lemma \ref{lem:mlma:ConnectionSceneNecessity} proves that \emph{all} points of $N_{ij}$ are needed; 
\cref{lem:mlma:ConnectionSceneSufficiency} proves that \emph{no other} points are needed.

\begin{lemma} 
	When $C_{ij}$ is opened, every obstacle point in $N_{ij}$ is a nearest obstacle for at least one point in $Z_{ij}$.
	\label{lem:mlma:ConnectionSceneNecessity}
\end{lemma}
\begin{proof}
	When the connection is still closed, every point $p \in N_{ij}$ is a nearest obstacle for at least one point $z$ on the \emph{boundary} of $Z_{ij}$, by definition.
	When $C_{ij}$ is opened, $p$ will still be nearest to $z$ because opening the connection only exposes $z$ to obstacles that are farther away.
	Hence, there remains at least one point in $Z_{ij}$ (namely $z$) for which $p$ is a nearest obstacle.
	This means that all points of $N_{ij}$ are required. 
\end{proof}

\begin{lemma} 
	When $C_{ij}$ is opened, $N_{ij}$ contains all possible nearest obstacle points for any point in $Z_{ij}$.
	\label{lem:mlma:ConnectionSceneSufficiency}
\end{lemma}
\begin{proof}
	We prove the lemma for $Z_i$; the proof for $Z_j$ is analogous.
	For any point $p \in Z_i$, the nearest obstacle points currently lie in $N_i$ or on $C_{ij}$, by definition.
	Removing the interior of $C_{ij}$ cannot cause other obstacle points on the same side $S_i$ to suddenly become nearest to $p$.
	The only remaining option is that an obstacle on the \emph{other} side $S_j$ becomes nearest to $p$.
	Such an obstacle must definitely lie in $N_j$: by definition, all other obstacle points of $S_j$ were already not closest to $C_{ij}$ itself, 
	so they cannot be closest to a point \emph{beyond} $C_{ij}$. 
	Therefore, all possible nearest obstacle points for $Z_i$ are included in $N_i$ and $N_j$. 
\end{proof}

\subsection{Opening a Connection} \label{s:mlma:opening}

To open a closed connection $C_{ij}$, we now know that we only need to update the medial axis inside the influence zone $Z_{ij}$.
\NewContentInline{Thus, to convert the current medial axis \MedialAxisOld\ into \MedialAxisNew\ with $C'' = C' \setminus \{C_{ij}\}$}, 
it is sufficient to \emph{only} compute $\MedialAxisNew \cap Z_{ij}$ and then replace $\MedialAxisOld \cap Z_{ij}$ by it.
We will refer to the new medial axis part, $\MedialAxisNew \cap Z_{ij}$, as $M_Z$ for convenience. 

Thus, our goal is to compute $M_Z$.
Lemmas \ref{lem:mlma:ConnectionSceneNecessity} and \ref{lem:mlma:ConnectionSceneSufficiency} guarantee that $M_Z$ is defined by the obstacles in the neighbor set $N_{ij}$.
An example of $M_Z$ is shown in \cref{fig:mlma-merge-3}.

\begin{lemma} \label{lem:mlma:OpeningTree}
	The medial axis $M_Z$ is a tree that can be projected onto $P$ without overlap.
\end{lemma}
\begin{proof}
	$M_Z$ can only contain cycles if $Z_{ij}$ contains holes; otherwise, there are no obstacles to circumnavigate.
	Furthermore, $M_Z$ can only consist of multiple connected components if $Z_{ij}$ is consists of multiple disconnected shapes.
	Corollary \ref{cor:mlma:InfluenceZonePlanar} states that $Z_{ij}$ is a single shape without holes. 
	Therefore, $M_Z$ is a single tree.
	
	Corollary \ref{cor:mlma:InfluenceZonePlanar} also states that $Z_{ij}$ is non-overlapping when projected onto $P$. 
	Because $M_Z$ lies entirely inside $Z_{ij}$, it can be projected onto $P$ without overlap as well.
\end{proof}

In previous work \cite{vanToll2011-MultiLayered}, we computed $M_Z$ by projecting all obstacles of $N_{ij}$ onto the ground plane $P$ 
and computing the 2D medial axis of this projection.
This is equivalent to a deletion of a site from a 2D Voronoi diagram \cite{Devillers1999-VoronoiDeletion}; 
the algorithm takes \BigO{m \log m} time where $m$ is the complexity of $N_{ij}$. 
Also, the algorithm is easy to implement by using any available library for Voronoi diagrams in 2D.  
We therefore still use it in our current implementation (\cref{s:implementation}).

However, due to the multi-layered structure of \Env, this algorithm does not work in all environments. 
A single common projection onto $P$ may cause an obstacle of $N_{ij}$ to influence parts of $Z_{ij}$ to which it is actually not closest.
\cref{fig:mlma-overlap} shows an example in which the obstacles $N_i$ on side $S_i$ cannot be treated as a planar set. 

\begin{figure*}[htb]
	\centering
	\subfigure[3D view \label{fig:mlma-overlap-3d}]{
		\includegraphics[width=0.315\textwidth]{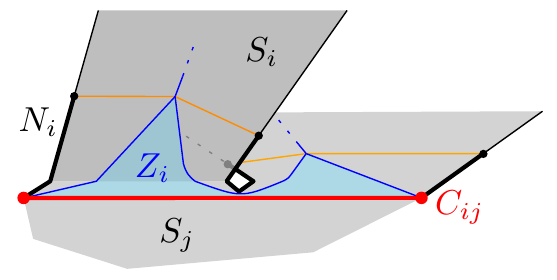}
	}%
	\subfigure[Projection onto $P$ \label{fig:mlma-overlap-2d}]{
		\includegraphics[width=0.315\textwidth]{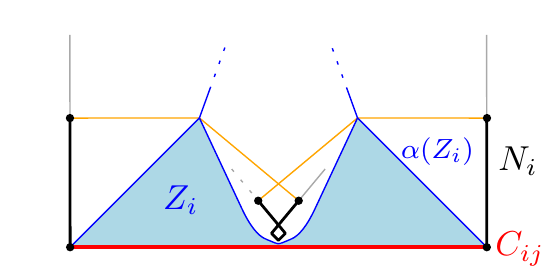}
	}%
	\subfigure[Projection of $N_i$ lifted to 3D \label{fig:mlma-overlap-3d-obstacles}]{
		\includegraphics[width=0.315\textwidth]{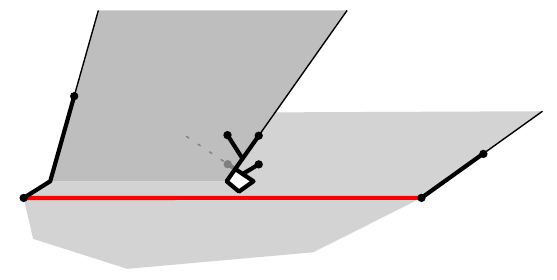}
	}
	\caption{Example in which projecting the entire neighbor set onto the ground plane $P$ leads to problems.
	\NiceSubref{fig:mlma-overlap-3d} One side $S_i$ of a connection $C_{ij}$ contains a ramp and a flat surface. 
	(In this view, the flat surface is partly occluded by the ramp. The occluded boundary part is shown in dotted gray.)
	$N_i$ (shown in bold) contains obstacle points from both parts.	
	\NiceSubref{fig:mlma-overlap-2d} A projection of the same situation onto $P$. 
	\NiceSubref{fig:mlma-overlap-3d-obstacles} If we project all of $N_i$ onto $P$ at the same time, 
	we effectively treat the points of $N_i$ as obstacles in all surfaces. 
	This will yield an incorrect medial axis for $Z_i$.
	\label{fig:mlma-overlap}}
\end{figure*}

\noindent
We now propose an improved algorithm that uses projected distances \emph{without} explicitly projecting all of $N_i$ (or $N_j$) onto $P$ at the same time. 
Our new approach starts at the boundary of $Z_{ij}$ and traces the medial axis from there, based on the obstacles that are locally nearest.
This avoids the problem of \cref{fig:mlma-overlap}.
The new approach is outlined in \cref{fig:mlma-opening}: \cref{fig:mlma-opening-before} shows the current situation with $C_{ij}$ closed, 
and the other subfigures represent the algorithm for opening $C_{ij}$.

Let $M_{Z,i}$ be $M_Z$ under the assumption that there are no obstacles in $N_j$ and that $Z_j$ extends to infinity.
Thus, $M_{Z,i}$ is defined solely by the obstacles of $N_i$.
By the same arguments as before, 
the version of $Z_{ij}$ in which $Z_j$ extends to infinity is a simple shape when projected onto $P$ (Corollary \ref{cor:mlma:InfluenceZonePlanar}), 
and $M_{Z,i}$ is a tree that does not overlap in $P$ either (Lemma \ref{lem:mlma:OpeningTree}).
An example is shown in \cref{fig:mlma-opening-i}.

Let $M_{Z,j}$ be defined analogously (\cref{fig:mlma-opening-j}).
We compute $M_{Z,i}$ and $M_{Z,j}$ separately and then merge them to obtain $M_Z$ (\cref{fig:mlma-opening-after}).

\begin{figure*}[htb]
	\centering
	\subfigure[Before opening \label{fig:mlma-opening-before}]{
		\includegraphics[width=0.235\textwidth]{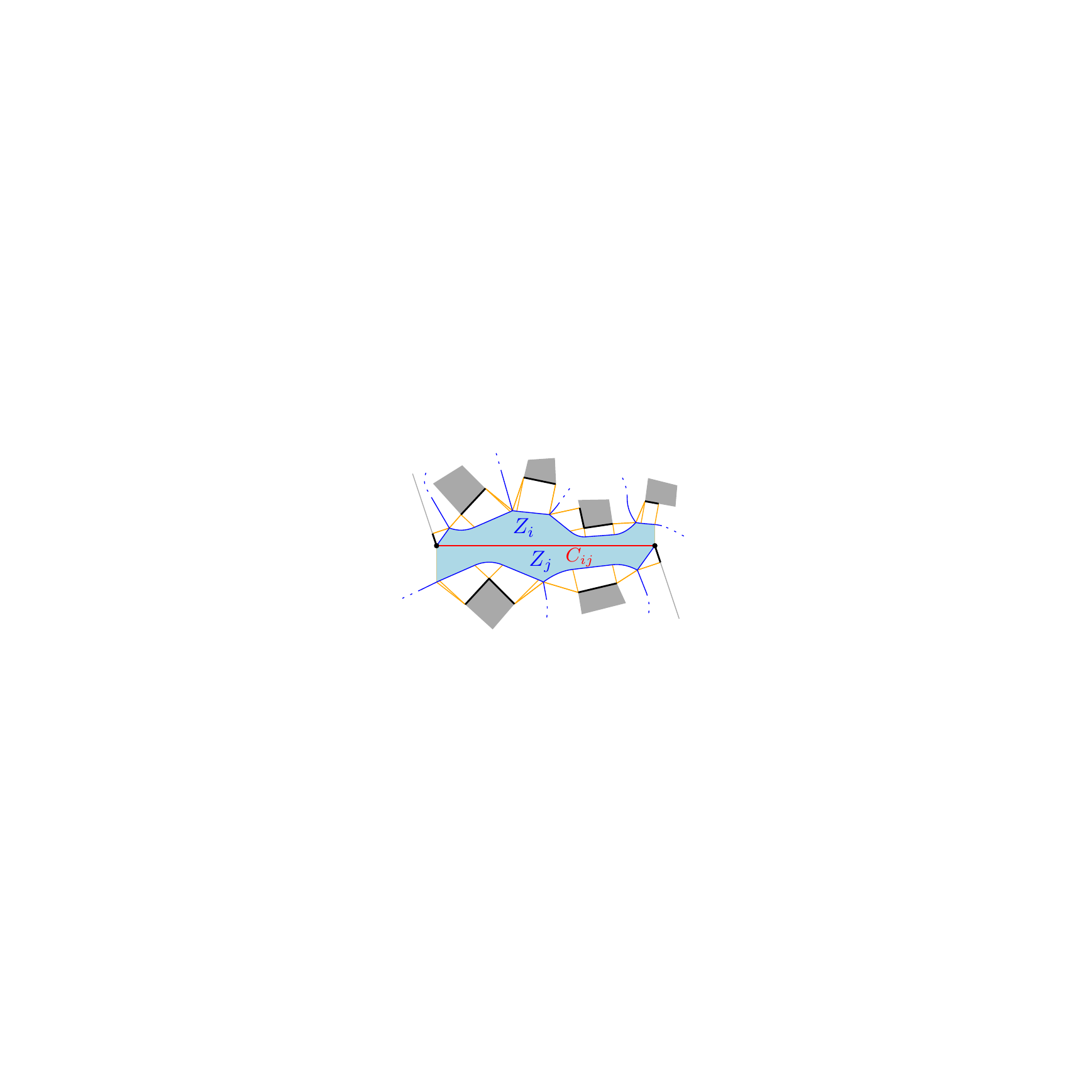}
	}%
	\subfigure[$M_{Z,i}$ \label{fig:mlma-opening-i}]{
		\includegraphics[width=0.235\textwidth]{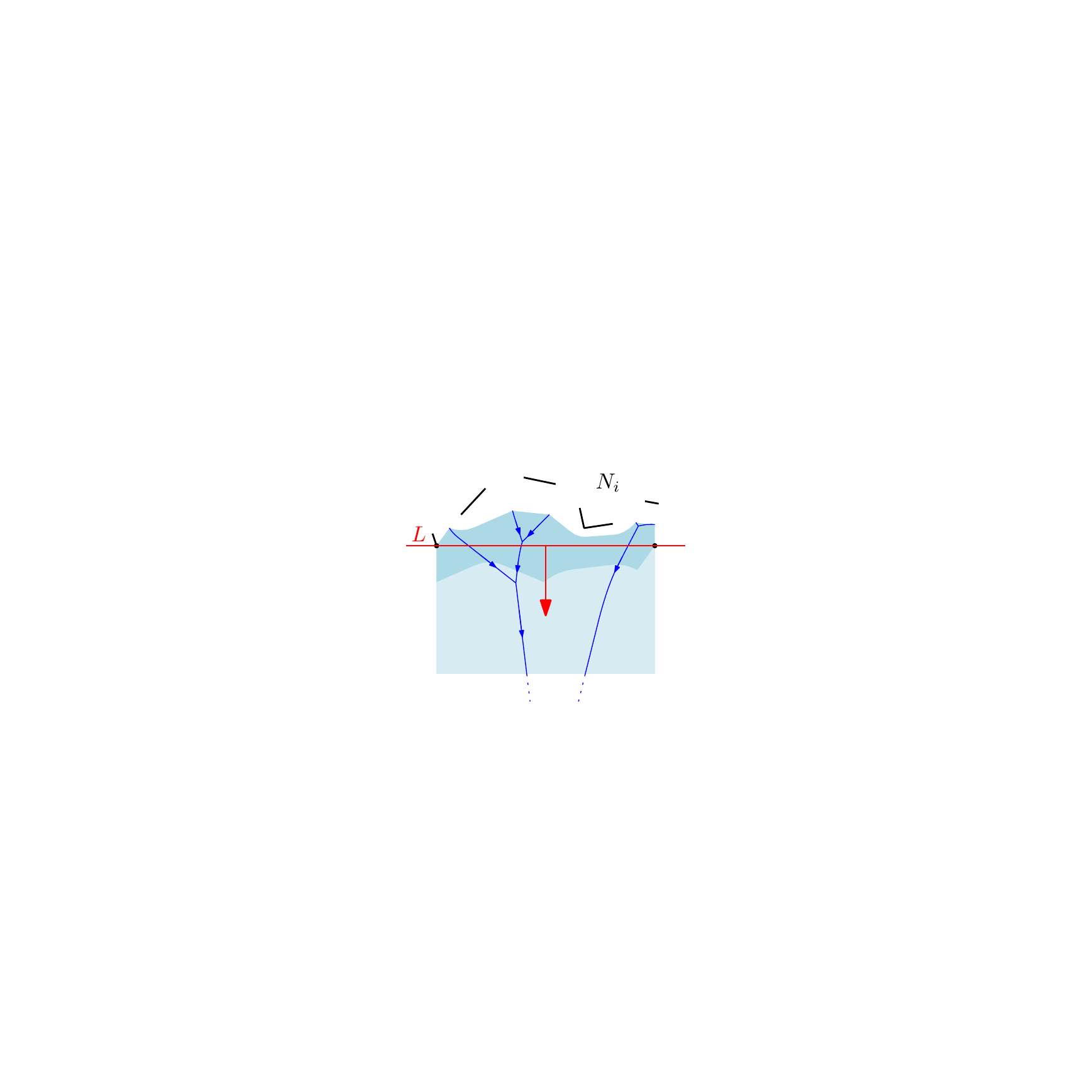}
	}%
	\subfigure[$M_{Z,j}$ \label{fig:mlma-opening-j}]{
		\includegraphics[width=0.235\textwidth]{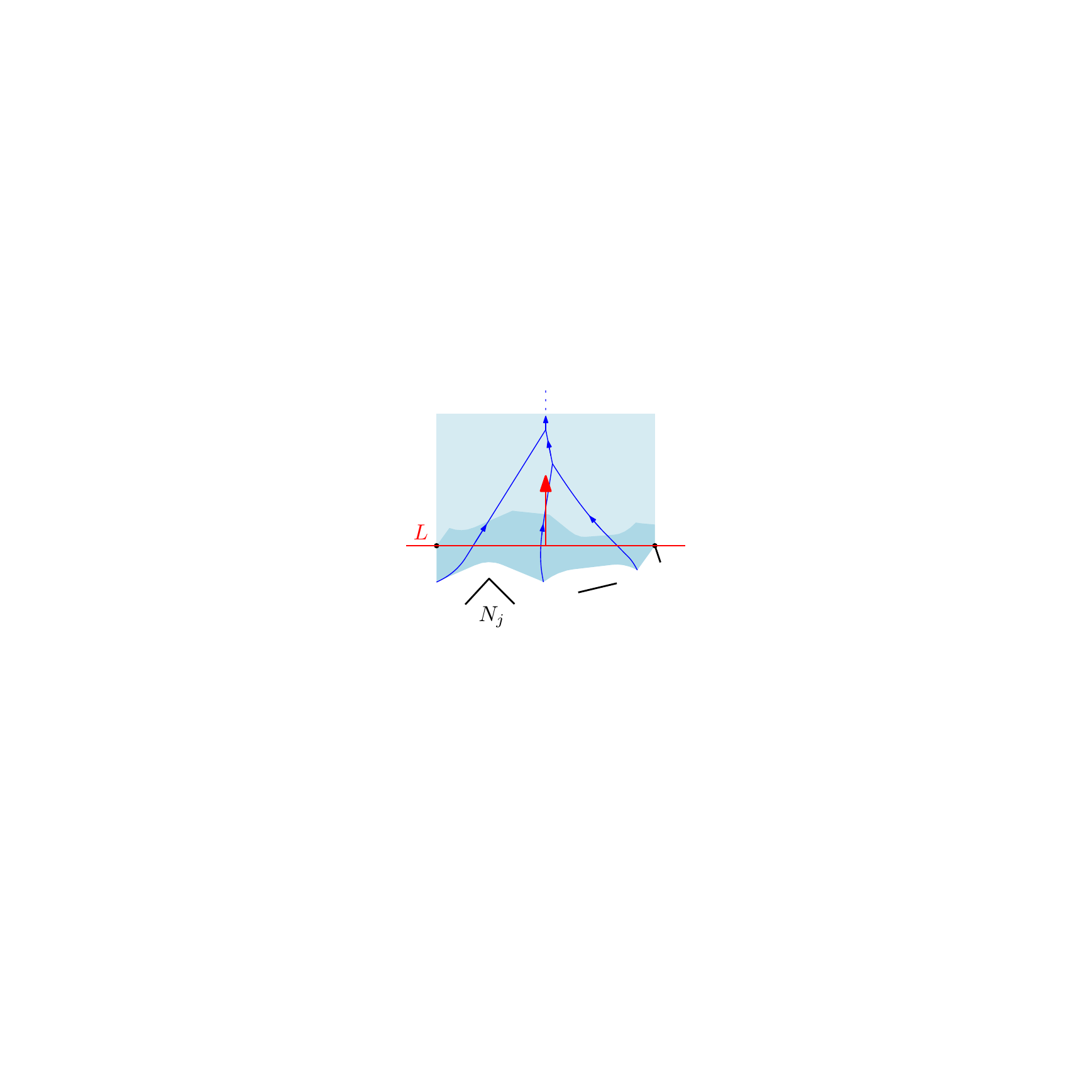}
	}%
	\subfigure[Merging into $M_Z$ \label{fig:mlma-opening-after}]{
		\includegraphics[width=0.235\textwidth]{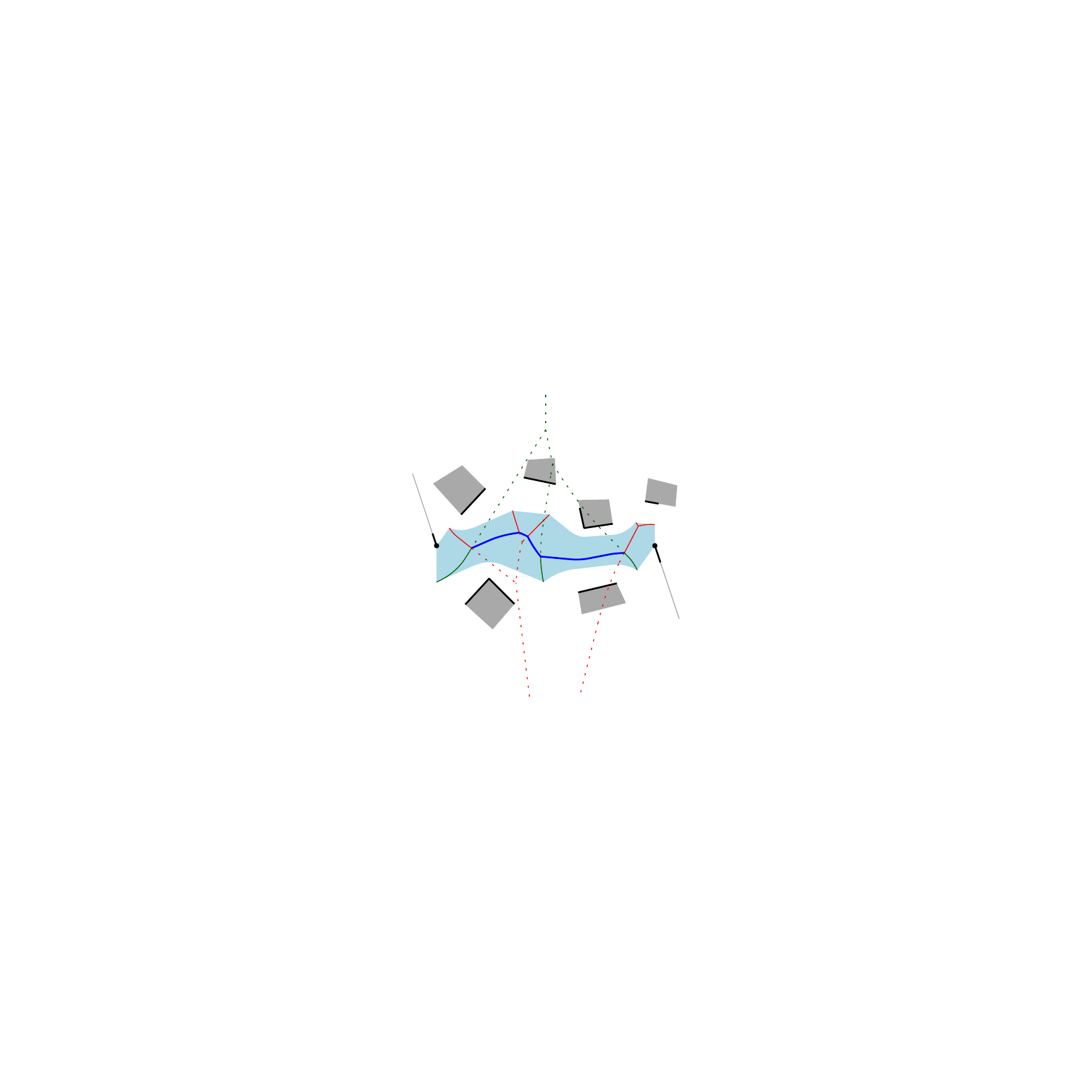}
	}
	\caption{We compute the medial axis $M_Z$ in three steps.
	\NiceSubref{fig:mlma-opening-before} The medial axis when $C_{ij}$ is still closed. 
	\NiceSubref{fig:mlma-opening-i} $M_{Z,i}$ uses only the obstacles of $N_i$ and assumes that $Z_j$ extends to infinity. 
	We compute it using a plane sweep on the ground plane $P$, starting with the sweep line $L$ at $C_{ij}$,
	\emph{without} explicitly projecting all of $N_i$ onto $P$ at the same time.
	\NiceSubref{fig:mlma-opening-j} Analogously, $M_{Z,j}$ uses only $N_j$.
	\NiceSubref{fig:mlma-opening-after} We merge the two parts to obtain $M_Z$.
	\label{fig:mlma-opening}}
\end{figure*}

\subsubsection{Computing a Single Part} \label{s:mlma:opening:parts}

To compute $M_{Z,i}$, we use the \emph{plane sweep algorithm} by Fortune \shortcite{Fortune1987-Voronoi}, 
which traces a Voronoi diagram (VD) by moving a horizontal sweep line $L$ downwards. 
This algorithm is defined for sites in 2D, but we will show how to apply it to our multi-layered problem.

A thorough analysis of Fortune's algorithm has been given by de Berg \etal \shortcite{deBerg2008-CompGeom}.
We will repeat the most important features.
The algorithm maintains an $x$-monotone `beach line' consisting of bisector arcs; each arc is defined by the sweep line $L$ and an input site above $L$.
The endpoints of these beach line arcs (which are referred to as `break points') are the centers of the largest empty disks in the environment that are tangent to $L$.
The VD below the beach line is yet to be determined.
As the sweep line moves downwards, the beach line changes, and its break points trace the edges of the VD.
There are two types of events: \emph{site} events when $L$ reaches a new site, 
and \emph{circle} events when $L$ reaches the lowest point of a circle through three sites defining adjacent arcs on the beach line.
Each event indicates that a site starts or stops generating a particular arc on the beach line.
\bigskip

\noindent
We apply Fortune's algorithm to our multi-layered problem by initializing the algorithm 
in such a way that all \emph{site} events are already handled and all \emph{circle} events can be processed just as in 2D. 
Assume without loss of generality that $C_{ij}$ is horizontal and that $Z_i$ lies above it. 
We start with a sweep line at the height of $C_{ij}$ and initialize the beach line as the sequence of arcs $\alpha(Z_i)$.
Lemma \ref{lem:mlma:InfluenceZoneBoundary2} states that this beach line is $x$-monotone, given that $C_{ij}$ is horizontal.
By the same argument, it will remain $x$-monotone during the sweep. 
After initialization, we move the sweep line downwards, and the algorithm proceeds exactly as if we were working in 2D.

The essential difference from a 2D problem is that the beach line now represents empty disks in the MLE and not on a single plane.
To explain this further, \cref{fig:mlma-sweep-overlap} shows the initial sweep line situation for the self-overlapping environment of \cref{fig:mlma-overlap}. 
An empty disk on the beach line is highlighted in blue.
If we would project all of $N_i$ onto $P$ at the same time (as in our old algorithm \cite{vanToll2011-MultiLayered}), 
then this disk would suddenly contain obstacles from other layers. 
However, these obstacles are \emph{not nearest obstacles} according to our distance function $d_P$. 
This is why the old algorithm failed: it did not respect the empty-circle property of an MLE. 
The new algorithm succeeds because it \emph{does} respect this property.

\begin{figure}[htb]
	\centering
	\includegraphics[width=0.4\textwidth]{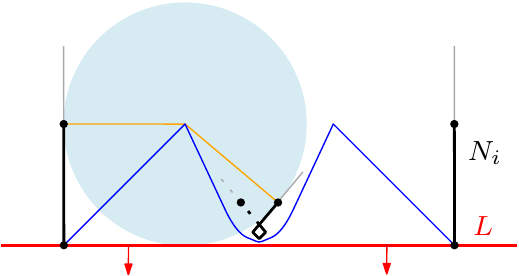}
	\caption{The environment from \cref{fig:mlma-overlap-2d} when the sweep algorithm begins. 
	A disk on the beach line may contain obstacles from other layers when projected onto $P$, but these are not \emph{nearest} obstacles. 
	\label{fig:mlma-sweep-overlap}}
\end{figure}

Each individual event in the sweep algorithm relies only on a point on the beach line and its nearest obstacles.
Due to the straight-line and empty-circle properties given in \cref{s:mlma:definition}, 
an event can be projected onto $P$, and its geometric computations will then work exactly as in 2D: 
disks are still disks, and paths to nearest obstacles are still straight line segments.
The overall algorithm is a combinatorial sequence of events, so it does not require a projection of all events onto $P$ at the same time.
Therefore, the algorithm is not affected by the multi-layered structure of $N_i$.

We will now explain further which events occur and how they can be processed.
\bigskip

\noindent \textit{Site events:} When the sweep begins, the endpoints of $C_{ij}$ lie exactly on the sweep line. 
Both endpoints induce a site event that needs to be processed immediately. 
The following lemma implies that all other sites lie above $C_{ij}$, so there are no other site events.

\begin{lemma} \label{lem:mlma:HalfPlane}
	Each point of $N_i$ either lies above $C_{ij}$ or is an endpoint of $C_{ij}$.
\end{lemma}
\begin{proof}
	We will prove that any obstacle point on side $S_i$ that does \emph{not} lie above $C_{ij}$ cannot be in $N_i$. 
	By definition, the interior of $C_{ij}$ is \emph{excluded} from $N_i$, and the endpoints of $C_{ij}$ are \emph{included}. 
	Thus, we only need to consider the \emph{other} obstacle points of $S_i$.
	
	Recall that $C_{ij}$ is still a closed obstacle when the set $N_i$ is determined. 
	This means that paths cannot yet go through $C_{ij}$.
	Let $q$ be any obstacle point on side $S_i$ that lies below or on the horizontal line $C$ through $C_{ij}$.
	Let $z$ be an arbitrary point in $Z_i$, and let $c$ be the endpoint of $C_{ij}$ that is nearest to $z$. 
	Note that the shortest path $\pi^*(z,c)$ is unobstructed.
	We can prove that $\pi^*(z,q)$ is always longer:
	\begin{itemize}
		\item If the line segment \linesegment{z}{q}\ does not intersect $C_{ij}$ when projected onto $P$, then the shortest path $\pi^*(z,q)$ is at best unobstructed, 
		so it is at least as long as \linesegment{z}{q}. 
		See \cref{fig:multilayered-proof-halfplane-a}.
		\item Otherwise, $\pi^*(z,q)$ must navigate around $C_{ij}$, so it is at best a sequence of two line segments that bends around $c$ (or the other endpoint of $C_{ij}$).
		See \cref{fig:multilayered-proof-halfplane-b}. 
	\end{itemize}
	In both cases, $\pi^*(z,q)$ is clearly longer than $\pi^*(z,c)$.
	Therefore, $q$ cannot be nearest to any point in $Z_i$, and $q$ cannot occur in $N_i$.
\end{proof}

\begin{figure}[htb]
	\centering
	\subfigure[Straight-line path \label{fig:multilayered-proof-halfplane-a}]{
		\includegraphics[width=0.4\textwidth]{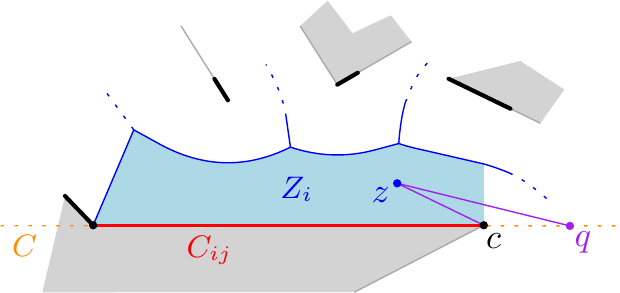}
	}
	\hspace{5mm}
	\subfigure[Path that bends around $C_{ij}$ \label{fig:multilayered-proof-halfplane-b}]{
		\includegraphics[width=0.4\textwidth]{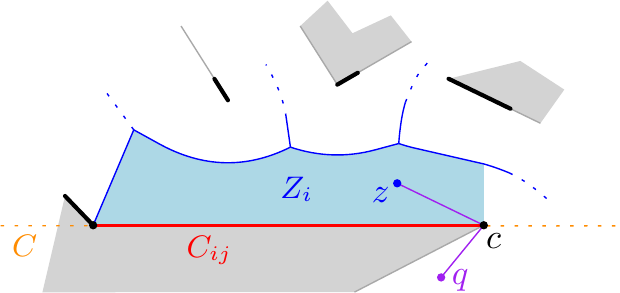}
	}
	\caption{One side $S_i$ of a horizontal connection $C_{ij}$. 
	An arbitrary point in $z \in Z_i$ is highlighted.
	Any obstacle point $q$ below or on the supporting line $C$ of $C_{ij}$ cannot belong to the neighboring obstacles $N_i$. 
	\NiceSubref{fig:multilayered-proof-halfplane-a} The case in which $\pi^*(z,q)$ does not navigate around $C_{ij}$.
	\NiceSubref{fig:multilayered-proof-halfplane-b} The case in which $\pi^*(z,q)$ navigates around $C_{ij}$.
	In both cases, an endpoint $c$ of $C_{ij}$ will be closer to $z$ than $q$ is.
	Therefore, $q$ cannot be in $N_i$.
	\label{fig:multilayered-proof-halfplane}}
\end{figure}

\noindent \textit{Circle events:} Initially, the potential circle events can be obtained by inspecting all 3-tuples of adjacent arcs in $\alpha(Z_i)$, just as in 2D.
Because we look at adjacent arcs only, we will only trace Voronoi edges of obstacles that are actually Voronoi neighbors in the MLE. 
Since $M_{Z,i}$ will be a tree and $Z_{ij}$ is planar when projected onto $P$, only adjacent Voronoi edges traced on the beach line can meet in a Voronoi vertex.
Therefore, all circle events are discovered.

The effect of a circle event is the same as in 2D.
Two of the empty disks on the beach line merge into one, a site locally disappears as a nearest site, and an arc disappears from the beach line. 
In terms of the VD, two Voronoi edges merge into a Voronoi vertex, and a new Voronoi edge starts being traced.
The disappearance of an arc from the beach line induces two new 3-tuples of adjacent arcs. 
The circles tangent to the corresponding 3-tuples of sites may induce new circle events below the sweep line.
These new events will be added to the event queue (sorted by $y$ coordinate).
\bigskip

\noindent In summary, we generate the initial events using $\alpha(Z_i)$ as the beach line. 
After that, each individual event (i.e.\ each change of a nearest site) can be processed exactly as in 2D. 
Therefore, all events are recognized, and the algorithm correctly computes $M_{Z,i}$.

As shown in \cref{fig:mlma-opening-i}, we end up tracing a tree of medial axis arcs, starting at the leaves and moving towards the root as we sweep downwards.
By Lemma \ref{lem:mlma:HalfPlane}, the endpoints of $C_{ij}$ are the lowest sites, 
so the final event occurs when the endpoints of $C_{ij}$ become the only remaining nearest obstacles to the sweep line.
These endpoints generate an infinite final edge that is perpendicular to $C_{ij}$.

\subsubsection{Merging the Two Parts} \label{s:mlma:opening:merging}

We compute $M_{Z,j}$ similarly to $M_{Z,i}$, but by starting with $\alpha(Z_j)$ as the beach line and moving the sweep line upwards instead of downwards. 
Next, we merge $M_{Z,i}$ and $M_{Z,j}$ to obtain $M_{Z}$. 
We do this by using the merge procedure for Voronoi diagrams from Shamos and Hoey \shortcite{Shamos1975-Voronoi}. 
The merge procedure traverses the Voronoi cells of $M_{Z,i}$ and $M_{Z,j}$ simultaneously and builds a new monotone sequence of Voronoi edges between them. 
Afterwards, it removes the parts of $M_{Z,i}$ and $M_{Z,j}$ that are no longer needed. 
\begin{NewContent}
\cref{fig:mlma-opening-merging} shows an adapted version of \cref{fig:mlma-opening-after} in which the result is easier to see.

\begin{figure}[htb]
	\centering
	\includegraphics[width=0.4\textwidth]{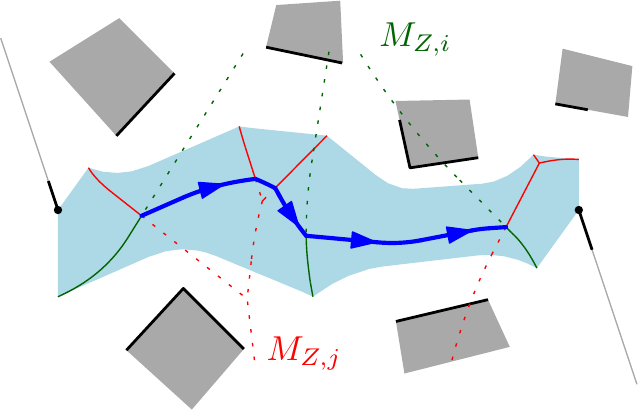}
	\caption{Merging the medial axes $M_{Z,i}$ (green) and $M_{Z,j}$ (red) to obtain $M_Z$. 
	The result is a combination of green edges, red edges, and a monotone sequence of newly generated edges (shown in bold blue) that can be computed from left to right.
	\label{fig:mlma-opening-merging}}
\end{figure}

The merge procedure is originally defined in 2D, but it relies on two main requirements that are also met in our multi-layered version of the problem.
The first requirement is that $M_{Z,i}$ and $M_{Z,j}$ must be planar; \end{NewContent} 
we have shown in \cref{s:mlma:opening:parts} that this requirement is fulfilled.
The second requirement is that $N_i$ and $N_j$ must lie in separate half-planes.
Lemma \ref{lem:mlma:HalfPlane} implies that this holds: $N_i$ and $N_j$ are separated by $C$.
The only exceptions are the endpoints of $C_{ij}$ at which the merge starts and ends.

In each step of the merge procedure, there is one nearest site (a point or a line segment) $n_i \in N_i$ and one nearest site $n_j \in N_j$, 
and the new Voronoi edge is the bisector of $n_i$ and $n_j$.
This bisector arc ends when either of the nearest sites changes.
Because the medial axes $M_{Z,i}$ and $M_{Z,j}$ are correct, they represent for all points in $Z_{ij}$ the nearest sites from $N_i$ and $N_j$, respectively. 
They therefore store all the information required to detect the changes in nearest sites.
Furthermore, the computations in each step work exactly as in 2D due to the straight-line and empty-circle properties.
Thus, the merge procedure \cite{Shamos1975-Voronoi} is not affected by the potential multi-layered nature of $N_{ij}$, and it correctly computes $M_Z$.

\subsubsection{Summary}

We now summarize our algorithm for opening a connection.
Let \Env\ be a multi-layered environment, and let \MedialAxisOld\ be the medial axis computed so far, 
in which a non-empty set of connections $\ConnectionSetOld \subseteq \mathcal{C}$ is closed.
We open a connection $C_{ij} \in \ConnectionSetOld$ to obtain \MedialAxisNew, where $\ConnectionSetNew = C' \setminus \{C_{ij}\}$.
Opening $C_{ij}$ works as follows:
\begin{enumerate}
	\item Remove the arcs from \MedialAxisOld\ that are nearest to the interior of $C_{ij}$. 
	These are the arcs of \MedialAxisOld\ that bound the influence zone $Z_{ij}$ described in \cref{s:mlma:connectionproperties}.
	\item Compute $M_Z = \MedialAxisNew \cap Z_{ij}$ as described in \cref{s:mlma:opening:parts,s:mlma:opening:merging}. 
	\item Insert $M_Z$ into \MedialAxisOld\ to obtain \MedialAxisNew.
\end{enumerate}

\subsection{Algorithm Correctness}

The overall construction algorithm outlined in \cref{s:mlma:algorithm-outline} first computes the medial axis with all connections as closed obstacles. 
It then iteratively opens a closed connection, using the algorithm from \cref{s:mlma:opening}, until all connections are open.
The following theorem states that this algorithm correctly computes the multi-layered medial axis:

\begin{theorem}
	Let \Env\ be an MLE.
	Computing \MedialAxis{\Env,\mathcal{C}}\ and then iteratively opening each connection in $\mathcal{C}$ as described in \cref{s:mlma:opening} 
	yields the medial axis \MedialAxis{\Env}.
\end{theorem}
\begin{proof}
	Each iteration of this algorithm starts with a correct medial axis \MedialAxisOld\ and computes a correct medial axis \MedialAxisNew\ 
	in which one more connection has been removed as an obstacle. 
	By induction over the number of iterations, the final result is the correct medial axis \MedialAxis{\Env} in which all connections are traversable. 
\end{proof}

The connections can be opened in any order without affecting the correctness of the algorithm. 
However, we will see in the next section that opening the connections in a particular order can affect the \emph{running time} of the algorithm.


\subsection{Algorithm Complexity} \label{s:mlma:complexity}

In this section, we analyze the asymptotic running time of our construction algorithm.
The first step of the algorithm computes the medial axis of all layers with all connections closed.
This can be achieved using a single 2D algorithm because the medial axis consists of separate 2D components that do not yet influence each other.
Lemma \ref{lem:mle:NumberOfConnections} has shown that the number of connections $k$ is linear in the number of obstacle points $n$ in the MLE.
Thus, the presence of $k$ connections does not affect the asymptotic complexity, and we essentially compute a 2D medial axis of an input with complexity \BigO{n}. 
\cref{s:2dma:complexity} has shown that this can be performed in \BigO{n \log n} time. 

\subsubsection{Running Time to Open One Connection} \label{s:mlma:complexity:opening}

\begin{lemma} \label{lem:mlma:OpeningTime}
	A single connection $C_{ij}$ can be opened in \BigO{m \log m} time, where $m$ is the complexity of the neighbor set $N_{ij}$.
\end{lemma}
\begin{proof}
	Our algorithm for opening a connection starts with two instances of a sweep line algorithm \cite{Fortune1987-Voronoi}. 
	Both instances take \BigO{m \log m} time because they involve \BigO{m} circle events that need to be maintained in sorted order.
	Next, we perform one \BigO{m}-time merge step of a divide-and-conquer algorithm \cite{Shamos1975-Voronoi}.
	Therefore, the total running time is \BigO{m \log m}.
\end{proof}

In many practical scenarios, the connections and obstacles are spread throughout the environment, 
and $m$ will be constant in most iterations of the algorithm.
However, if many obstacles are close to the connection, $m$ can be $\BigTheta{n}$.

\subsubsection{Total Running Time} \label{s:mlma:complexity:total}

Based on Lemma \ref{lem:mlma:OpeningTime}, iteratively opening \emph{all} connections takes \BigO{\SumOfIterations} time in total, 
where $m_i$ is the neighbor set complexity of the connection that is opened in iteration $i$.
Note that the neighbor sets of closed connections can change during the algorithm because obstacles may turn out to influence other connections when a nearby connection is opened.
In other words, a neighbor set's complexity depends on what lies beyond the nearby connections that are already open.
This suggests that the total construction time depends on the \emph{order} in which the connections are opened. 

In many environments, each obstacle will only affect the algorithm in a constant number of iterations regardless of this order, 
which means that the total construction time will remain \BigO{n \log n}.
However, there are environments in which opening the connections in an `unlucky' order leads to a worse construction time. 
The following lemma analyzes this:

\begin{lemma} \label{lem:mlma:WorstCaseTime}
	There exists an environment with $n$ obstacle vertices and $k$ connections such that opening the connections in an inefficient order 
	gives \BigTheta{k} neighbor sets of complexity \BigTheta{n}.
\end{lemma}
\begin{proof}
	\cref{fig:mlma-worst-case} shows an example in which a ramp has been subdivided into a chain of small layers.
	All connections are close together, and the bottom connection has a row of \BigTheta{n} neighboring obstacles on the ground plane.
	If the connections are opened from the bottom to the top, the first neighbor set has complexity \BigTheta{n}.
	When the first connection is open, the same obstacles have become neighbors of the second connection, so the second neighbor set will also have complexity \BigTheta{n}.
	By repeating this argument, we see that this holds for each of the $k$ iterations. 
\end{proof}

\begin{figure}[htb]
	\centering
	\includegraphics[width=0.45\textwidth]{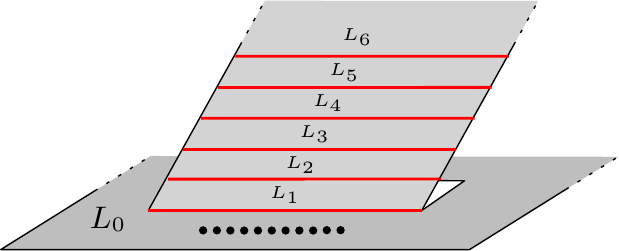}
	\caption{A worst-case example for our incremental algorithm. 
	A ramp has been subdivided into a sequence of small layers. The ground floor contains a row of \BigTheta{n} obstacles.
	If we open the connections from bottom to top, each connection will have \BigTheta{n} obstacles in its neighbor set.
	\label{fig:mlma-worst-case}}
\end{figure}

When using the \BigO{m \log m}-time algorithm from \cref{s:mlma:opening} in each iteration, 
the total running time for this unfortunate order becomes \BigO{kn \log n}. 
In previous work \cite{vanToll2011-MultiLayered}, we reported this as the worst-case running time of our algorithm.

However, the following lemmas show that we can \emph{always} construct the medial axis in \BigO{n \log n \log k} time by choosing an `easy' connection in each iteration.
We begin by showing that the medial axis has linear size throughout the entire algorithm.

\begin{lemma} \label{lem:mlma:LinearSize}
	For any MLE, the medial axis has size \BigO{n} in each iteration of our algorithm.
\end{lemma}
\begin{proof}
	In the initial step of the algorithm, we compute a medial axis of \BigO{n} sites, which has size \BigO{n}.
	Since opening a connection is analogous to deleting a Voronoi site, the asymptotic complexity of the graph cannot increase during the algorithm. 
\end{proof}

\begin{lemma} \label{lem:mlma:EasiestConnection}
	When $q$ connections are still closed, there is at least one connection with a neighbor set complexity of \BigO{\frac{n}{q}}.
\end{lemma}
\begin{proof}
	By Lemma \ref{lem:mlma:LinearSize}, the medial axis always has \BigO{n} arcs.
	At any point in the incremental algorithm, every arc bounds the influence zone of at most two connections, namely one on each side of the arc.
	Therefore, the combined complexity of all influence zones (or, equivalently, of all neighbor sets) is \BigO{n}.
	When this \BigO{n} complexity is shared by $q$ connections, there must be at least one connection with a neighbor set complexity of \BigO{\frac{n}{q}}. 
\end{proof}

\begin{lemma} \label{lem:mlma:OpeningAllConnections}
	The $k$ connections can be opened in \BigO{n \log n \log k} total time.
\end{lemma}
\begin{proof}
	We will repeatedly open the connection with the smallest neighbor set complexity.
	To achieve this, we first compute the neighbor set complexities of all $k$ closed connections. 
	This can be done in \BigO{n} time by traversing the medial axis once and incrementing the complexity for a connection $C_x$ whenever an arc has $C_x$ as a nearest obstacle.
	Next, we sort the connections by complexity in a balanced binary search tree $\mathcal{T}$.
	This requires \BigO{k \log k} time, which is \BigO{n \log n} because $k$ is \BigO{n}.
	
	Assume for now that we can maintain the sorting order in $\mathcal{T}$ such that we can always get the connection with the smallest complexity.
	Let $C_q$ be this easiest connection when $q$ connections are closed.
	By Lemma \ref{lem:mlma:EasiestConnection}, it has a complexity of \BigO{\frac{n}{q}}. 
	Using the algorithm of \cref{s:mlma:opening}, we can open it in \BigO{\frac{n}{q} \log \frac{n}{q}} time.
	
	Next, we show that $\mathcal{T}$ can indeed be maintained efficiently.
	The neighbor sets of other closed connections can change due to opening $C_q$.
	However, any connection that is affected \emph{must} be one of the neighboring obstacles of $C_q$.
	Therefore, the number of neighbor sets that can change is \BigO{\frac{n}{q}}.
	We can update the complexities of these neighbor sets in \BigO{\frac{n}{q}} time: 
	for each added or removed arc with a connection $C_x$ as a nearest obstacle, we increment or decrement the neighbor set complexity for $C_x$.
	Afterwards, we update the search tree $\mathcal{T}$ by deleting and re-inserting all complexities that have changed. 
	This takes \BigO{\frac{n}{q} \log q} time because it requires \BigO{\frac{n}{q}} update operations.
	Thus, opening $C_q$ and updating $\mathcal{T}$ afterwards takes 
	$\BigO{\frac{n}{q}(\log \frac{n}{q} + \log q)} = \BigO{\frac{n}{q} \log n}$ time.
	
	Because $\mathcal{T}$ is maintained correctly, we can always open the connection with the lowest neighbor set complexity.
	For all $k$ iterations combined, we obtain a running time of 
	$\BigO{\sum_{q=1}^k \frac{n}{q} \log n} = \BigO{\sum_{q=1}^k (n \log n \cdot \frac{1}{q})} = \BigO{n \log n \sum_{q=1}^k \frac{1}{q}} = \BigO{n \log n \cdot H_k}$.
	Here, $H_k$ is the $k$th harmonic number, which is known to be \BigTheta{\log k}.
	Therefore, the total running time to open all connections is \BigO{n \log n \log k}.
\end{proof}

Combined with the \BigO{n \log n} running time for the first step (i.e.\ computing the medial axis of each layer with all connections closed), 
we obtain the following result:

\begin{theorem}\label{th:mlma:RunningTime}
	The medial axis of a multi-layered environment with $n$ obstacle vertices and $k$ connections can be computed in \BigO{n \log n \log k} time.
\end{theorem}

\subsection{Storage Complexity} \label{s:mlma:complexity:size}

Finally, we give the storage complexity of the multi-layered medial axis when all connections have been opened. 
It follows immediately from Lemma \ref{lem:mlma:LinearSize}: the complexity is \BigO{n} at each point in the algorithm, including at the end.

\begin{theorem}
	The medial axis of a multi-layered environment with $n$ obstacle vertices and $k$ connections has a storage complexity of \BigO{n}.
\end{theorem}

%% file: section-applications.tex
\section{Application: The Explicit Corridor Map} \label{s:ecm}

\begin{NewContent}
The Explicit Corridor Map (ECM) is a navigation mesh that is closely related to the medial axis.
In this section, we give an improved definition of the ECM, we analyze its complexity, and we show several useful geometric operations for path planning.
\end{NewContent}


\subsection{Definition} \label{s:ecm:definition}

The ECM is a graph representation of the medial axis annotated with nearest-obstacle information. 
It describes each medial axis arc and its surrounding free space in an efficient manner. 
As such, it is a compact navigation mesh that can be used to find paths for characters of any radius.

\begin{figure*}
	\centering
	\subfigure[Medial axis \label{fig:u-ma2}]{
		\includegraphics[height=0.3\textwidth]{fig/u-ma.pdf}
	}
	\subfigure[Explicit Corridor Map \label{fig:u-ecm}]{
		\includegraphics[height=0.3\textwidth]{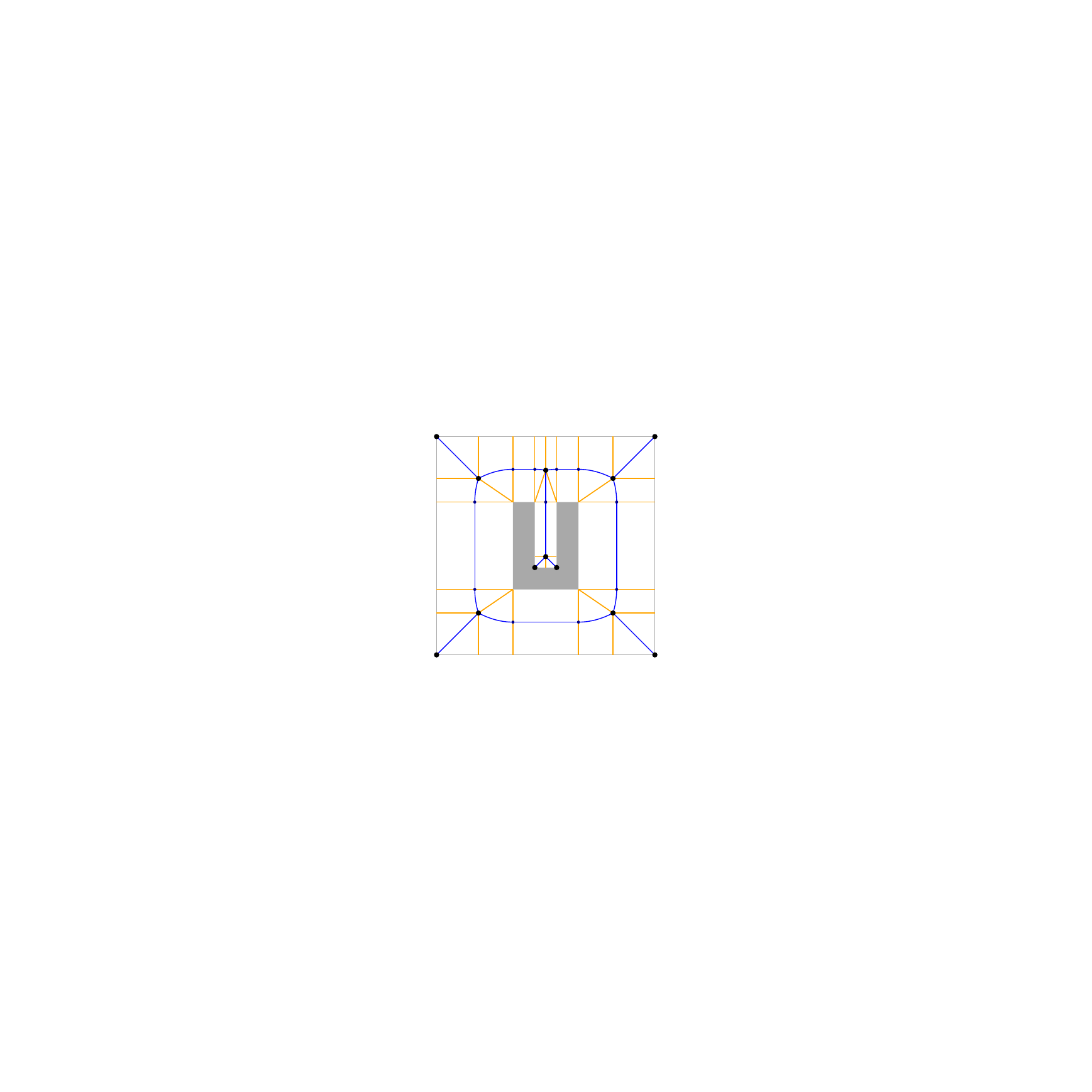}
	}
	\subfigure[ECM edge \label{fig:u-ecm-closeup}]{
		\includegraphics[height=0.3\textwidth]{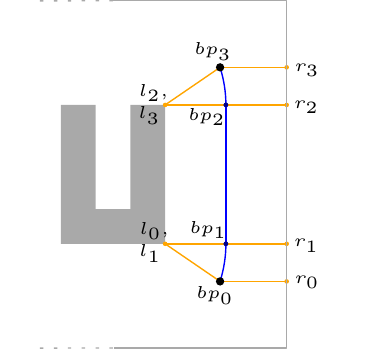}
	}
	\caption{\NiceSubref{fig:u-ma2} The medial axis of a 2D environment, repeated from \cref{fig:u-ma}.
	\NiceSubref{fig:u-ecm} The ECM is a medial axis with nearest-obstacle annotations,
	shown as orange line segments between vertices and their nearest obstacle points.
	These segments are not edges in the graph.
	\NiceSubref{fig:u-ecm-closeup} Details of an ECM edge with four bending points. 
	Each bending point $bp_k$ stores its position $p_k$ and its nearest obstacle points $l_k$ and $r_k$.
	\label{fig:u-ma-ecm}}
\end{figure*}

\begin{definition}[Explicit Corridor Map]
For a 2D environment, WE, or MLE, the Explicit Corridor Map $\textit{ECM}(\Env)$ 
is an extended representation of the medial axis \MedialAxis{\Env}\ as an undirected graph $G = (V,E)$ \NewContentInline{with the following properties:}
\begin{itemize}
	\item $V$ is the set of true vertices of the medial axis (i.e.\ all medial axis vertices except those of degree 2).
	\item $E$ is the set of edges of the medial axis.
	\item Each edge $e_{ij} \in E$ represents the medial axis arcs between two true vertices $v_i,v_j \in V$.
	It is represented by a sequence of $n' \geq 2$ bending points\footnote{We originally referred to bending points as \emph{event points}. \cite{Geraerts2010-ECM}} 
	$bp_0, \ldots, bp_{n'-1}$ where $bp_0=v_i$, $bp_{n'-1}=v_j$,
	and $bp_1, \ldots, bp_{n'-2}$ are the remaining semi-vertices along the edge.
	\item Each bending point is a medial axis vertex annotated with nearest-obstacle information.
	A bending point $bp_k$ on an edge stores its two nearest obstacle points $l_k$ and $r_k$ on the left and right side of the edge, respectively.
\end{itemize}
\end{definition}

\cref{fig:u-ecm} shows the ECM of our example environment.
Since the ECM is an undirected graph, any edge $e_{ij}$ could also be described as an edge $e_{ji}$, 
with the list of bending points reversed and all left and right obstacle points swapped.
Furthermore, a true vertex occurs as the first or last bending point for each of its incident edges. 
Each such bending point has its own sense of left and right; together, they store all nearest obstacle points for the true vertex.
Thus, it is sufficient to store only two obstacle points for each bending point.
We also emphasize that the orange line segments in \cref{fig:u-ecm} are \emph{not} graph edges; 
they merely denote the relation between bending points and their nearest obstacles.
\cref{fig:u-ecm-closeup} shows the details of an ECM edge.

Annotating the medial axis with nearest-obstacle information has many advantages.
One advantage is that the \emph{clearance} (the distance to the nearest obstacle) is known at each bending point. 
This enables path planning for characters of any radius; that is, we do not have to inflate the obstacles using Minkowski sums for a particular radius.
\bigskip

\noindent
\begin{NewContent}
Our definition of the ECM applies not only to 2D environments, but also to walkable and multi-layered environments without requiring any adjustments.
In a WE or MLE, the nearest obstacle to a point $q \in \Efree$\ may lie in another layer than $q$ itself. 
Therefore, a nearest obstacle point to an ECM bending point $bp_k$ may lie in another layer than $bp_k$. 
This does not change the ECM's definition or complexity.
Due to the straight-line property, the path from a bending point to any of its nearest obstacle points is a line segment when projected onto $P$.


\subsection{Complexity} \label{s:ecm:complexity}

We have shown that the medial axis has \BigO{n} complexity in a 2D environment with $n$ obstacle vertices (\cref{s:2dma:complexity}), 
in a WE with $n$ boundary vertices (\cref{s:mlma:complexity:size}), 
and in an MLE with $n$ boundary vertices (regardless of the number of connections). 
The next lemma states that the medial axis can easily be converted to an ECM of the same complexity.
\end{NewContent}
	
\begin{lemma}
	A medial axis with complexity \BigO{n} can be converted to an ECM of complexity \BigO{n} in \BigO{n} time.
\end{lemma}
\begin{proof}
	The ECM converts medial axis vertices to bending points by adding nearest-obstacle annotations. 
	\NewContentInline{A degree-2 vertex becomes a single bending point, and any other vertex receives a separate bending point for each incident edge. 
	The nearest-obstacle annotations can easily be added in a post-processing step (or even during the construction algorithm itself). 
	After all, a medial axis arc is defined by an obstacle part (a line segment or a point) on both sides of the arc.
	The two nearest obstacle \emph{points} for a bending point $bp_k$ are simply the nearest points to $bp_k$ on these obstacle parts. 
	Thus, adding these annotations takes constant time per bending point.}

	What remains to be analyzed is the total number of bending points.
	Each medial axis vertex of degree 1 or 2 occurs as a bending point exactly once. 
	A vertex of degree $d \geq 3$ occurs as a bending point $d$ times: once for each edge that contains this vertex as an endpoint.
	Since the sum of all vertex degrees is \BigO{n}, there are \BigO{n} bending points in total, 
	each of which requires \BigO{1} storage.
	Thus, the ECM adds a linear amount of information to the medial axis in linear time. 
\end{proof}

\NewContentInline{The remainder of this section defines a number of operations on the ECM that are useful for path planning.}
These concepts apply to both 2D and multi-layered environments.
For more information on how these concepts fit into a generic crowd simulation framework, 
we refer the reader to an overview paper \cite{vanToll2015-Framework}.


\begin{NewContent}
\subsection{Computing Nearest Obstacles and Retractions} \label{s:ecm:retraction}
\end{NewContent}
\noindent The ECM event points and their nearest-obstacle annotations \NewContentInline{subdivide the free space \Efree\ into non-overlapping polygonal cells. 
This subdivision has \BigO{n} edges and vertices.} 
We can therefore perform \emph{point-location queries} to find out in which ECM cell a query point is located.
Such queries can be answered in \BigO{\log n} time using e.g.\ a trapezoidal map, which can be built in \BigO{n \log n} time \cite{deBerg2008-CompGeom}.
For MLEs, we create a separate point-location data structure for each layer. 
A query point in an MLE is specified by a 2D point and a layer ID.

Once the cell containing a query point $p$ has been determined, we can compute the \emph{nearest obstacle point} $np(p)$ in constant time because each cell has constant complexity.
In a crowd simulation application, we can use this to easily determine how far each character is removed from the nearest boundary.
\bigskip

\noindent
Points in the free space can be \emph{retracted} onto the medial axis.
In robot motion planning, the term `retraction' is used for a function that maps points in \Efree\ onto the medial axis \cite{Wilmarth1999-Retraction},
as well as for complete planning methods based on this principle \cite{ODunlaing1985-Retraction}.
We use the following definition:

\begin{definition}[Retraction]
For any point $p$ in the free space \Efree, the \emph{retraction} \Retraction{p} is a unique projection of $p$ onto the medial axis.
\begin{enumerate}
	\item If $p$ lies on the medial axis, then $\Retraction{p} = p$.
	\item If $p$ does not lie on the medial axis, let $l$ be the half-line that starts at $np(p)$ and passes through $p$.
	\Retraction{p} is the first intersection of $l$ with the medial axis.
\end{enumerate}
\end{definition}

\cref{fig:u-retraction} shows examples of retractions. 
A retraction can be computed in \BigO{\log n} time by using a point-location query followed by constant-time geometric operations.

\begin{figure*}[t]
	\centering
	\subfigure[Retractions \label{fig:u-retraction}]{
		\includegraphics[height=0.3\textwidth]{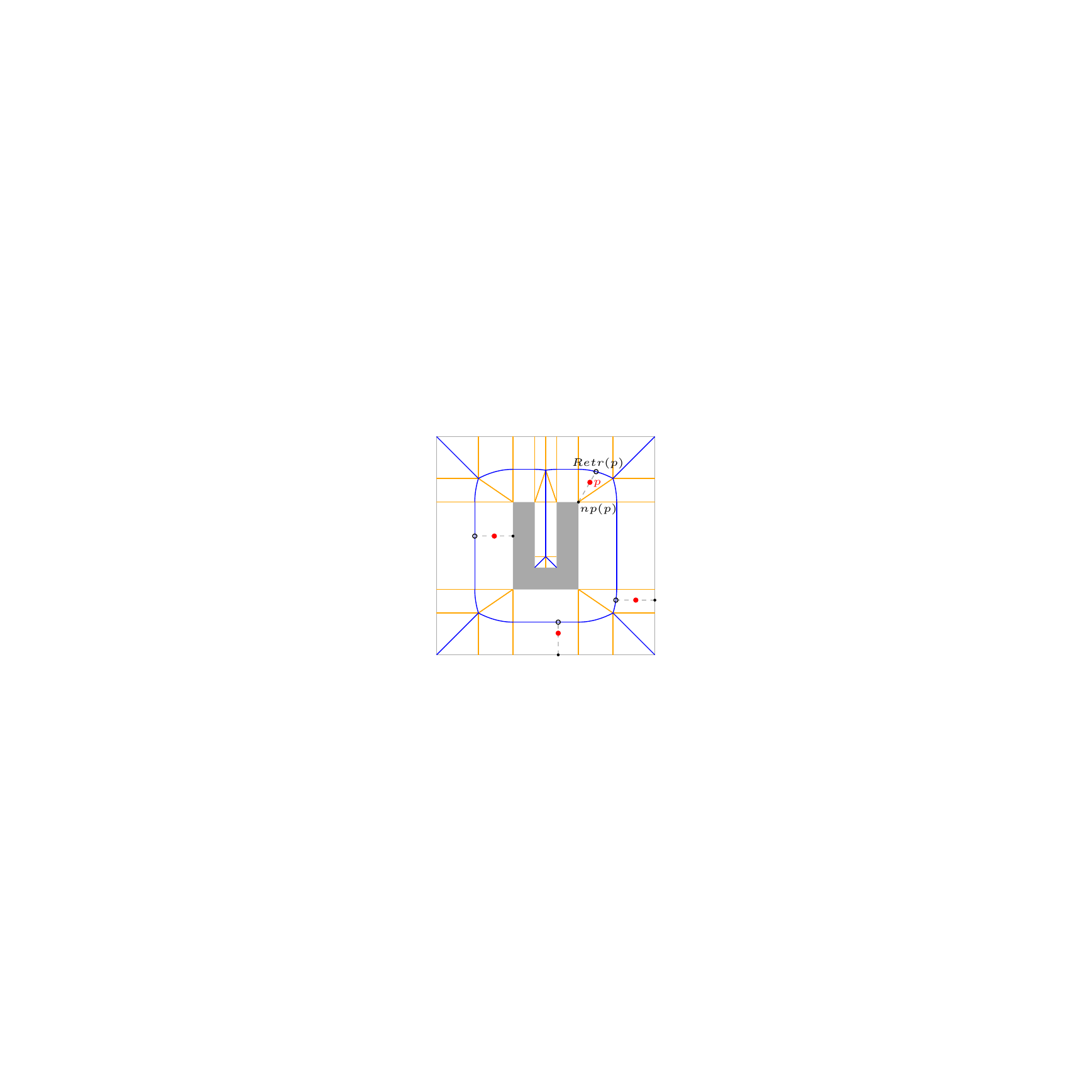}
	}
	\subfigure[ECM path + corridor \label{fig:u-path-corridor}]{
		\includegraphics[height=0.3\textwidth]{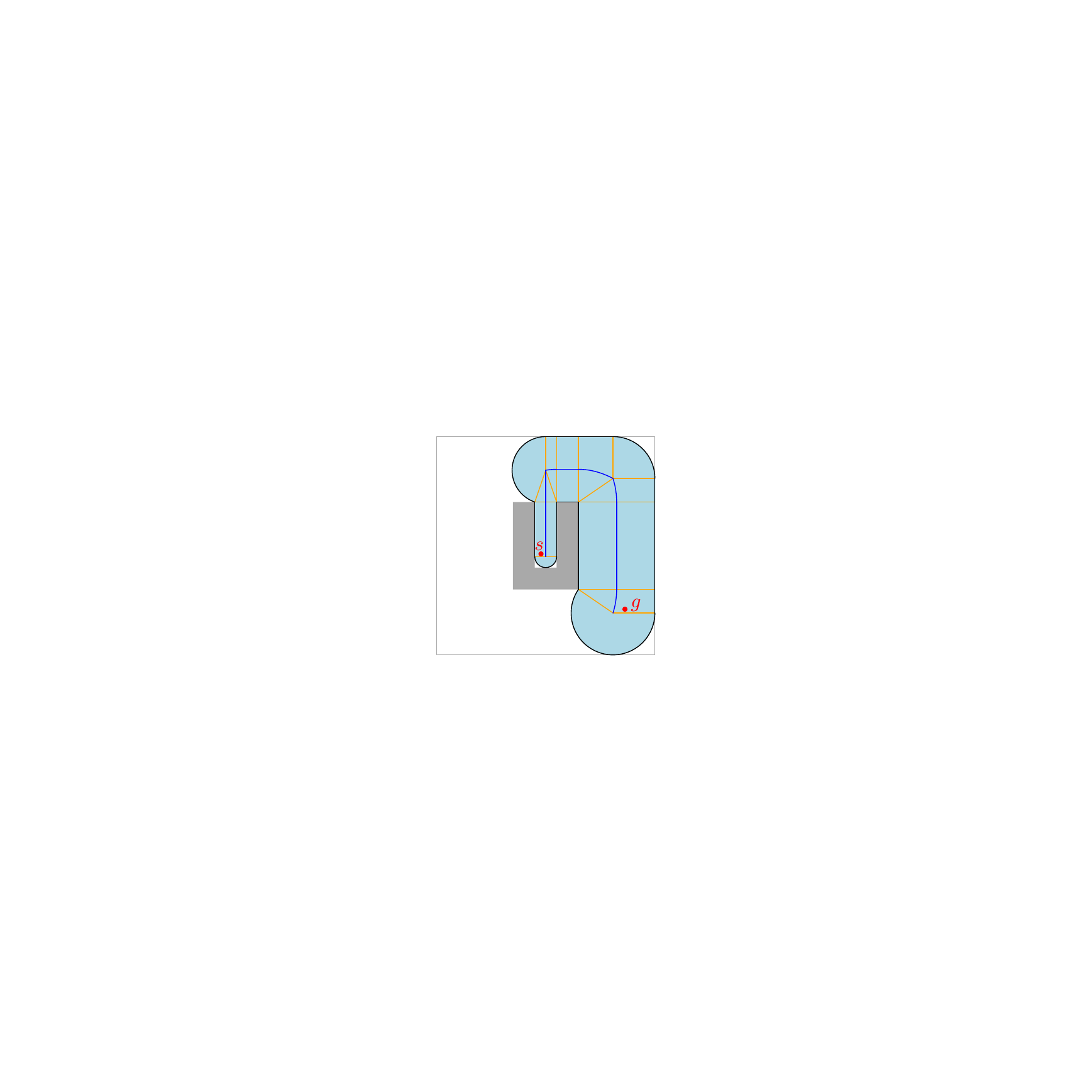}
	}
	\subfigure[Indicative routes \label{fig:u-indicativeroutes}]{
		\includegraphics[height=0.3\textwidth]{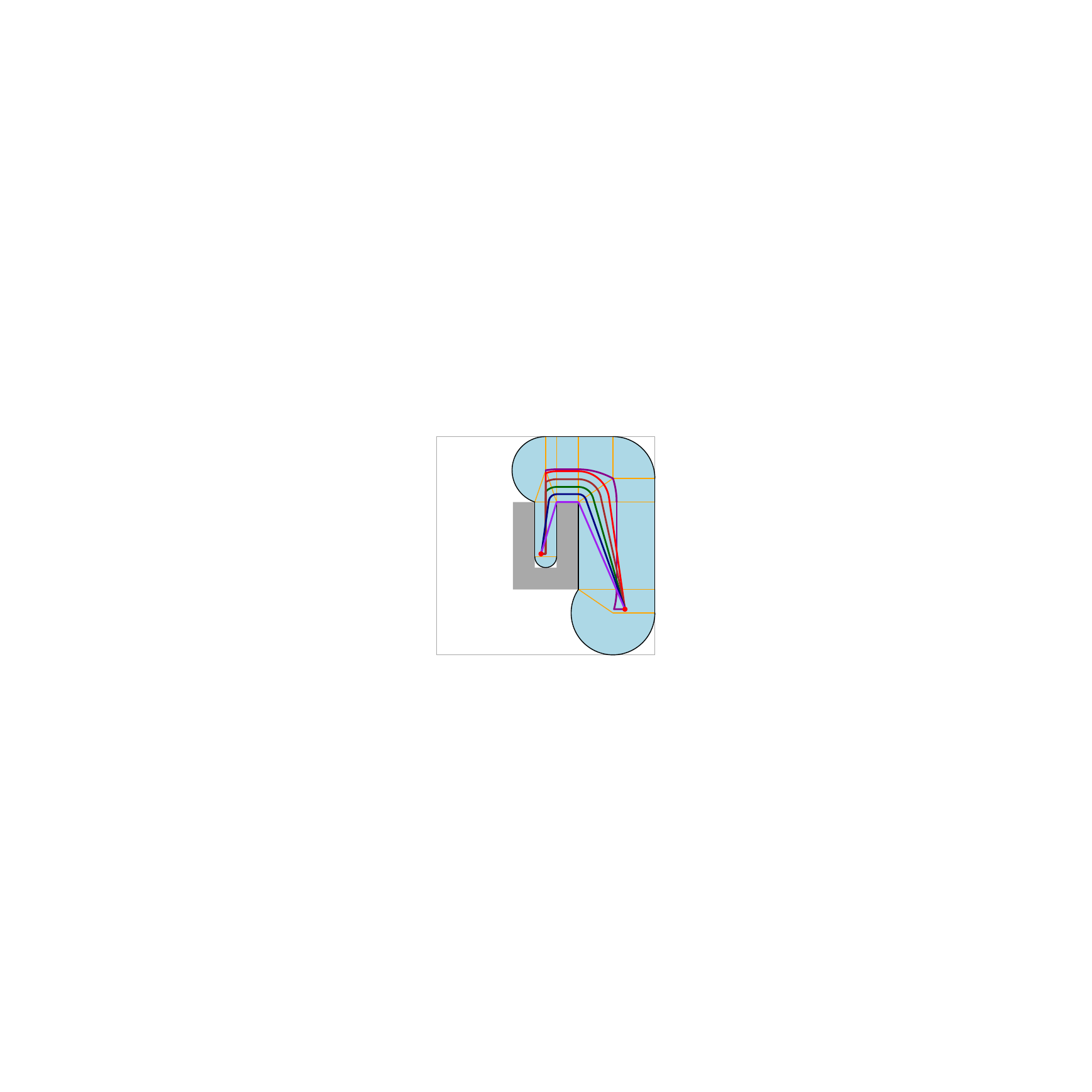}
	}
	\caption{
	\NiceSubref{fig:u-retraction} Examples of query points (shown as large dots), their nearest obstacle points (small dots), and their retractions (circles).
	\NiceSubref{fig:u-path-corridor} Given two positions $s$ and $g$, the retraction method is used to compute a path from $s$ to $g$ along the medial axis.
	A corridor describes the free space around this path.
	\NiceSubref{fig:u-indicativeroutes} Within the corridor, we can compute various types of indicative routes, e.g.\ with an amount of preferred clearance. 
	}
\end{figure*}


\subsection{Computing a Path} \label{s:ecm:path}

To plan a path for a disk-shaped character in the ECM, we first find a path along the medial axis that has sufficient clearance.
This is equivalent to the \emph{retraction method} for motion planning \cite{ODunlaing1985-Retraction}: 
given a start position $s$ and a goal position $g$ in \Efree, we compute their retractions, 
and then we compute an optimal path on the medial axis from \Retraction{s} to \Retraction{g} using the A* search algorithm.
This search is efficient because the medial axis is a sparse graph compared to e.g.\ a grid.
The clearance information stored in each ECM bending point allows us to precompute the \emph{minimum clearance} along each edge. 
The search can then skip edges for which the clearance is too low for our disk to pass through.

The free space around a medial axis path can be described using a \emph{corridor}, 
which is the sequence of ECM cells along the path combined with the maximum-clearance disks at its ECM vertices \cite{Geraerts2010-ECM}.
\cref{fig:u-path-corridor} shows an example.


\subsection{Computing an Indicative Route} \label{s:ecm:indicativeroute}

An ECM path can be converted into an \emph{indicative route}: \NewContentInline{a curve for the character to follow.} 
Various types of indicative routes can be obtained in \BigO{m} time, where $m$ is the number of ECM cells along the path.
For instance, we can use a funnel algorithm to obtain the \emph{shortest path} within a corridor, while keeping a preferred distance to obstacles whenever possible \cite{Geraerts2010-ECM}.
Examples are displayed in \cref{fig:u-indicativeroutes}.
It is also easy to compute indicative routes that stay on the left or right side of the free space (or any interpolation of these extremes).
Varying the `side preference' among characters is a convenient way to obtain diversity in the crowd.


\subsection{Dynamic Updates} \label{s:ecm:dynamic}

In dynamic environments, large obstacles can appear or disappear during the simulation. 
In previous work \cite{vanToll2012-Dynamic}, we have presented algorithms that update the ECM locally due to the insertion or deletion of a convex polygonal obstacle $P$.
A dynamic insertion takes \BigO{m + \log n} time, where $m$ is the combined complexity of the neighboring obstacles for $P$. 
A dynamic deletion requires \BigO{m \log m + \log n} time. 
\NewContentInline{For more details, we refer interested readers to the corresponding publication \cite{vanToll2012-Dynamic}.}
After a dynamic event, a character can efficiently re-plan a new optimal path in the updated mesh based on its previous path \cite{vanToll2015-Replanning}.


%% file: section-implementation.tex
\section{Implementation} \label{s:implementation}

We have implemented the Explicit Corridor Map as part of a crowd simulation framework \cite{vanToll2015-Framework}.
The software was written in C++ in Visual Studio 2013.

To compute the \NewContentInline{medial axis and} ECM, 
we have integrated two different libraries for computing Voronoi diagrams: Vroni \cite{Held2011-Vroni} and a package of Boost \shortcite{Boost}.
Since the Boost Voronoi library requires integer coordinates as input, we multiply all coordinates by $10{,}000$ and round them to the nearest integer.
For convenience, we use these rounded coordinates in Vroni as well.
We use meters as units, so this scaling implies that we represent all coordinates within a precision of $0.1$ millimeters.

In an earlier publication \cite{vanToll2011-MultiLayered}, we used an approximating GPU-based ECM implementation based on the work of Hoff \etal \shortcite{Hoff1999-GVD}, 
However, we will not report the details of this approach because it has proven to be less efficient and less practical than the other implementations.

Both Vroni and Boost assume that the input sites are interior-disjoint line segments. 
In practice, environments are often drawn by hand and may contain overlapping geometry.
Therefore, before computing the ECM, we use another component of Boost to convert obstacles to interior-disjoint segments,
using the scaled integer coordinates described earlier.
After computing the ECM, we remove all graph components that lie inside obstacle polygons. 
These steps will be included in our time measurements.

In some environments, the medial axis may contain edges that run across many layers.
We ensure that each edge can be associated with a single layer, mainly for visualization purposes.
We do this by splitting each edge wherever it intersects one of the (now opened) connections.
Computing these intersections takes extra time, 
and it can increase the worst-case complexity of the graph to \BigO{kn} if many edges intersect many connections, such as in \cref{fig:mlma-worst-case}. 
We consider this post-processing to be optional and not part of the main algorithm.

\cref{s:experiments:ecm} will show that our implementation of the multi-layered ECM construction algorithm is very fast in practice. 
For future work, there are two potential improvements. 
First, we currently open the connections in the order in which they are listed in the environment, which is not necessarily optimal.
Second, we open the connections using our old algorithm \cite{vanToll2011-MultiLayered}, so we cannot yet handle self-overlap near connections such as in \cref{fig:mlma-overlap}.
However, these theoretical issues are not a problem for any of the real-world environments in our test set.

\begin{NewContent}
For simplicity, we have implemented indicative routes as piecewise linear curves. 
Circular or parabolic arcs in an indicative route are approximated by sequences of line segments. 
\end{NewContent}

%% file: section-experiments.tex
\section{Experiments} \label{s:experiments}

This section assesses the performance of our ECM implementations in a range of 2D and multi-layered environments.
All experiments were run on a Windows 7 PC with a 3.20 GHz Intel i7-3930K CPU, an NVIDIA GeForce GTX 680 GPU, and 16 GB of RAM. 
Only one CPU core was used, except at the end of \cref{s:experiments:ecm} where we will use multi-threading to improve the performance in MLEs.

\subsection{Environments}

The 2D environments are shown in \cref{fig:exp-environments-2d}; more details can be found in the upper part of \cref{tab:exp-environments}.
\emph{Military} is a simple environment with a small number of obstacles. \emph{City} is a more complex virtual city.
\emph{Zelda} is an environment from a computer game.
\emph{Zelda2x2}, \emph{Zelda4x4}, and \emph{Zelda8x8} are adapted versions of \emph{Zelda} that have been duplicated in a $2 \times 2$, $4 \times 4$, and $8 \times 8$ grid pattern. 
\NewContentInline{We have also used these environments in previous publications \cite{vanToll2012-Dynamic,vanToll2016-ComparativeStudy}.}

\begin{figure*}[ht]
	\centering
	\subfigure[Military	\label{fig:exp-military}]{
		\includegraphics[width=0.23\textwidth]{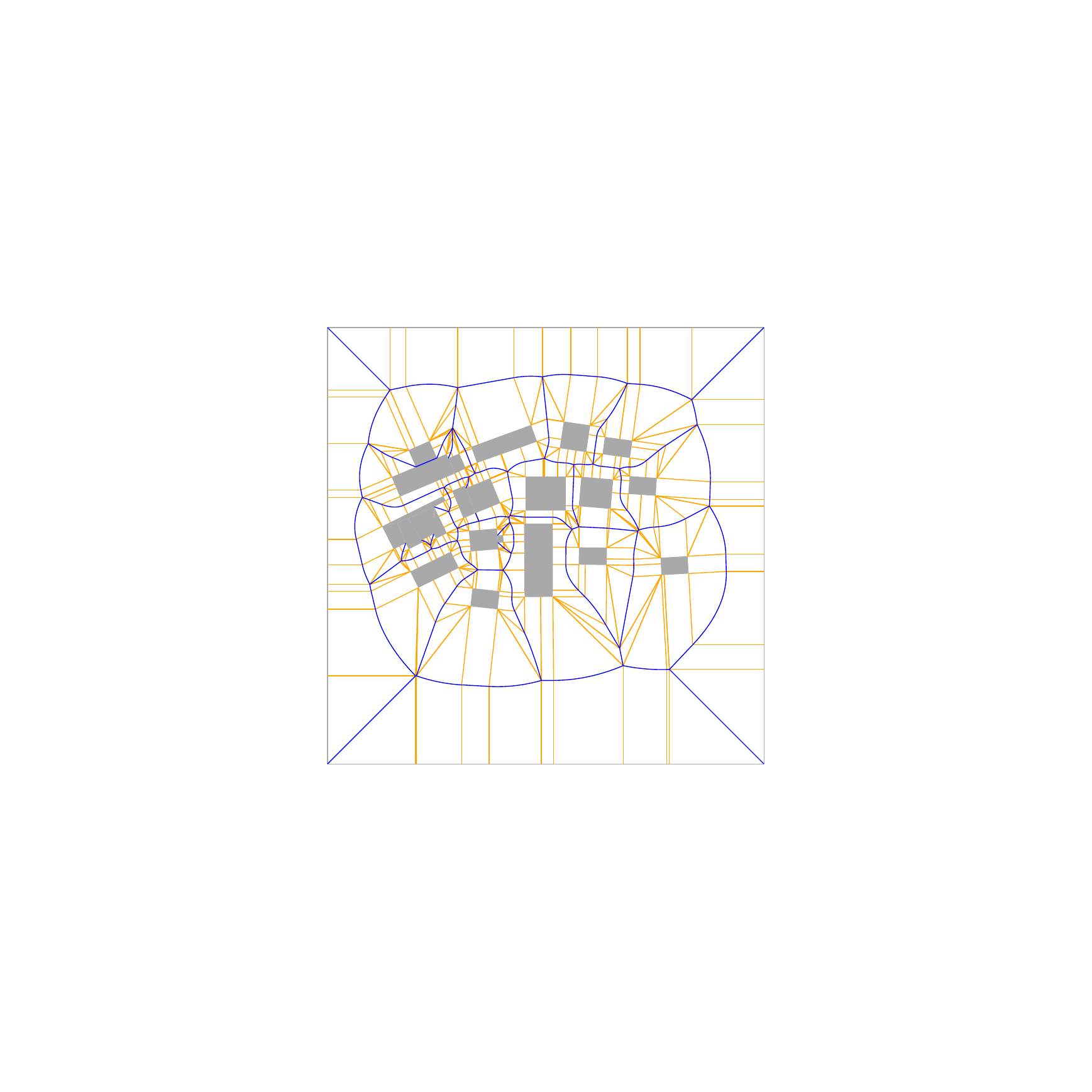}
	}%
	\subfigure[City \label{fig:exp-city}]{
		\includegraphics[width=0.23\textwidth]{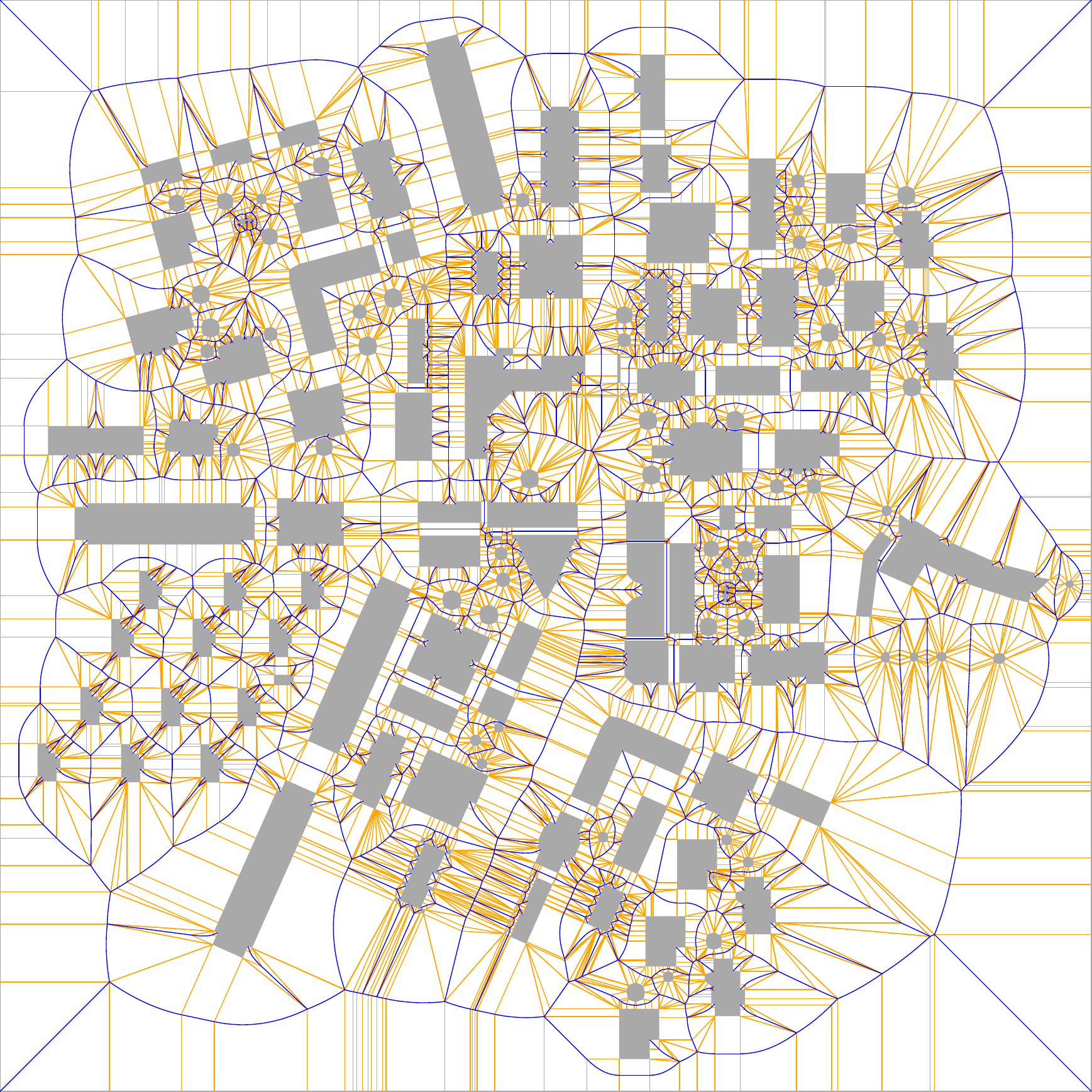}
	}%
	\subfigure[Zelda \label{fig:exp-zelda}]{
		\includegraphics[width=0.23\textwidth]{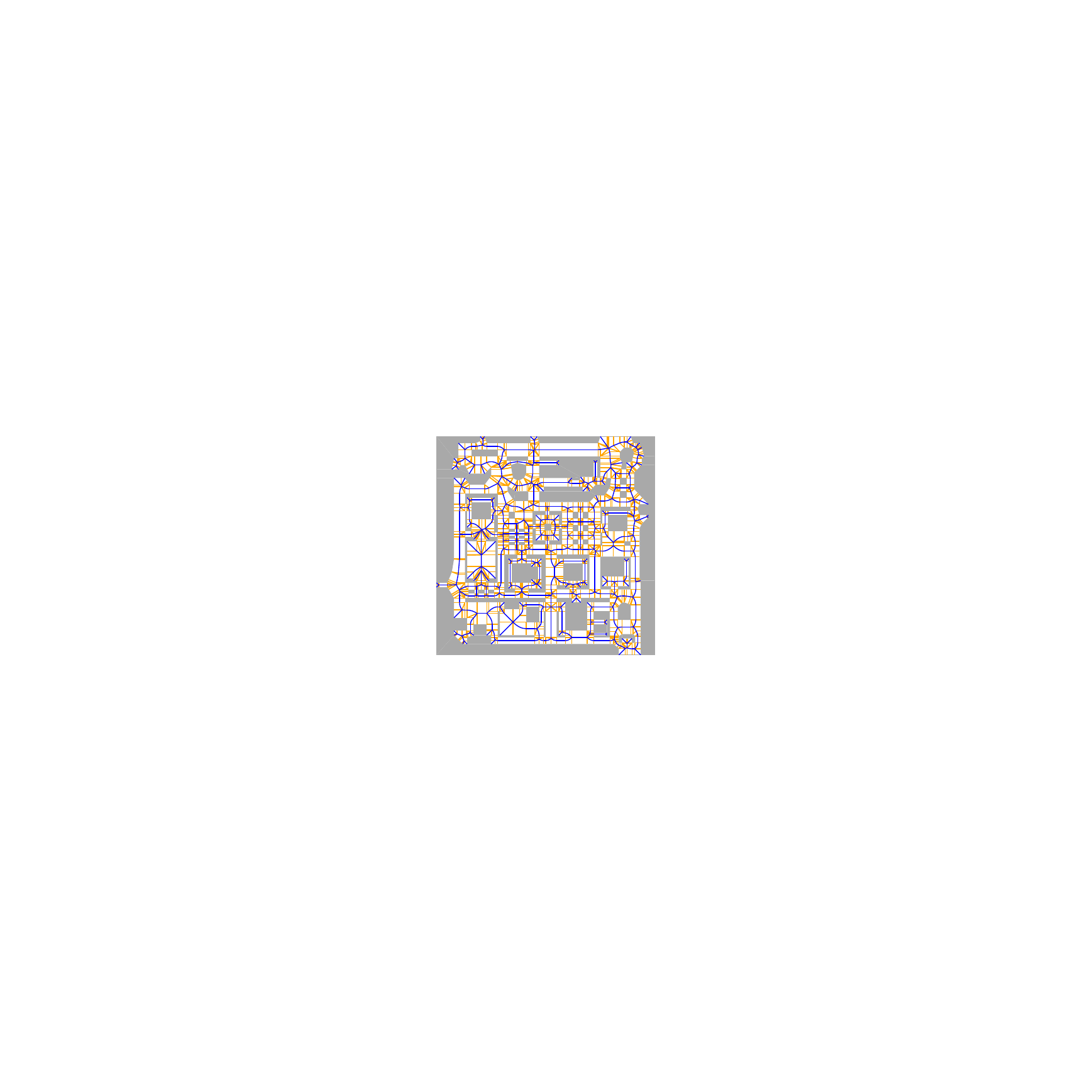}
	}%
	\subfigure[Zelda2x2 \label{fig:exp-zelda2x2}]{
		\includegraphics[width=0.23\textwidth]{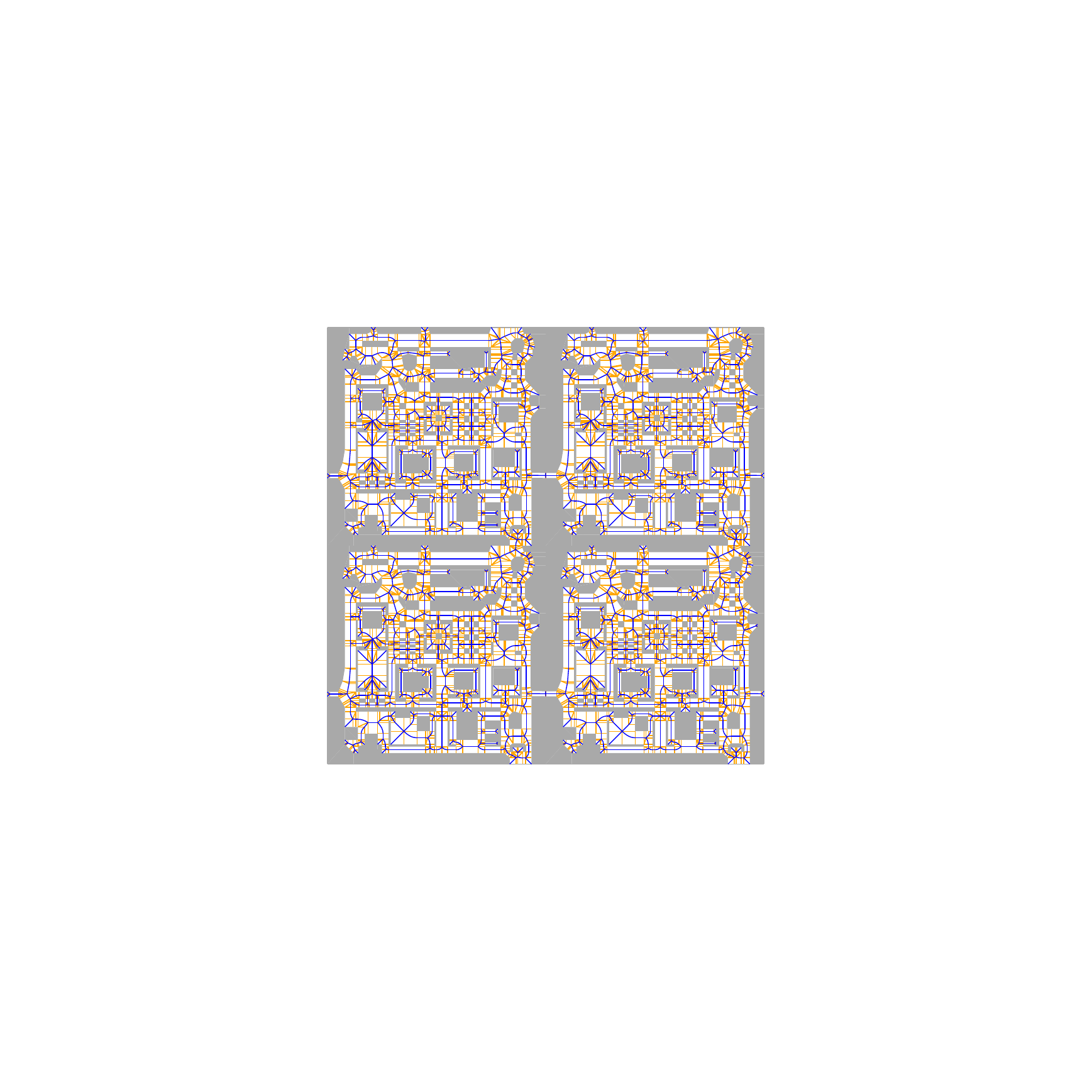}
	}
	\caption{The set of 2D environments and their ECMs.
	\emph{Zelda4x4} and \emph{Zelda8x8} are not shown because they are very large and they are structurally similar to \emph{Zelda2x2}.
	\label{fig:exp-environments-2d}}
\end{figure*}

\noindent
The multi-layered environments are shown in \cref{fig:exp-environments-mle,fig:exp-environments-mle2} and are described in the lower part of \cref{tab:exp-environments}.
\begin{itemize}
	\item \emph{Ramps} consists of three flat layers connected by four ramps. Each ramp is modelled as a separate layer for simplicity.
	\cref{fig:ramps-3d} at the beginning of this paper shows \emph{Ramps} and its ECM in 3D.
	\item \emph{Ramps2} is a version of \emph{Ramps} in which we have added 56 polygonal obstacles to the flat layers.
	\item \emph{Library} is a simplified model of the Utrecht University library.
	\item \emph{Station} is a model of a train station with one main hall and one layer containing all platforms; these two layers are connected by 32 ramps.
	Again, each ramp has been modelled as a separate layer.
	\item \emph{Tower} is a complex multi-storey apartment building.
	\item \emph{Stadium} is a model of an American football stadium with many staircases and obstacles.
	Since it has been drawn manually based on real-world data, it contains small gaps that generate disconnected graph components.
	It also features sequences of nearly-collinear points that generate medial axis edges when the input coordinates have been rounded and scaled.
	These graph elements seem redundant, but they are correct in our scaled integer coordinate system.
	\item \emph{BigCity} is a combination of the 2D city environment, six instances of \emph{Tower}, and two instances of \emph{Library}.
	The towers are highly detailed compared to the rest of the environment. 
	Voxel-based navigation mesh algorithms would require a very high resolution to capture all details.
	\item \emph{BigCity2x2} consists of four tiled instances of \emph{BigCity}. It measures 1 km$^2$ and contains 784 connections.
\end{itemize}

\begin{NewContent}
These environments were chosen to broadly reflect a range of complexities in terms of the number of layers, obstacle vertices, and connections. 
In previous work \cite{vanToll2011-MultiLayered}, we also explored theoretically challenging `toy examples', 
such as copies of an environment with an increasing number of connections.
In this paper, we choose to focus on realistic environments instead.
For more environments, including ones that have been used in research on other navigation meshes, 
we refer the reader to our recent comparative study \cite{vanToll2016-ComparativeStudy}.
\end{NewContent}

\begin{figure*}[p]
	\centering
	\subfigure[Ramps \label{fig:exp-ramps-ground}]{
		\includegraphics[width=0.23\textwidth]{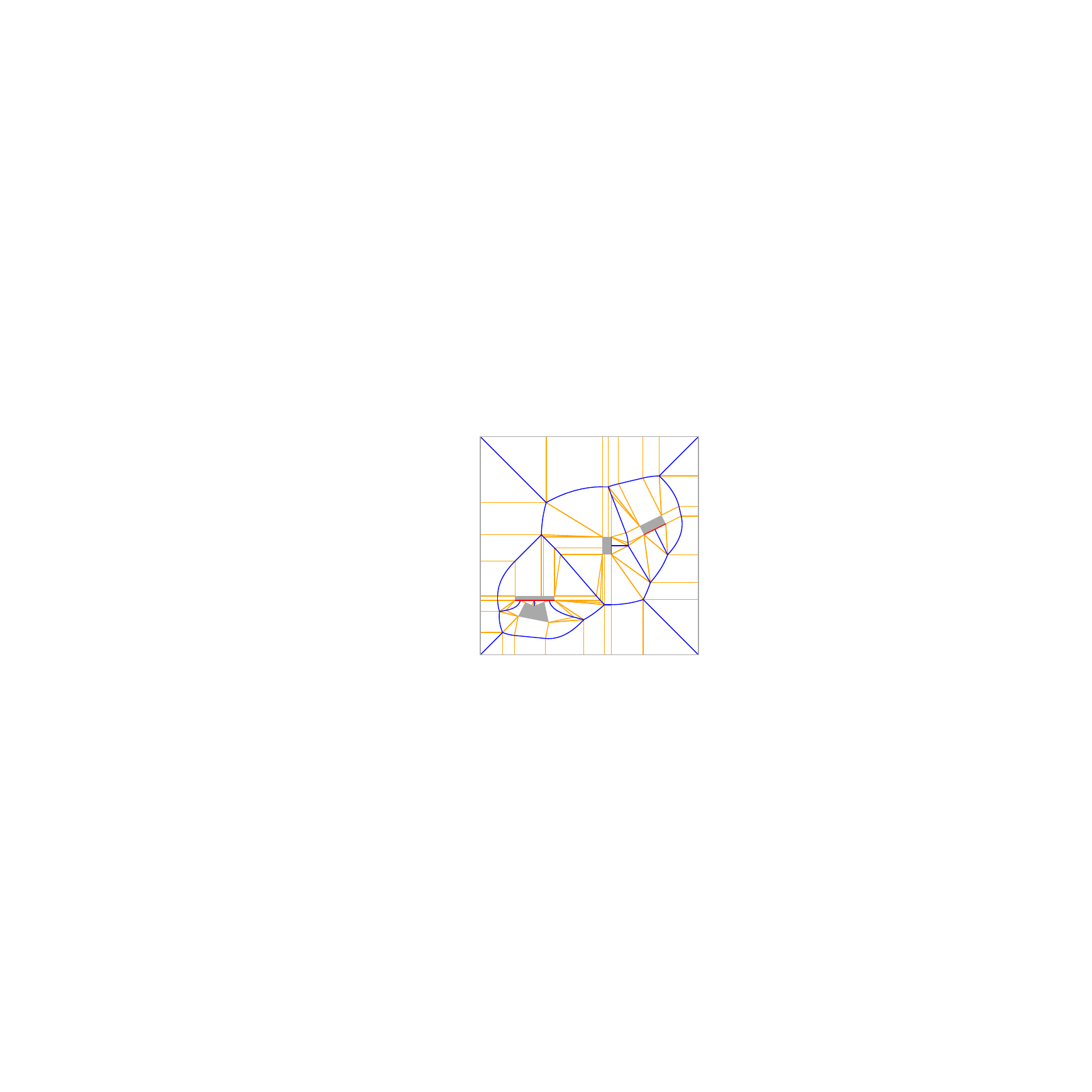}
	}
	\subfigure[Ramps (continued) \label{fig:exp-ramps-other}]{
		\includegraphics[width=0.23\textwidth]{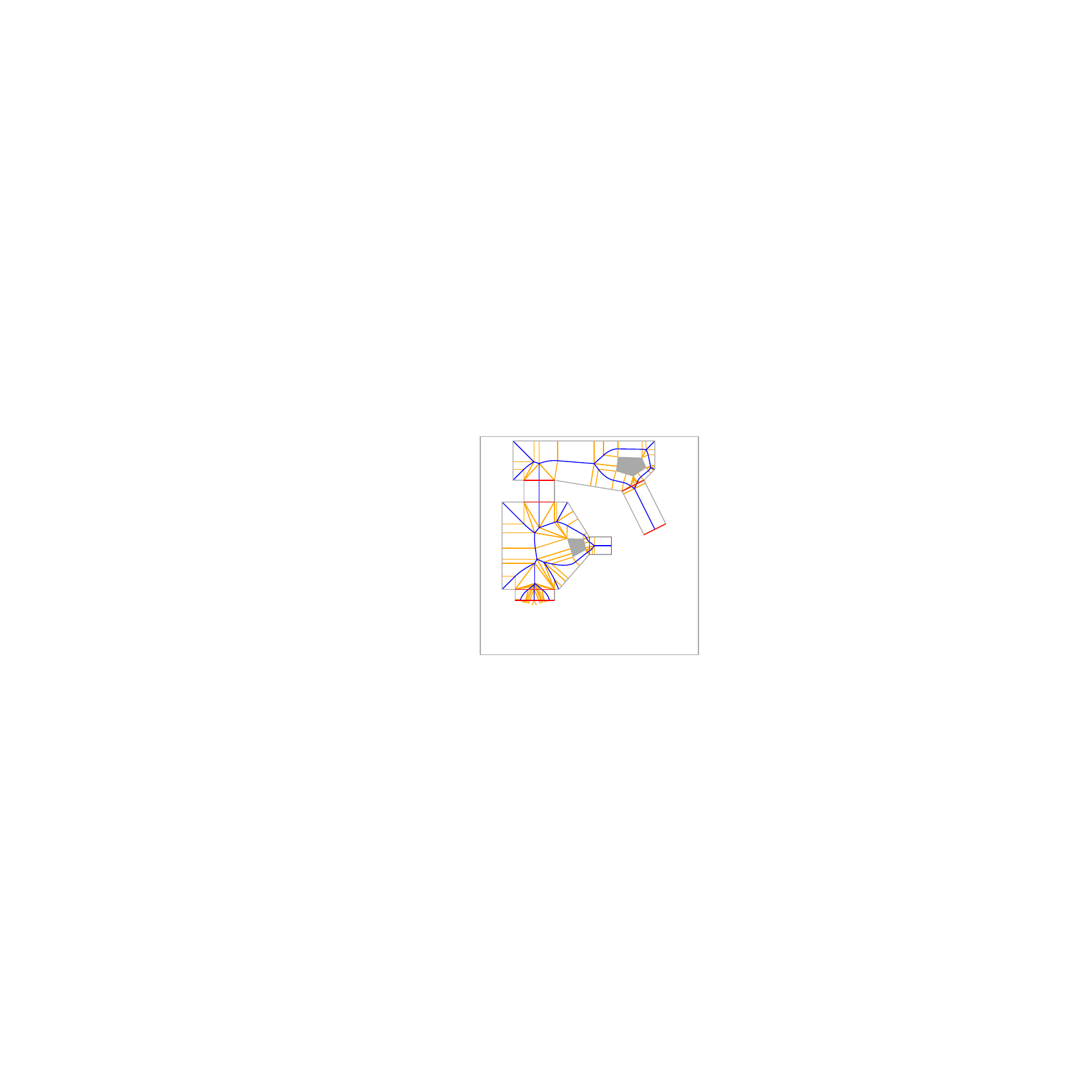}
	}
	\subfigure[Ramps2 \label{fig:exp-ramps2-ground}]{
		\includegraphics[width=0.23\textwidth]{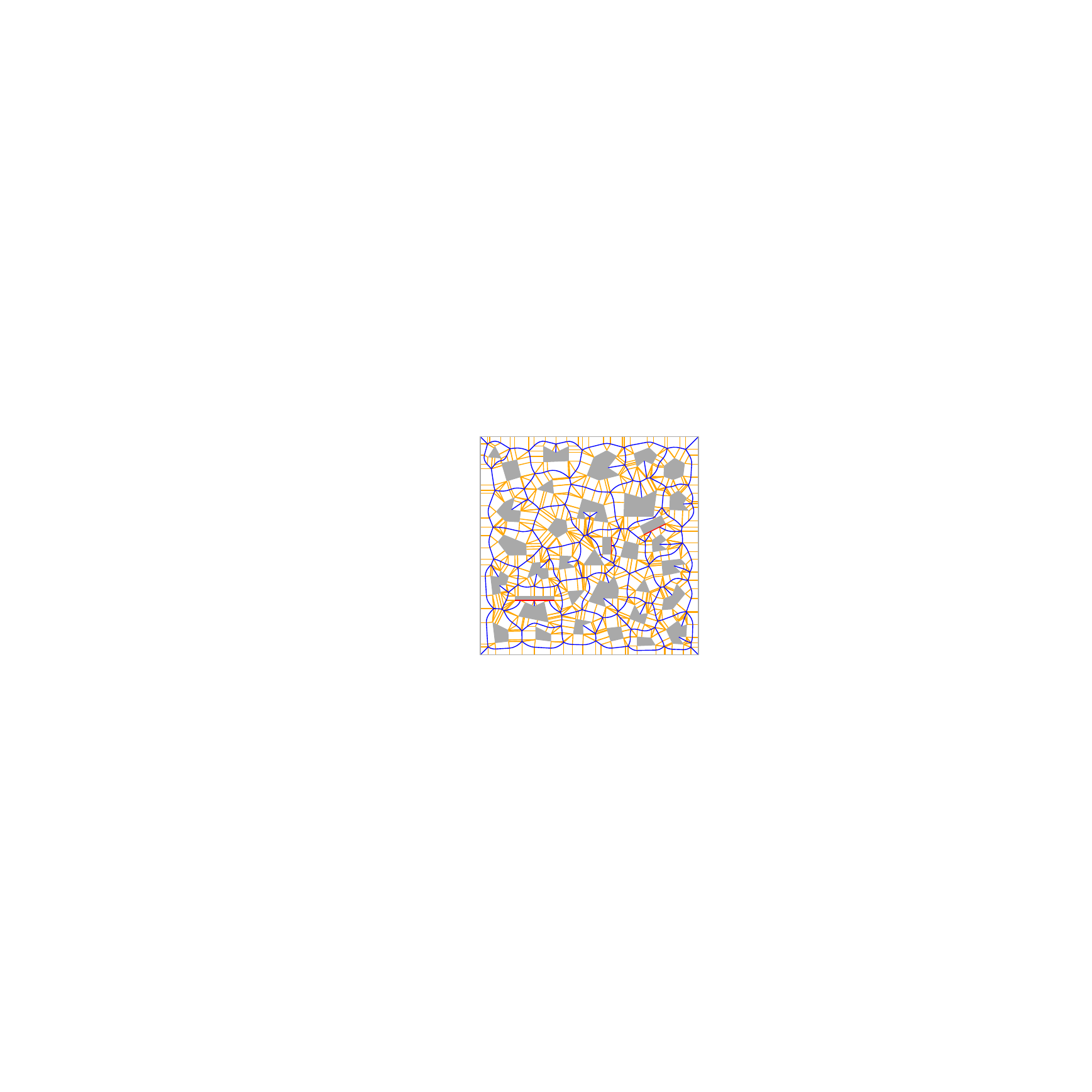}
	}
	\subfigure[Ramps2 (continued) \label{fig:exp-ramps2-other}]{
		\includegraphics[width=0.23\textwidth]{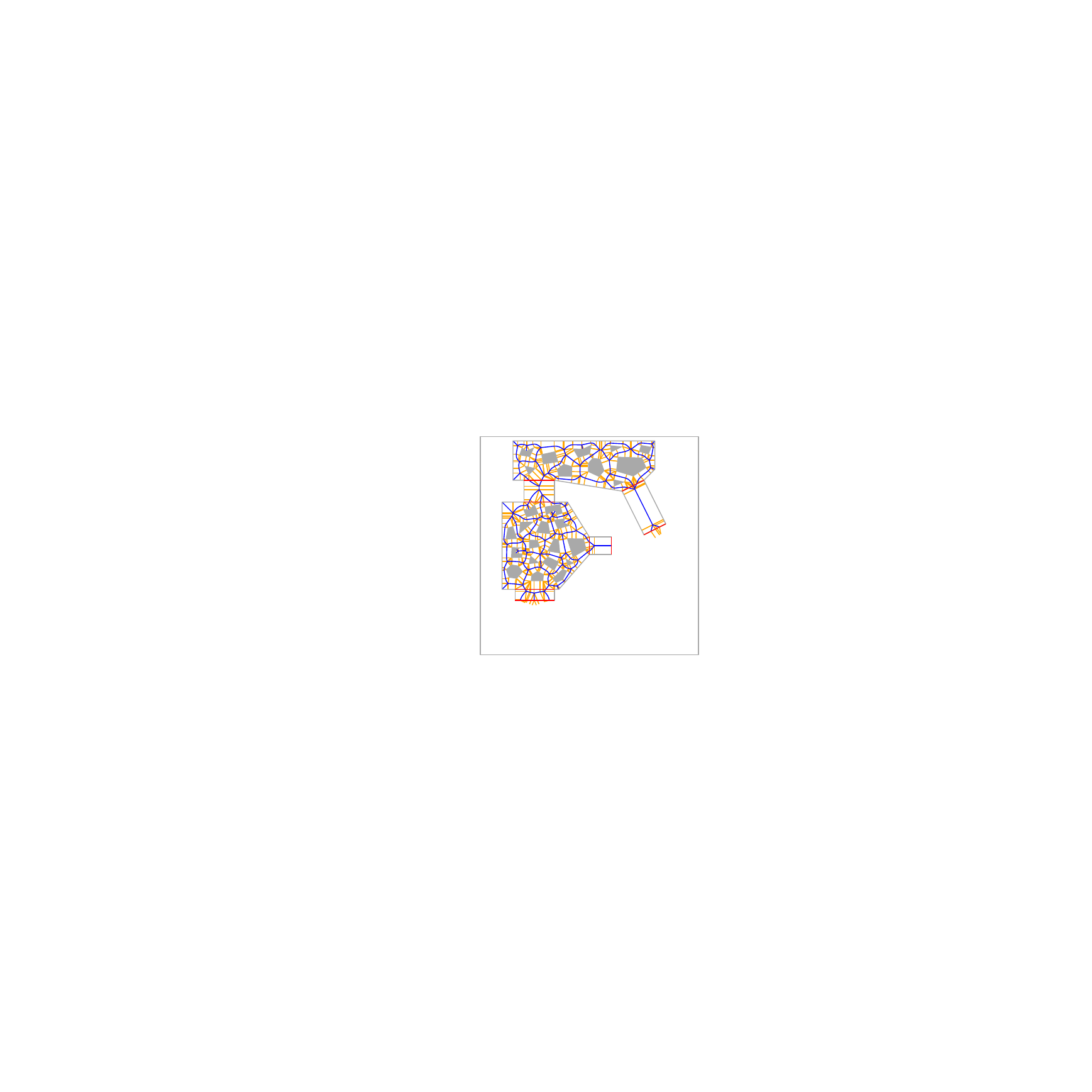}
	}
	\subfigure[Station \label{fig:exp-station-hall}]{
		\includegraphics[width=0.48\textwidth]{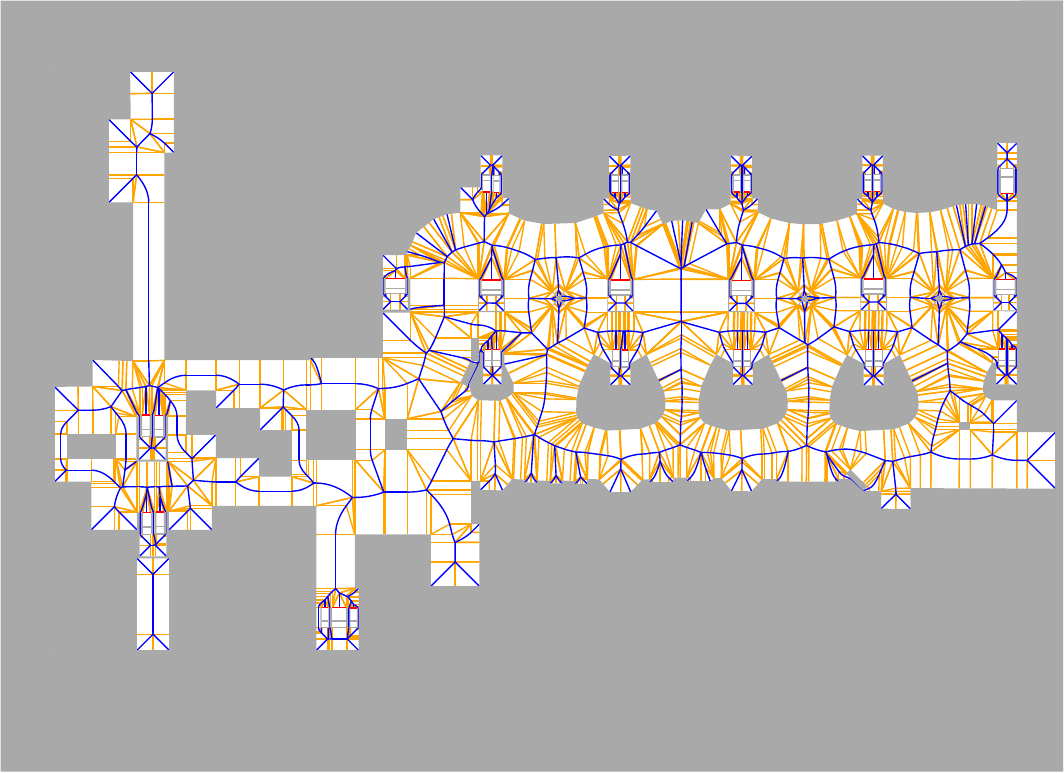}
	}
	\subfigure[Station (continued) \label{fig:exp-station-platforms}]{
		\includegraphics[width=0.48\textwidth]{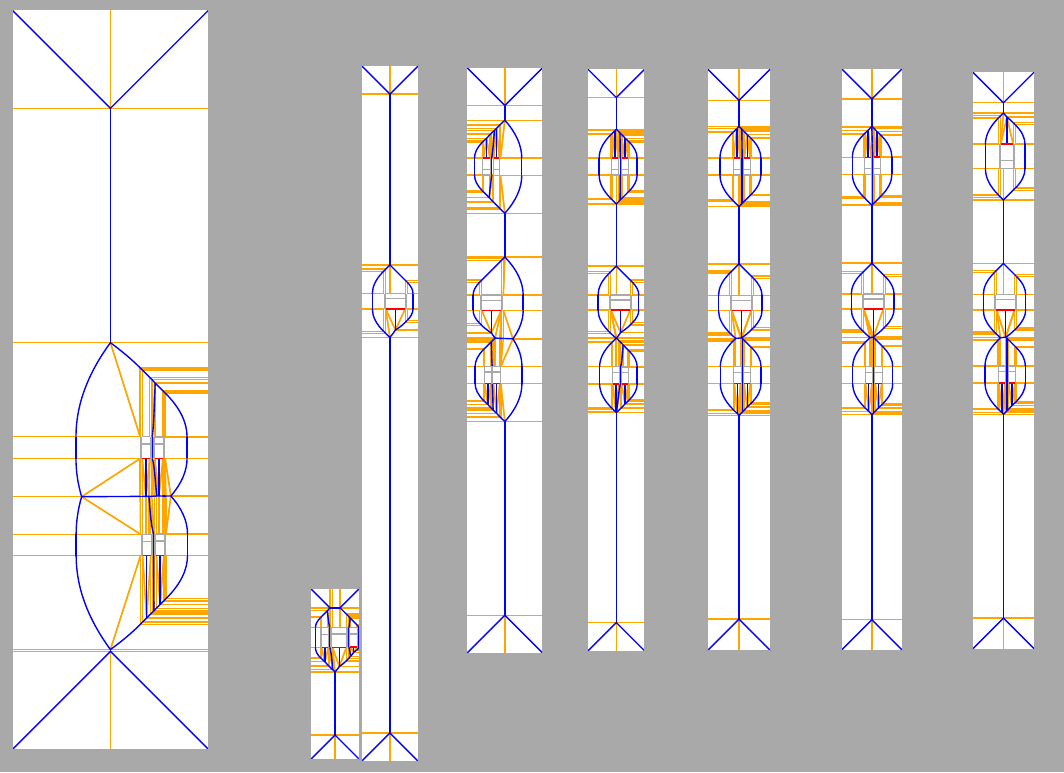}
	}
	\subfigure[Stadium \label{fig:exp-stadium-L0-etc}]{
		\includegraphics[width=0.48\textwidth]{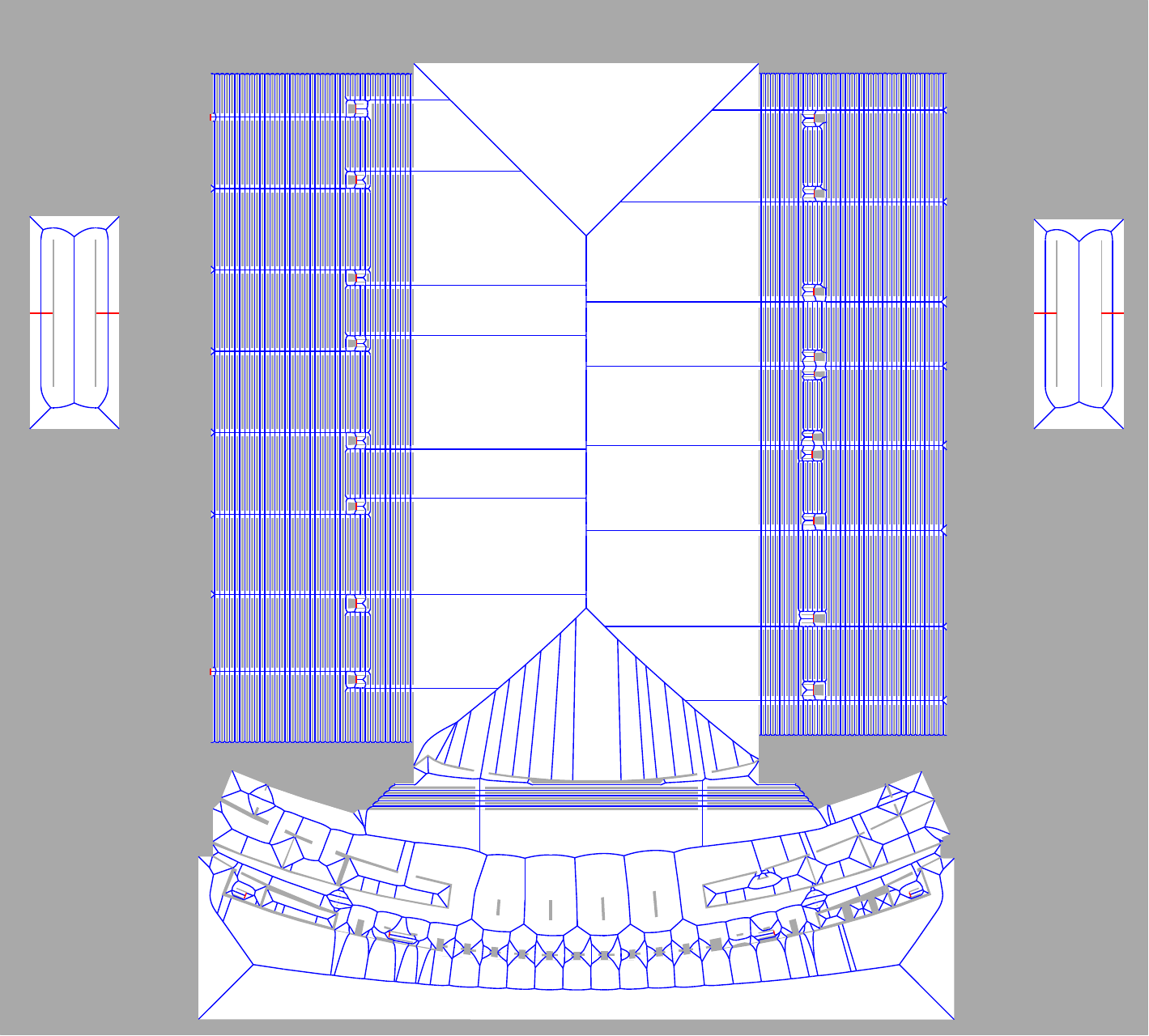}
	}
	\subfigure[Stadium (continued) \label{fig:exp-stadium-L2-etc}]{
		\includegraphics[width=0.48\textwidth]{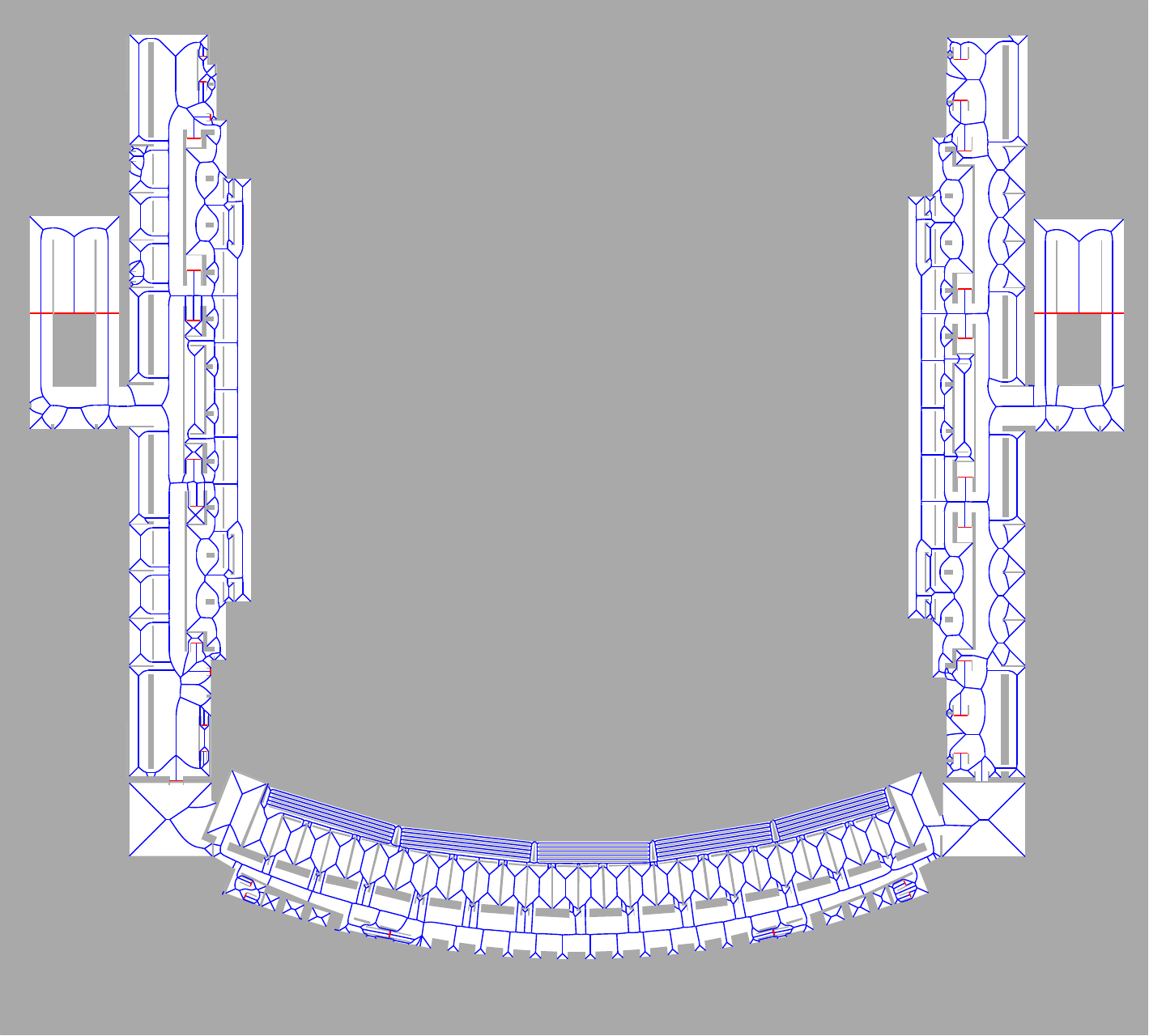}
	}
	\caption{2D views of some of the multi-layered environments used in our experiments.
	For \emph{Station} and \emph{Stadium}, not all layers are shown.
	For \emph{Stadium}, nearest-obstacle annotations have been omitted for clarity.
	Some inaccuracies in the geometry of \emph{Stadium} are too small to see in this image.
	\label{fig:exp-environments-mle}}
\end{figure*}

\begin{figure*}[p]
	\centering
	\subfigure[Library \label{fig:exp-library3d}]{
		\includegraphics[height=38mm]{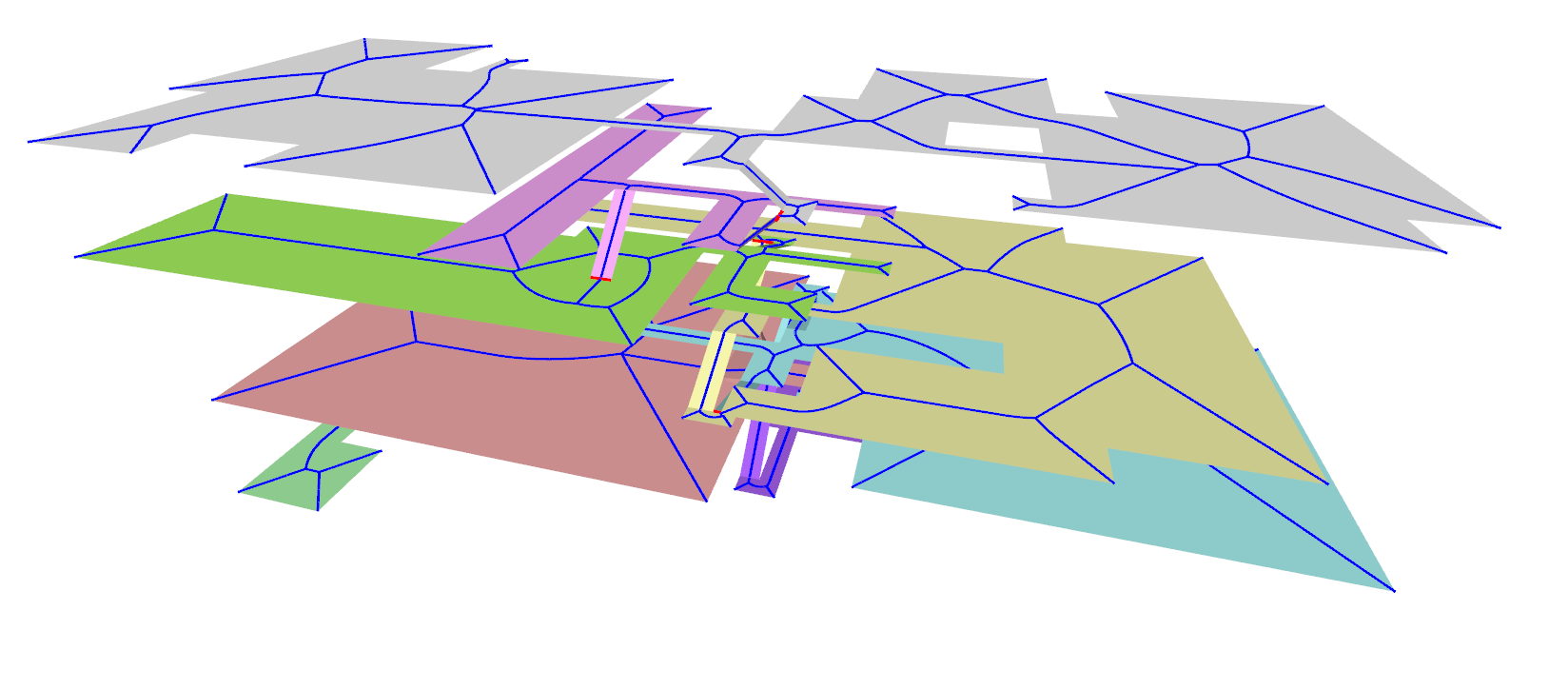}
	}%
	\subfigure[Tower \label{fig:exp-max3d}]{
		\includegraphics[height=44mm]{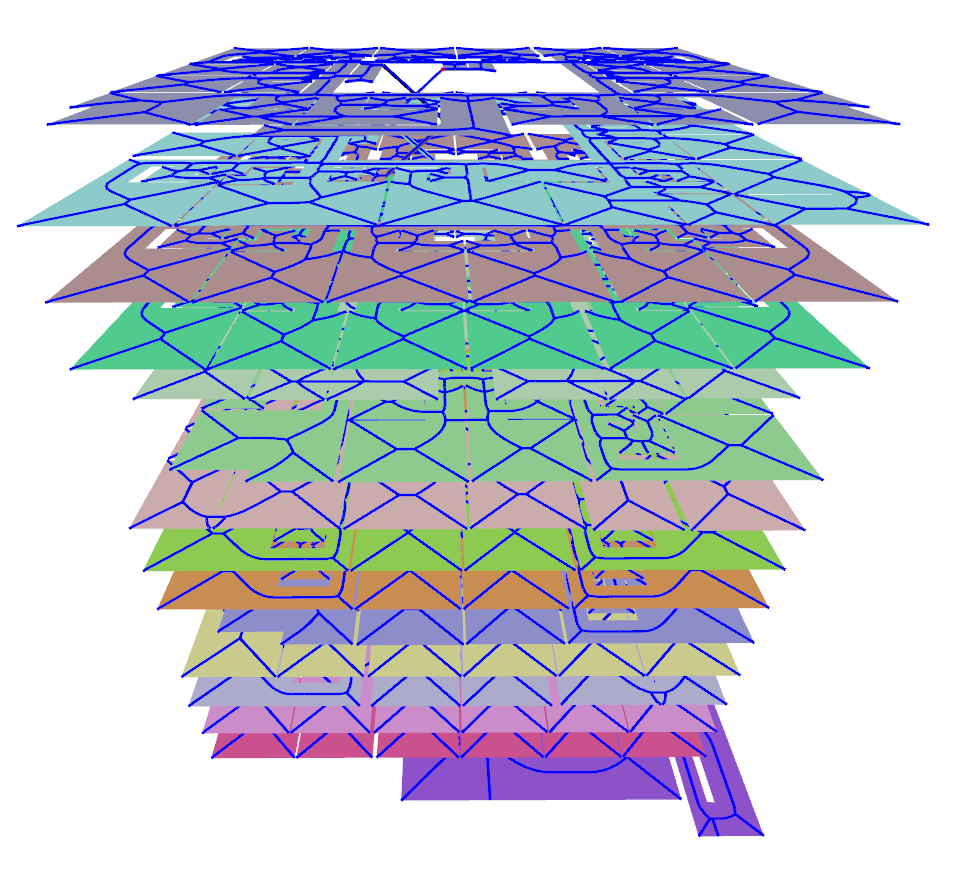}
	}\\
	\subfigure[BigCity \label{fig:exp-bigcity3d}]{
		\includegraphics[width=\textwidth]{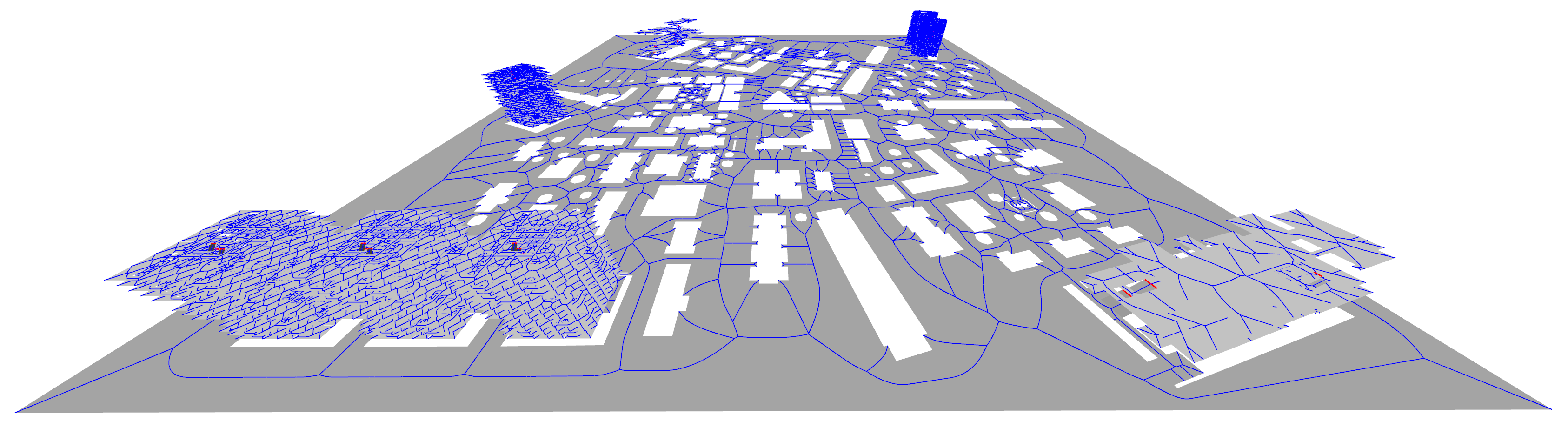}
	}
	\caption{3D views of the \emph{Library}, \emph{Tower}, and \emph{BigCity} environments and their medial axes. 
	For clarity, the nearest-obstacle annotations of the ECM have been omitted.
	We have now colored the free space rather than the obstacle space.
	The gray obstacles from the 2D \emph{City} environment are now modelled as holes in the free space of \emph{BigCity}.
	\emph{BigCity2x2} is not shown because it is very large and structurally similar to \emph{BigCity}.
	\label{fig:exp-environments-mle2}}
\end{figure*}

\begin{table*}[p]
\centering
\small{
\begin{tabular}{l|rr|rr}
\FL 
\textbf{Environment} & \multicolumn{2}{c|}{\textbf{Geometry}} & \multicolumn{2}{c}{\textbf{Multi-layered}} \NN
& \#Obstacle vertices & Size (m) & \#Layers & \#Connections
\ML
\textbf{Military}	& 108   & $200 \times 200$ & 1  & 0  \NN
\textbf{City}		& 2102  & $500 \times 500$ & 1  & 0  \NN
\textbf{Zelda}		& 564   & $100 \times 100$ & 1  & 0  \NN
\textbf{Zelda2x2}	& 2304  & $200 \times 200$ & 1  & 0  \NN
\textbf{Zelda4x4}	& 9180  & $400 \times 400$ & 1  & 0  \NN
\textbf{Zelda8x8}	& 36684 & $800 \times 800$ & 1  & 0 
\ML
\textbf{Ramps}		& 147   & $100 \times 100$ & 7  & 8  \NN
\textbf{Ramps2}		& 422   & $100 \times 100$ & 7  & 8  \NN
\textbf{Library}	& 717   & $60  \times 24$  & 9  & 8  \NN
\textbf{Station}	& 2242  & $153 \times 111$ & 34 & 64 \NN
\textbf{Tower}		& 6058  & $35  \times 35$  & 17 & 30 \NN
\textbf{Stadium}	& 12915 & $280 \times 184$ & 18 & 82 \NN
\textbf{BigCity}	& 49476 & $500 \times 500$ & 113 & 196 \NN
\textbf{BigCity2x2}	& 197884 & $1000 \times 1000$ & 449 & 784 \NN
\FL 
\end{tabular}
}
\caption{Details of the environments used in our experiments.
The \emph{Geometry} columns show the number of obstacle vertices and the physical width and height of the environment (in meters). 
The \emph{Multi-layered} columns show the number of layers and connections of each environment.
}
\label{tab:exp-environments}
\end{table*}

\subsection{Computing the ECM} \label{s:experiments:ecm}

\NewContentInline{For the \emph{2D environments}, the ECM complexities and construction times are shown in the upper part of \cref{tab:exp-ecms}.}
Vroni and Boost handle degenerate cases such as degree-4 vertices differently, which leads to slightly different ECM complexities for both implementations.
The complexities in \cref{tab:exp-ecms} were taken from the Boost version.

The Vroni-based implementation was faster than the Boost-based implementation in all environments.
For the most complex 2D environment, \emph{Zelda8x8}, computing the ECM took just under 1 second when using Vroni.
Hence, even complex ECMs can be computed quickly. 
This allows the navigation mesh to be generated interactively (e.g.\ when loading a game level, or when used in a tool for designing environments).
\bigskip

\noindent
\NewContentInline{For the \emph{multi-layered environments}, the ECM complexities and construction times are shown in the lower part of \cref{tab:exp-ecms}.}
While it can be seen that a multi-layered ECM takes more time to compute than a 2D ECM of the same complexity, 
the construction time is still well under a second for all environments except the two \emph{BigCity} variants.
For the largest environment, \emph{BigCity2x2}, the construction takes about 9 seconds when using Vroni.

The Boost implementation for Voronoi diagrams is thread-safe, so we can use \emph{multi-threading} to compute the initial ECMs of all layers in parallel.
The running times for this multi-threaded Boost version are also shown in \cref{tab:exp-ecms}. 
We used OpenMP with 5 parallel threads and dynamic scheduling. 
This version performed particularly well in environments with many complex layers; 
in particular, the ECM of \emph{BigCity2x2} was computed in approximately $4.2$ seconds.
Standard deviations among running times were higher because the threads were scheduled in an unpredictable way.
Still, this implementation shows that multi-threading is a promising addition.

\begin{table*}[h]
\centering
\small{
\begin{tabular}{l|rrr|rrr}
\FL 
\textbf{Environment} & \multicolumn{3}{c|}{\textbf{ECM complexity}} & \multicolumn{3}{c}{\textbf{ECM time (ms)}} \NN
& \#Vertices & \#Edges & \#Bending pts & Vroni & Boost & Boost (MT) 
\ML
\textbf{Military}	& 56    & 71    & 288   & 3.5  [0.1]    & 6.7 [0.1]     & -- \NN
\textbf{City}		& 1442  & 1621  & 6306  & 70.9 [0.3]    & 130.0 [0.7]   & -- \NN
\textbf{Zelda}		& 296   & 351   & 1258  & 14.9 [0.1]    & 23.0 [0.2]    & -- \NN
\textbf{Zelda2x2}	& 1184  & 1408  & 5082  & 59.2 [0.5]    & 92.0 [0.6]    & -- \NN
\textbf{Zelda4x4}	& 4720  & 5624  & 20329 & 235.4 [2.4]   & 367.7 [1.6]   & -- \NN
\textbf{Zelda8x8}	& 18848 & 22480 & 81365 & 997.4 [5.4]   & 1529.1 [3.8]  & --
\ML
\textbf{Ramps}		& 54    & 61    & 181   & 5.8 [0.2]     & 9.4 [0.3]     & 7.6 [0.1] \NN
\textbf{Ramps2}		& 228   & 290   & 1118  & 16.3 [0.4]    & 29.5 [0.6]    & 21.2 [0.1] \NN
\textbf{Library}	& 219   & 222   & 599   & 14.7 [0.3]    & 20.2 [0.4]    & 9.2 [0.1] \NN
\textbf{Station}	& 660   & 768   & 2804  & 68.6 [0.3]    & 97.3 [0.5]    & 74.3 [0.3] \NN
\textbf{Tower}		& 4948  & 4979  & 14407 & 248.8 [2.2]   & 383.4 [1.3]   & 110.1 [3.4] \NN
\textbf{Stadium}	& 6303  & 7754  & 26323 & 442.4 [8.2]   & 572.2 [1.6]   & 263.5 [4.7] \NN
\textbf{BigCity}	& 32264 & 32652 & 104002 & 2168.6 [19.6] & 3430.0 [10.9]  & 925.7 [29.6] \NN
\textbf{BigCity2x2}	& 129147 & 130702 & 416411 & 8972.5 [32.7] & 14287.0 [41.7] & 4219.3 [91.9] \NN
\FL 
\end{tabular}
}
\caption{Details of the ECMs for our experiments.
The \emph{ECM complexity} columns show the number of vertices, edges, and bending points in the ECM computed using Boost.
The \emph{ECM time} columns show the ECM construction time for the three implementations: Vroni, Boost, and Boost with 5 parallel threads (for multi-layered environments). 
All times are in milliseconds and have been averaged over 10 runs. Standard deviations are shown between square brackets.
}
\label{tab:exp-ecms}
\end{table*}


\subsection{Path Planning} \label{s:experiments:pathplanning}

In each environment, we have computed indicative routes between $10{,}000$ pairs of random start and goal points.
For each point, we first chose a random layer (if applicable) such that the probability of a layer being chosen was proportional to its surface area.
The point itself was then chosen by uniformly sampling in the layer's bounding box until an obstacle-free point was found.

For each query pair $(s,g)$, we computed the shortest path between \Retraction{s} and \Retraction{g} on the medial axis using A* search, 
with the 2D Euclidean distance to \Retraction{g} as a heuristic.
We then converted this path to a short indicative route with a preferred distance of $0.5$ m to obstacles.

\cref{tab:exp-planning} show the average running times of the complete path planning query per environment.
The running time depends heavily on the complexity of the resulting path; this explains the high standard deviations.
It can be seen that queries require only a few milliseconds on average in the most complex environments.
Thus, the ECM allows real-time path planning for large crowds of characters with individual goals.
In very complex environments, grid-based planning would be much slower because a high grid resolution would be required to capture all details.

\begin{table*}[h]
\centering
\small{
\begin{tabular}{l|r}
\FL 
\textbf{Environment} & \textbf{Planning time (ms)}
\ML
\textbf{Military}	& 0.21 [0.15] \NN
\textbf{City}		& 1.12 [0.70] \NN
\textbf{Zelda}		& 0.41 [0.21] \NN
\textbf{Zelda2x2}	& 0.96 [0.47] \NN
\textbf{Zelda4x4}	& 1.99 [1.01] \NN
\textbf{Zelda8x8}	& 4.65 [2.75]
\ML
\textbf{Ramps}		& 0.13 [0.08] \NN
\textbf{Ramps2}		& 0.37 [0.18] \NN
\textbf{Library}	& 0.53 [0.32] \NN
\textbf{Station}	& 0.79 [0.48] \NN
\textbf{Tower}		& 1.58 [0.68] \NN
\textbf{Stadium}	& 2.22 [1.43] \NN
\textbf{BigCity}	& 3.03 [2.55] \NN
\textbf{BigCity2x2}	& 8.44 [7.59] \NN
\FL 
\end{tabular}
}
\caption{Results of the path planning experiments.
The \emph{Planning time} column shows the running time to compute a path, averaged over $10{,}000$ random queries.
All times are in milliseconds. Standard deviations are shown between square brackets.
}
\label{tab:exp-planning}
\end{table*}

\noindent
We refer the reader to previous work \cite{vanToll2015-Framework} for experiments on crowd simulation.
This previous publication has shown that the running time of a simulation \emph{without} collision avoidance scales linearly with the number of characters.
Collision avoidance is inherently the most expensive step because characters need to find their neighbors and respond to their movement.
By using 4 CPU cores and 8 parallel threads, we can currently simulate around $15{,}000$ characters in real-time with collision avoidance.
The ECM framework has a small memory footprint that allows simulations of at least 1 million characters.
It has been successfully used to simulate crowded real-life events 
such as the 2015 \emph{Tour de France} opening in Utrecht (The Netherlands) and evacuations of future metro stations in Amsterdam.

%% file: section-conclusions.tex
\section{Conclusions and Future Work} \label{s:conclusions}

In robotics, simulations, and gaming applications, virtual walking characters often need to compute and traverse paths through an environment.
A navigation mesh is a subdivision of the \NewContentInline{free space} into polygons that allows real-time path planning for crowds of characters.

\NewContentInline{With this application area in mind, we have studied the medial axis for 3D environments with a consistent direction of gravity.
A walkable environment (WE) describes where characters can walk.}
A multi-layered environment (MLE) is a WE that has been subdivided into 2D layers connected by $k$ \NewContentInline{connections with certain geometric properties.} 
We have defined \NewContentInline{the medial axis} for WEs and MLEs based on projected distances on the ground plane.
We have presented an algorithm that computes this medial axis by initially treating all connections as closed obstacles and then opening them incrementally.
Compared to previous work \cite{vanToll2011-MultiLayered}, we have presented a new and correct algorithm for opening a connection, 
and we have improved the overall running time to \BigO{n \log n \log k} by opening the connections in an efficient order. 

\NewContentInline{We have presented an improved definition of the Explicit Corridor Map (ECM), which is a} navigation mesh based on the medial axis.
The ECM enables path planning for disk-shaped characters of any radius.
It supports efficient geometric operations such as retractions, nearest-obstacle queries, dynamic updates, and the computation of short paths with preferred clearance to obstacles.

Our implementation computes the ECM efficiently for large 2D and multi-layered environments. 
The ECM supports important applications and operations, such as path planning and dynamic updates. 
It is a useful basis for simulating large crowds of heterogeneous characters in real-time \cite{vanToll2015-Framework}.
\bigskip

\noindent 
We have given an \BigO{m \log m}-time algorithm for opening a connection bounded by $m$ medial axis arcs. 
This algorithm is not necessarily optimal. 
For 2D Voronoi diagrams, a line segment site can be removed in \BigO{m} time \cite{Khramtcova2015-LinearDeletion}. 
While this linear-time algorithm cannot immediately be applied to our multi-layered domain, 
it suggests that a \BigO{m}-time solution for opening a connection might exist.
Finding such an algorithm is a challenge for future work. 
It would allow us to improve the overall construction time for the multi-layered ECM to an optimal \BigO{n \log n}.

A drawback of navigation meshes in general is that the shortest path in the dual graph (or, in our case, the medial axis) 
is not necessarily homotopic to the shortest path in the entire environment. 
Therefore, even a shortened path as described in \cref{s:ecm:indicativeroute} can be longer than the overall shortest path.
For future work, we want to analyze path lengths in the ECM and investigate how to improve them. 
One option is to combine the ECM with a visibility graph, which yields shortest paths but has a size of \BigO{n^2} \cite{Wein2007-VisibilityVoronoi,Ghosh2007-Visibility}. 

\begin{NewContent}
While it is common to use projected distances for navigation meshes, 
we would like to investigate how navigation meshes can be extended to encode information about \emph{height differences} as well. 
As explained in \cref{s:relatedwork:navmeshes-3d}, one could simply use the WE itself as a data structure for path planning, 
but this approach lacks some of the advantages of the ECM (such as the ability to compute paths for disks of any radius).

Currently, state-of-the art 3D navigation meshes convert a raw 3D environment to a WE by using voxel-based filtering algorithms.
Our comparative study of navigation meshes \cite{vanToll2016-ComparativeStudy} has shown that this approach is not ideal: 
it requires parameter tuning, it is not scalable to large environments, and it sometimes sacrifices precision. 
We are therefore interested in developing \emph{exact} algorithms for obtaining walkable and multi-layered environments from arbitrary 3D geometry.
Note that an exact algorithm will recognize very small details in the environment that may not be relevant for the application, especially when using imprecise real-world data.

If a walkable environment is given as a set of triangles, converting it to an MLE with a minimum number of connections is NP-hard in the number of triangles \cite{Hillebrand2016-MLE}. 
However, there may be other criteria by which an MLE can be considered `optimal', such as the total length of all connections combined. 
Finding an optimal MLE (or a constant-factor approximation of it) based on such criteria is an interesting topic for future work.
\end{NewContent}

We are also interested in lifting other 2D data structures \NewContentInline{to the multi-layered domain, 
including visibility graphs and related concepts \cite{Pocchiola1993-VisibilityComplex,Wein2007-VisibilityVoronoi}}.
We believe that multi-layered environments give rise to an interesting new class of problems for future research.

%% file: MedialAxis-arXiv.bbl
\begin{thebibliography}{10}

\bibitem{Aggarwal1989-LinearVoronoi}
A.~Aggarwal, L.J. Guibas, J.~Saxe, and P.W. Shor.
\newblock A linear-time algorithm for computing the {V}oronoi diagram of a
  convex polygon.
\newblock {\em Discrete \& Computational Geometry}, 4:591--604, 1989.

\bibitem{Aurenhammer2013-Voronoi}
F.~Aurenhammer, R.~Klein, and D.T. Lee.
\newblock {\em Voronoi Diagrams and Delaunay Triangulations}.
\newblock World Scientific Publishing Company, 2013.

\bibitem{deBerg2008-CompGeom}
M.~{\noop{berg}}de Berg, O.~Cheong, M.~van Kreveld, and M.~Overmars.
\newblock {\em Computational Geometry: Algorithms and Applications}.
\newblock Springer, 3rd edition, 2008.

\bibitem{vandenBerg2011-ORCA}
J.P. {\noop{berg}}van~den Berg, S.J Guy, M.C. Lin, and D.~Manocha.
\newblock Reciprocal n-body collision avoidance.
\newblock In {\em Proceedings of the 14th International Symposium on Robotics
  Research}, pages 3--19, 2011.

\bibitem{Berseth2015-NavMesh3D}
G.~Berseth, M.~Kapadia, and P.~Faloutsos.
\newblock {ACCLMesh}: Curvature-based navigation mesh generation.
\newblock In {\em Proceedings of the 8th ACM SIGGRAPH International Conference
  on Motion in Games}, pages 97--102, 2015.

\bibitem{Blum1967-MedialAxis}
H.~Blum.
\newblock A transformation for extracting new descriptors of shape.
\newblock {\em Models for the Perception of Speech and Visual Form}, pages
  362--380, 1967.

\bibitem{Boost}
Boost.
\newblock The {B}oost \textsc{C++} library.
\newblock \url{http://www.boost.org/}, 2016.

\bibitem{Chin1999-MedialAxis}
F.~Chin, J.~Snoeyink, and C.A. Wang.
\newblock Finding the medial axis of a simple polygon in linear time.
\newblock {\em Discrete \& Computational Geometry}, 21:405--420, 1999.

\bibitem{Deusdado2008-MultiLayered}
L.~Deusdado, A.R. Fernandes, and O.~Belo.
\newblock Path planning for complex 3{D} multilevel environments.
\newblock {\em Proceedings of the 24th Spring Conference on Computer Graphics},
  pages 187--194, 2008.

\bibitem{Devillers1999-VoronoiDeletion}
O.~Devillers.
\newblock On deletion in {D}elaunay triangulations.
\newblock In {\em Proc. 15th Annual ECM Symposium on Computational Geometry},
  pages 181--188, 1999.

\bibitem{Fortune1987-Voronoi}
S.~Fortune.
\newblock A sweepline algorithm for {V}oronoi diagrams.
\newblock {\em Algorithmica}, 2:153--174, 1987.

\bibitem{Garber2004-VoronoiDiagramMP}
M.~Garber and M.C. Lin.
\newblock {\em Constraint-Based Motion Planning Using Voronoi Diagrams}, pages
  541--558.
\newblock Springer, 2004.

\bibitem{Garcia2014-GridPlanning}
F.M. Garc\'ia, M.~Kapadia, and N.M. Badler.
\newblock {GPU}-based dynamic search on adaptive resolution grids.
\newblock In {\em Proceedings of the 31st IEEE International Conference on
  Robotics and Automation}, pages 1631--1638, 2014.

\bibitem{Geraerts2010-ECM}
R.~Geraerts.
\newblock Planning short paths with clearance using {E}xplicit {C}orridors.
\newblock In {\em Proceedings of the 27th IEEE International Conference on
  Robotics and Automation}, pages 1997--2004, 2010.

\bibitem{Ghosh2007-Visibility}
S.K. Ghosh.
\newblock {\em Visibility Algorithms in the Plane}.
\newblock Cambridge University Press, 2007.

\bibitem{Green1978-IncrementalVoronoi}
P.J. Green and R.~Sibson.
\newblock Computing {D}irichlet tessellations in the plane.
\newblock {\em The Computer Journal}, 21(2):168--173, 1978.

\bibitem{Hart1968-AStar}
P.~Hart, N.~Nilsson, and B.~Raphael.
\newblock A formal basis for the heuristic determination of minimum cost paths.
\newblock {\em IEEE Transactions on Systems Science and Cybernetics},
  4(2):100--107, 1968.

\bibitem{Helbing1995-SocialForces}
D.~Helbing and P.~Moln\'ar.
\newblock Social force model for pedestrian dynamics.
\newblock {\em Physical Review E}, 51(5):4282--4286, 1995.

\bibitem{Held2011-Vroni}
M.~Held.
\newblock {VRONI} and {A}rc{VRONI}: {S}oftware for and applications of
  {V}oronoi diagrams in science and engineering.
\newblock In {\em Proceedings of the 8th International Symposium on Voronoi
  Diagrams in Science and Engineering}, pages 3--12, 2011.

\bibitem{Hershberger1999-ShortestPathMap}
J.~Hershberger and S.~Suri.
\newblock An optimal algorithm for {E}uclidean shortest paths in the plane.
\newblock {\em SIAM Journal on Computing}, 28(6):2215--2256, 1999.

\bibitem{Hillebrand2016-MLE}
A.~Hillebrand, J.M. van~den Akker, R.~Geraerts, and J.A. Hoogeveen.
\newblock Performing multicut on walkable environments.
\newblock In {\em Proceedings of the 10th International Conference on
  Combinatorial Optimization and Applications}, pages 311--325, 2016.

\bibitem{Hillebrand2016-PEEL}
A.~Hillebrand, J.M. van~den Akker, R.~Geraerts, and J.A. Hoogeveen.
\newblock Separating a walkable environment into layers.
\newblock In {\em Proceedings of the 9th ACM SIGGRAPH International Conference
  on Motion in Games}, pages 101--106, 2016.

\bibitem{Hoecker2010-FastestPath}
M.~H\"ocker, V.~Berkhahn, A.~Kneidl, A.~Borrmann, and W.~Klein.
\newblock Graph-based approaches for simulating pedestrian dynamics in building
  models.
\newblock In {\em eWork and eBusiness in Architecture, Engineering and
  Construction}, pages 389--394, 2010.

\bibitem{Hoff1999-GVD}
K.E. {Hoff III}, T.~Culver, J.~Keyser, M.~Lin, and D.~Manocha.
\newblock Fast computation of generalized {V}oronoi diagrams using graphics
  hardware.
\newblock {\em International Conference on Computer Graphics and Interactive
  Techniques}, pages 277--286, 1999.

\bibitem{Holleman2000-MedialAxisMP}
C.~Holleman and L.~E. Kavraki.
\newblock A framework for using the workspace medial axis in prm planners.
\newblock In {\em Proceedings of the IEEE International Conference on Robotics
  and Automation}, volume~2, pages 1408--1413, 2000.

\bibitem{Jaklin2013-MIRAN}
N.S. Jaklin, A.F. {Cook IV}, and R.~Geraerts.
\newblock Real-time path planning in heterogeneous environments.
\newblock {\em Computer Animation and Virtual Worlds}, 24(3):285--295, 2013.

\bibitem{Kallmann2014-LCT}
M.~Kallmann.
\newblock Dynamic and robust {L}ocal {C}learance {T}riangulations.
\newblock {\em ACM Transactions on Graphics}, 33(5), 2014.

\bibitem{Kallmann2016-PathPlanningBook}
M.~Kallmann and M.~Kapadia.
\newblock {\em Geometric and Discrete Path Planning for Interactive Virtual
  Worlds}.
\newblock Morgan \& Claypool Publishers, 2016.

\bibitem{Karamouzas2010-CollisionAvoidance}
I.~Karamouzas and M.H. Overmars.
\newblock A velocity-based approach for simulating human collision avoidance.
\newblock In {\em Proceedings of the 10th International Conference on
  Intelligent Virtual Agents}, pages 180--186, 2010.

\bibitem{Karavelas2004-GVD}
M.I. Karavelas.
\newblock A robust and efficient implementation for the segment {V}oronoi
  diagram.
\newblock In {\em Proceedings of the 1st International Symposium on Voronoi
  Diagrams in Science and Engineering}, pages 51--62, 2004.

\bibitem{Kaul2013-ShortestPathsOnTerrains}
M.~Kaul, R.C.-W. Wong, B.~Yang, and C.S. Jensen.
\newblock Finding shortest paths on terrains by killing two birds with one
  stone.
\newblock {\em Proceedings of the VLDB Endowment}, 7(1):73--84, 2013.

\bibitem{Kavraki1996-PRM}
L.E. Kavraki, P.~{\v{S}}vestka, J.-C. Latombe, and M.H. Overmars.
\newblock Probabilistic roadmaps for path planning in high-dimensional
  configuration spaces.
\newblock {\em IEEE Transactions on Robotics and Automation}, 12(4):566--580,
  1996.

\bibitem{Khramtcova2015-LinearDeletion}
E.~Khramtcova and E.~Papadopoulou.
\newblock Linear-time algorithms for the farthest-segment {V}oronoi diagram and
  related tree structures.
\newblock In {\em Proceedings of the 26th International Symposium on Algorithms
  and Computation}, pages 404--414, 2015.

\bibitem{Kirkpatrick1979-MedialAxis}
D.G. Kirkpatrick.
\newblock Efficient computation of continuous skeletons.
\newblock In {\em Proceedings of the IEEE 54th Annual Symposium on Foundations
  of Computer Science}, pages 18--27, 1979.

\bibitem{Koenig2002-DStarLite}
S.~Koenig and M.~Likhachev.
\newblock D* {L}ite.
\newblock In {\em Proceedings of the 18th National Conference of Artificial
  Intelligence}, pages 476--483, 2002.

\bibitem{Kuffner2000-RRT}
J.J. Kuffner and S.M. LaValle.
\newblock {RRT}-{C}onnect: An efficient approach to single-query path planning.
\newblock In {\em Proceedings of the 17th IEEE International Conference on
  Robotics and Automation}, pages 995--1001, 2000.

\bibitem{Latombe1991-MotionPlanning}
J.C. Latombe.
\newblock {\em Robot Motion Planning}.
\newblock Kluwer Academic Publishers, 1991.

\bibitem{LaValle2006-PlanningAlgorithms}
S.M. LaValle.
\newblock {\em Planning Algorithms}.
\newblock Cambridge University Press, 2006.

\bibitem{Lee1982-MedialAxis}
D.T. Lee.
\newblock Medial axis transformation of a planar shape.
\newblock {\em IEEE Transactions on Pattern Analysis and Machine Intelligence},
  4(4):363--369, 1982.

\bibitem{Lee1981-GVD}
D.T. Lee and R.L. {Drysdale III}.
\newblock Generalization of {V}oronoi diagrams in the plane.
\newblock {\em SIAM Journal on Computing}, 10(1):73--87, 1981.

\bibitem{Lee2013-GridPlanning}
W.~Lee and R.~Lawrence.
\newblock Fast grid-based path finding for video games.
\newblock In {\em Advances in Artificial Intelligence}, volume 7884 of {\em
  Lecture Notes in Computer Science}, pages 100--111. Springer, 2013.

\bibitem{Lien2003-MedialAxisMP}
J.-M. Lien, S.L. Thomas, and N.M. Amato.
\newblock A general framework for sampling on the medial axis of the free
  space.
\newblock In {\em Proceedings of the IEEE International Conference on Robotics
  and Automation}, volume~3, pages 4439--4444, 2003.

\bibitem{Mononen-Recast}
M.~Mononen.
\newblock Recast {N}avigation.
\newblock \url{https://github.com/memononen/recastnavigation/}, 2016.

\bibitem{Moussaid2011-CollisionAvoidance}
M.~Moussa\"id, D.~Helbing, and G.~Theraulaz.
\newblock How simple rules determine pedestrian behavior and crowd disasters.
\newblock {\em Proceedings of the National Academy of Sciences},
  108:6884--6888, 2011.

\bibitem{Narain2009-DenseCrowds}
R.~Narain, A.~Golas, S.~Curtis, and M.C. Lin.
\newblock Aggregate dynamics for dense crowd simulation.
\newblock {\em ACM Transactions on Graphics}, 28:1--8, 2009.

\bibitem{ODunlaing1985-Retraction}
C.~{\'{O}}'D{\'{u}}nlaing and C.K. Yap.
\newblock A `retraction' method for planning the motion of a disc.
\newblock {\em Journal of Algorithms}, 6(1):104--111, 1985.

\bibitem{Okabe2000-Voronoi}
A.~Okabe, B.~Boots, K.~Sugihara, and S.N. Chiu.
\newblock {\em Spatial tessellations: Concepts and applications of {V}oronoi
  diagrams}.
\newblock John Wiley and Sons, Ltd., 2nd edition, 2000.

\bibitem{Oliva2013-Clearance}
R.~Oliva and N.~Pelechano.
\newblock A generalized exact arbitrary clearance technique for navigation
  meshes.
\newblock In {\em Proceedings of the 6th International Conference on Motion in
  Games}, pages 103--110, 2013.

\bibitem{Oliva2013-Neogen}
R.~Oliva and N.~Pelechano.
\newblock {NEOGEN}: {N}ear optimal generator of navigation meshes for 3{D}
  multi-layered environments.
\newblock {\em Computers \& Graphics}, 37(5):403--412, 2013.

\bibitem{Patil2010-NavigationFields}
S.~Patil, J.P. van~den Berg, S.~Curtis, M.C. Lin, and D.~Manocha.
\newblock Directing crowd simulations using navigation fields.
\newblock {\em IEEE Transactions on Visualization and Computer Graphics},
  17:244--254, 2010.

\bibitem{Pettre2005-NavGraph}
J.~Pettr\'{e}, {J.-P.} Laumond, and D.~Thalmann.
\newblock A navigation graph for real-time crowd animation on multilayered and
  uneven terrain.
\newblock In {\em Proceedings of the 1st International Workshop on Crowd
  Simulation}, pages 81--89, 2005.

\bibitem{Pocchiola1993-VisibilityComplex}
M.~Pocchiola and G.~Vegter.
\newblock The visibility complex.
\newblock In {\em Proceedings of the 9th Annual Symposium on Computational
  Geometry}, pages 328--337, 1993.

\bibitem{Polak2016-Filtering}
R.M. Polak.
\newblock Extracting walkable areas from 3{D} environments.
\newblock Master's thesis, Utrecht University, 2016.

\bibitem{Preparata1977-MedialAxis}
F.~Preparata.
\newblock The medial axis of a simple polygon.
\newblock In {\em Mathematical Foundations of Computer Science}, volume~53,
  pages 443--450. Springer, 1977.

\bibitem{Reynolds1999-Steering}
C.~Reynolds.
\newblock Steering behaviors for autonomous characters.
\newblock In {\em Proceedings of the Game Developers Conference}, pages
  763--782, 1999.

\bibitem{Ricks2014-WholeSurface}
B.C. Ricks and P.K. Egbert.
\newblock A whole surface approach to crowd simulation on arbitrary topologies.
\newblock {\em IEEE Transactions on Visualization and Computer Graphics},
  20:159--171, 2014.

\bibitem{Rodriguez2011-Roadmaps}
S.~Rodriguez and N.M. Amato.
\newblock Roadmap-based level clearing of buildings.
\newblock In {\em Proceedings of the 4th International Conference on Motion in
  Games}, pages 340--352, 2011.

\bibitem{Shamos1975-Voronoi}
M.I. Shamos and D.~Hoey.
\newblock Closest-point problems.
\newblock In {\em Proceedings of the 16th Annual IEEE Symposium on Foundations
  of Computer Science}, pages 151--162, 1975.

\bibitem{Snook2000-NavMeshes}
G.~Snook.
\newblock Simplified 3{D} movement and pathfinding using navigation meshes.
\newblock In Mark DeLoura, editor, {\em Game Programming Gems}, pages 288--304.
  Charles River Media, 2000.

\bibitem{Sturtevant2012-Benchmarks}
N.~Sturtevant.
\newblock Benchmarks for grid-based pathfinding.
\newblock {\em Transactions on Computational Intelligence and AI in Games},
  4(2):144--148, 2012.

\bibitem{Thalmann2013-CrowdSimulation}
D.~Thalmann and S.R. Musse.
\newblock {\em Crowd Simulation}.
\newblock Springer, 2nd edition, 2013.

\bibitem{vanToll2015-Replanning}
W.~{\noop{toll}}van Toll and R.~Geraerts.
\newblock {D}ynamically {P}runed {A}* for re-planning in navigation meshes.
\newblock In {\em Proceedings of the 28th IEEE/RSJ International Conference on
  Intelligent Robots and Systems}, pages 2051--2057, 2015.

\bibitem{vanToll2015-Framework}
W.~{\noop{toll}}van Toll, N.~Jaklin, and R.~Geraerts.
\newblock Towards believable crowds: A generic multi-level framework for agent
  navigation.
\newblock In {\em ASCI.OPEN / ICT.OPEN (ASCI track)}, 2015.

\bibitem{vanToll2016-ComparativeStudy}
W.~{\noop{toll}}van Toll, R.~Triesscheijn, M.~Kallmann, R.~Oliva, N.~Pelechano,
  J.~Pettr{\'e}, and R.~Geraerts.
\newblock A comparative study of navigation meshes.
\newblock In {\em Proceedings of the 9th ACM SIGGRAPH International Conference
  on Motion in Games}, pages 91--100, 2016.

\bibitem{vanToll2011-MultiLayered}
W.G. {\noop{toll}}van Toll, A.F. {Cook IV}, and R.~Geraerts.
\newblock Navigation meshes for realistic multi-layered environments.
\newblock In {\em Proceedings of the 24th IEEE/RSJ International Conference on
  Intelligent Robots and Systems}, pages 3526--3532, 2011.

\bibitem{vanToll2012-Dynamic}
W.G. {\noop{toll}}van Toll, A.F. {Cook IV}, and R.~Geraerts.
\newblock A navigation mesh for dynamic environments.
\newblock {\em Computer Animation and Virtual Worlds}, 23(6):535--546, 2012.

\bibitem{Treuille2006-ContinuumCrowds}
A.~{\noop{Treuille}}{Treuille}, S.~Cooper, and Z.~Popovi\'{c}.
\newblock Continuum crowds.
\newblock {\em ACM Transactions on Graphics}, 25:1160--1168, 2006.

\bibitem{Unity}
{\noop{unity}}{Unity3D Game Engine}.
\newblock \url{http://www.unity3d.com/}, 2016.

\bibitem{Wein2007-VisibilityVoronoi}
R.~Wein, J.P. van~den Berg, and D.~Halperin.
\newblock The {V}isibility-{V}oronoi complex and its applications.
\newblock {\em Computational Geometry: Theory and Applications}, 36:66--87,
  2007.

\bibitem{Whiting2007-MultiLayered}
E.~Whiting, J.~Battat, and S.~Teller.
\newblock Topology of urban environments: Graph construction from
  multi-building floor plan data.
\newblock In {\em Proceedings of the 12th International Conference on
  Computer-Aided Architectural Design}, pages 115--128, 2007.

\bibitem{Wilmarth1999-Retraction}
S.A. Wilmarth, N.M. Amato, and P.F. Stiller.
\newblock {MAPRM}: A probabilistic roadmap planner with sampling on the medial
  axis of the free space.
\newblock In {\em Proceedings of the 16th IEEE International Conference on
  Robotics and Automation}, pages 1024--1031, 1999.

\end{thebibliography}
